\documentclass[11pt,british,refpage,intoc,bibliography=totoc,index=totoc,BCOR=7.5mm,captions=tableheading]{extarticle}
\usepackage[T1]{fontenc}
\usepackage[latin9]{inputenc}
\usepackage[a4paper]{geometry}
\usepackage{color}
\usepackage{babel}
\usepackage{caption}
\usepackage{subcaption}
\usepackage{mathrsfs}
\usepackage{amsmath}
\usepackage{graphicx}
\usepackage{amsthm}
\usepackage{amssymb}
\usepackage{float}
\usepackage{enumerate}
\usepackage{yhmath}
\usepackage[numbers]{natbib}
\usepackage[unicode=true,pdfusetitle,
 bookmarks=true,bookmarksnumbered=true,bookmarksopen=false,
 breaklinks=false,pdfborder={0 0 0},pdfborderstyle={},backref=false,colorlinks=true]
 {hyperref}
\hypersetup{
 linkcolor=black, citecolor=blue, urlcolor=black, filecolor=blue, pdfpagelayout=OneColumn, pdfnewwindow=true, pdfstartview=XYZ, plainpages=false}

\makeatletter
\numberwithin{equation}{section}
\numberwithin{figure}{section}
\theoremstyle{plain}
\newtheorem{thm}{\protect\theoremname}[section]
\theoremstyle{plain}
\newtheorem{lem}[thm]{\protect\lemmaname}
\theoremstyle{remark}
\newtheorem{rem}[thm]{\protect\remarkname}
\theoremstyle{definition}

\theoremstyle{plain}
\newtheorem{prop}[thm]{\protect\propositionname}
\theoremstyle{definition}
\newtheorem{defn}[thm]{\protect\definitionname}

\@ifundefined{date}{}{\date{}}
\usepackage{caption}
\usepackage[nottoc]{tocbibind}
\allowdisplaybreaks[4]
\usepackage{dsfont}
\DeclareMathAlphabet{\mathcal}{OMS}{cmsy}{m}{n}

\makeatother

\providecommand{\definitionname}{Definition}
\providecommand{\examplename}{Example}
\providecommand{\lemmaname}{Lemma}
\providecommand{\propositionname}{Proposition}
\providecommand{\remarkname}{Remark}
\providecommand{\theoremname}{Theorem}

\begin{document}



\global\long\def\P{\mathbb{P}}%


\global\long\def\R{\mathbb{R}}

\global\long\def\Pr{\mathrm{Pr}}

\newcommand{\rmd}{\mathrm{d}}

\title{Twin Brownian particle method for \\ the study of Oberbeck-Boussinesq
fluid flows}
\author{Jiawei Li\thanks{School of Mathematics, University of Edinburgh, James Clerk Maxwell Building, Peter Guthrie Tait Rd, Edinburgh, United Kingdom, EH9 3FD. Email:
\protect\href{mailto:jiawei.li@ed.ac.uk}{jiawei.li@ed.ac.uk}}, \ Zhongmin Qian\thanks{Mathematical Institute, University of Oxford, Oxford, United Kingdom, OX2 6GG, and Oxford Suzhou Centre for Advanced Research, Suzhou, China.  Email:
\protect\href{mailto:qianz@maths.ox.ac.uk}{qianz@maths.ox.ac.uk}}, and Mingyu Xu\thanks{Department of Mathematics, Fudan University, Shanghai, China. Email: 
\protect\href{mailto:xumy@fudan.edu.cn}{xumy@fudan.edu.cn}}}
\maketitle
\begin{abstract}
We establish stochastic functional integral representations for solutions of Oberbeck-Boussinesq equations in the form of 
McKean-Vlasov-type mean field equations, which 
can be used to design numerical schemes for calculating solutions and for implementing 
Monte-Carlo simulations of
Oberbeck-Boussinesq flows. Our approach is based on 
the duality of conditional laws for a class of diffusion processes
associated with solenoidal vector fields, which allows us to 
obtain a novel integral representation theorem for solutions
of some linear parabolic equations in terms of the
Green function and the pinned measure of the associated diffusion. We demonstrate via numerical
experiments the 
efficiency of the numerical schemes, which are capable 
of  revealing numerically the details of 
Oberbeck-Boussinesq flows within their thin boundary layer, including  B{\'e}nard's convection feature. 

\medskip

\emph{Key words}: Boussinesq approximation, conditional laws, diffusion processes, incompressible fluid flow,
Monte-Carlo simulation, random vortex method. 

\medskip

\emph{MSC classifications}: 76M35, 76M23, 60H30, 65C05, 68Q10
\end{abstract}

\newpage
\tableofcontents
\newpage
\section{Introduction}

By Oberbeck-Boussinesq flows, we mean fluid flows governed by approximation equations of motion
 for heat-conducting fluid flows
(see e.g. Landau-Lifshitz \citep[Chapters II and V]{LandauLifishitzFluid}),
for details, the reader may refer to Chandrasekhar \citep{Chandrasekhar1961}
and Drazin-Reid \citep[Chapter 2]{Drazin-Reid1981}. These
approximation equations in the form of partial differential equations
were proposed independently by Oberbeck \citep{Oberbeck1879} and Boussinesq \citep{Boussinesq1903}, cf. also Rayleigh \citep{Rayleigh1916}, in which the approximation equations were
derived under the assumption that the fluid density $\rho$ is almost
constant.

The primary goal of the paper is to develop Monte-Carlo-type numerical methods for the study of Oberbeck-Boussinesq flows based on exact stochastic formulations
of the Oberbeck-Boussinesq flows to be established in the paper. This will be achieved by establishing the functional integral representations for
solutions of the Oberbeck-Boussinesq equations. These stochastic integral representations, as well as the approach presented in this work, appear to have independent interests on  their own. 
Indeed we hope these ideas will be useful in the study, 
theoretically and numerically,  of other non-linear systems of partial differential equations.

The Oberbeck-Boussinesq equations are composed of the Navier-Stokes
equations for the velocity $u(x,t)$ coupled with a transport equation
for the temperature $\theta(x,t)$:
\begin{equation}
\frac{\partial}{\partial t}u+(u\cdot\nabla)u=\nu\Delta u-\nabla P+f(\theta),\label{NS-g1}
\end{equation}
\begin{equation}
\nabla\cdot u=0,\label{NS-g2}
\end{equation}
and 
\begin{equation}
\frac{\partial}{\partial t}\theta+(u\cdot\nabla)\theta=\kappa\Delta\theta,\label{NS-g3}
\end{equation}
where $u(x,t)$ is subject to the no-slip condition, i.e. $u(x,t)$
vanishes along a solid boundary, and $\theta(x,t)$ the temperature generally possesses a non-trivial value since the heat is supplied to the fluid system from the solid boundary.
$\nu>0$ is the kinematic viscosity, and $\kappa>0$ is the thermal diffusivity. Hence the Oberbeck-Boussinesq equations are more sophisticated
than the Navier-Stokes equations alone. Nevertheless, they still serve as 
approximation equations for the much more complicated motion equations
governing viscous fluid flows with thermal conduction
when the fluid density remains relatively constant during the heating process.
The Oberbeck-Boussinesq model provides a good explanation for the regular cellular pattern of the fluid motion - when the fluid at the bottom receives heat and expands with increasing temperature to the level at which
the buoyancy dominates over the viscosity effect, a phenomenon called
the B\'enard convection occurs, reported first by B\'enard \citep{Benard1900}. 

The Oberbeck-Boussinesq model has been investigated as one of the very
successful examples in the theory of hydrodynamic stability, see Chandrasekhar
\citep{Chandrasekhar1961}, Joseph \citep{Joseph1965,Joseph1966,Joseph1976}
and Drazin-Reid \citep{Drazin-Reid1981}. In recent years, with
the development of computational power, various numerical methods have
been employed in the study of various fluid flows, including the questions
of hydrodynamic stability and transition to turbulence, see \cite{Criminale-StabilityB2003} for example.

In the study of fluid dynamics over a century, statistical and probabilistic ideas have penetrated gradually into the research area of fluid mechanics, though at the beginning, only very primitive concepts in statistics were borrowed to the study of isotropic and homogeneous turbulent flows, as seen in the seminal work by Taylor \cite{Taylor1935}. In fact, in the statistical theory of turbulence put forward by Taylor \cite{Taylor1935}, Von~K{\'a}rm{\'a}n \cite{Karman1931}, and etc., only the idea of averaging was adopted to describe the mean motions of turbulent flows. Later in K41 theory, Kolmogorov \cite{K41a, K41b} applied more sophisticated concept of conditional laws for random fields and introduced the concept of locally isotropic turbulent flows. 
Moreover, the ideas of random walks and diffusions in fluid flows emerged as means to describe 
turbulent flows. Taylor \cite{Taylor1921} formally introduced Brownian fluid particles into the study of fluid dynamics and made the study of diffusions in turbulence a useful tool in the description of 
various aspects of turbulent flows, see \cite{Pope2000}, \cite{Majda and Bertozzi 2002} and \cite{Falkovich2001} for
a very detailed review. These probabilistic studies of fluid dynamics, i.e. statistical fluid mechanics, occurred before the 
major development in probability theory, such as 
the creation of stochastic calculus by Doob, It{\^o}, etc. Since then, stochastic calculus has been
gradually applied to the study of fluid mechanics too. For example, in LeJan and Sznitman \cite{LeJanSznitman1997}, the energy dissipation cascades in turbulence were interpreted in terms of random walks, while in LeJan and Raimond \cite{LejanRaimond2002, LejanRaimond2004}, Brownian particles were studied.  

Particularly in the study of incompressible fluid flows, 
vorticity has been singled out
as a crucial fluid dynamical variable, in addition to the flow velocity. 
It seems that Helmholtz \citep{Helmholtz1858} was the first person who emphasised
the significance of vortex motions in the study of fluid mechanics. 
Since then, motions of vortices in fluid flows have
been studied throughout the history of fluid dynamics.  
The random vortex method, originated by Chorin \citep{Chorin 1973}, 
is a probabilistic method for incompressible flows developed based on the following simple but fundamental observation. 
The motion of vortices in turbulent flows in nature exhibits (approximately) statistical independence, 
leading to a potentially easier description of the fluid flow via its vortex motion. 
The vorticity equations governing the evolution of the vorticity, which appears as a parabolic transport equation, demonstrate that the vortices are transported along the fluid flow. 
The random vortex method,
though limited to two-dimensional fluid flows, based on the exact fluid dynamic
equations, was first discovered by Goodman \citep{Goodman1987} (see also
 Long \citep{Long1988}). The
relatively novel applications of vortex dynamics in numerical
schemes for solving fluid dynamic equations have been rapidly established
as an important branch of fluid mechanics, see e.g.\citep{AndersonGreengard1988LNM,AndersonGreengard1991LAM},
\citep{CottetKoumoutsakos2000}, \citep{Majda and Bertozzi 2002} and etc.
for excellent reviews on the vortex method. 
%

According to Feynman \cite{Feynman1948} and Kac \cite{Kac1949}, it is possible to express
solutions of certain linear parabolic and elliptic equations in terms of path integrals. 
In fact, the idea of Feynman-Kac has been generalised to a class of semi-linear parabolic equations by Pardoux and Peng in \cite{PardouxPeng1990}, where a nonlinear version of Feynman-Kac formula was obtained via
backward stochastic differential equations. 
In recent years, the random vortex methods have been greatly enhanced under the name of the stochastic Lagrangian 
(vorticity) approach and great progress has been made in Holm \cite{Holm1998},  Busnello \cite{Bunello1999}, 
Busnello, Flandoli and Romito \cite{Busnello2005},  Constantin  \cite{Constantin2001a, Constantin2001b}. In particular
Constantin and Iyer \cite{ConstantinIyer2011} obtained a stochastic integral 
representation for solutions of the Navier-Stokes
equations, cf. Zhang \cite{Zhang2010} as well in which stochastic integral representation has been established
for solutions to backward stochastic Navier-Stokes equations. 
Their representations were established by using the version
of It{\^o}'s formula for stochastic flows established in \cite{LeJanSznitman1997}. 
These stochastic integral representations are applicable to incompressible fluid flows on the whole space 
or flows with periodic boundary conditions. Later on in Constantin and Iyer \cite{ConstantinIyer2011} 
the stochastic Lagrangian approach has been generalised to incompressible fluid
flows constrained in a domain with boundary, and further extended by Iyer \cite{Iyer2006} to inviscid fluid flows 
(with a stochastic perturbation). These stochastic integral representations 
not only
utilise 
Taylor's diffusion driven by the fluid flow velocity but also its backward flows. 

The stochastic Lagrangian approach, more precisely, various stochastic formulations and 
stochastic integral representations (though implicit) for solutions of incompressible fluid
flows are very useful in the study of fluid dynamics, in particular in gaining information
through numerical simulations. In a series of remarkable papers by   Drivas and Eyink \cite{Drivas2017a, Drivas2017},
Eyink, Gupta and Zaki \cite{EyinkGuptaZaki2020a, EyinkGuptaZaki2020b}, the stochastic Lagrangian approach
has been applied successfully to the study of isotropic turbulent flows, verifying the small-scale theory,
such as those proposed in \cite{K41a, K41b}, \cite{Oboukhov1949}, \cite{Corrsin1951} and \cite{Kraichnan1968}, also Wung-Tseng \cite{WungTseng1992}, Yu et al. \cite{YuKanovPerlman2012} and the literature therein.

There are other interesting applications of stochastic calculus in the mathematical study of 
the Navier-Stokes equations, and in the descriptions of fluid dynamics including turbulent flows, let us mention only a few of them: \cite{CelaniCenciniMazzinoVergassola2004}, \cite{Freidlin1985}, \cite{GriebelETC2007}, \cite{Peskin}, \cite{ShlesingerWestKlafter1987}, \cite{ThalbardKrstulovicBec2014}, \cite{Vanden-Eijnden2000, Vanden-Eijnden2001}, and certainly there are more interesting works
the present authors must apologize for their ignorance not due to less significance of the work we are not aware of. 

In this paper, based on the approach developed in recent works \cite{Qian2022Stochastic} and \cite{QSZ3D}, we aim to deal with highly complicated yet important fluid flows with heat transfer by bringing in several new ideas. From the perspective of solving the Navier-Stokes equations numerically
for fluid flows with heat conduction, we devise the following approach, inspired mainly by Taylor \cite{Taylor1921}, 
the fundamental 
ideas in Goodman \citep{Goodman1987} and Long \cite{Long1988}. 
For simplicity, let us illustrate our method for two-dimensional fluid flows. 
To determine the velocity $u(x,t)$ of the flow, it is equivalent to describe the "fictitious" Brownian motion particles with velocity $u(x,t)$, departed from all possible site $\xi$. These Brownian fluid particles are denoted by  $X^{\xi}=(X_{t}^{\xi})_{t\geq0}$, 
which are diffusion processes defined as the weak solutions of 
It{\^o}'s stochastic differential equations
\[
\rmd X_{t}^{\xi}=u(X_{t}^{\xi},t)\rmd t+\sqrt{2\nu}\rmd B_{t},\quad X_{0}^{\xi}=\xi, 
\]
(called Taylor's diffusion (cf. Taylor \citep{Taylor1921})). 
By borrowing the idea from the mean field theory, 
the key step in our approach 
is to turn the preceding stochastic differential equation into a 
McKean-Vlasov type stochastic differential equations (cf. \cite{McKea1966}) by using the law of the Brownian particles $X^{\xi}$ and the governing 
fluid dynamical equations (\ref{NS-g1}, \ref{NS-g2} and \ref{NS-g3}). 
In random vortex methods, this is achieved by using the vorticity transport equation for $\omega(x,t)$. Indeed, by taking
curl operation on both sides of the Navier-Stokes equation (\ref{NS-g1}), 
$\omega(x,t)$
evolves according to a linear (considering the velocity and the temperature $\theta$ as known
fluid dynamical variables) parabolic equation 
\[
\nu\Delta\omega-(u\cdot\nabla)\omega-\frac{\partial}{\partial t}\omega+g=0.
\]
where for simplicity we use $g$  to denote the curl of $f$. While handling the boundary condition imposed on the vorticity $\omega$ for 
wall-bounded flows poses a technical issue (for example see
\cite{Anderson1986BV} and \cite{Chorin1980}), we shall for now focus on unbounded flows on the plane for the sake of clarity.
Nonetheless, we should emphasize the importance of the interaction term $g$ in
the vorticity equation - which always occurs for 
wall-bounded flows, and the underlying reasons will be explained in the subsequent sections. 
Let $p(s,x,t,y)$ denote the transition probability density function of the diffusion process $X^{\xi}$.
Since $\nabla\cdot u=0$, $p(s,x,t,y)$ (for $s<t$ and $x,y\in\mathbb{R}^{2}$)
is also the Green function of the forward parabolic equation
\begin{equation}
    \left(\nu\Delta-u\cdot\nabla-\frac{\partial}{\partial t}\right)v=0.\label{lpe}
\end{equation}
Therefore, according to the vorticity transport equation, the following representation (see \cite{Friedman 1964}) holds:
\begin{equation}
  \omega(x,t)=\int\omega(y,0)p(0,y,t,x)\rmd y +\int^t_0
\int g(y,s) p(s,y,t,x) \rmd y. \label{crep-01} 
\end{equation}
 According to the Biot-Savart law, 
\begin{align}
u(x,t) & =\int K(y,x)\wedge\omega(y,t)\rmd y\nonumber\\
 & =\int \mathbb{E}\left[K(X^{z}_{t},x)\wedge\omega(z,0)\right]\rmd z
  +\int^t_0\iint K(y,x)\wedge g(\xi, s)p(s,\xi,t,y)\rmd y \rmd\xi\rmd s \label{lrep-01}
\end{align}
where $K(y,x)=(2\pi)^{-1}(y-x)/|y-x|^{2}$ is the Biot-Savart singular
integral kernel. The second integral on the right-hand side of (\ref{lrep-01}) may be written in terms of an expectation
\[
\int^t_0\int \mathbb{E}\left[K(X^{\xi,s}_t,x)\wedge g(\xi,s) \right] \rmd \xi\rmd s,
\]
where $X^{\xi,s}$ is the
Taylor's diffusion starting from all possible site $\xi$ and at instance $s$ (for all $s\geq 0$):
\begin{equation}
   \rmd X_{t}^{\xi,s}=u(X_{t}^{\xi,s},t)\rmd t+\sqrt{2\nu}\rmd B_{t},\quad X_{\tau}^{\xi,s}=\xi, \text{ for }\tau\leq s. 
   \label{sT01}
\end{equation}

Therefore 
\begin{equation}
  u(x,t)=\int \mathbb{E}\left[K(X^{z}_{t},x)\wedge\omega(z,0)\right]\rmd z
+\int^t_0\int \mathbb{E}\left[K(X^{\xi,s}_t,x)\wedge g(\xi,s) \right] \rmd \xi\rmd s, \label{lrep-02} 
\end{equation}
which allows us to reformulate the stochastic differential
equation (\ref{sT01}) into a McKean-Vlasov-type mean field 
equation. More precisely,
\begin{align}
   \rmd X_{t}^{\xi,s} = &\left.\int \mathbb{E}\left[K(X_{t}^{z,s},x)\wedge\omega(z,0)\right]\rmd z + \int^t_0\int \mathbb{E}\left[K(X^{\xi,s}_t,x)\wedge g(\xi,s) \right]\right|_{x=X_{t}^{\xi,s}} \rmd \xi\rmd s\nonumber \\
   &+\sqrt{2\nu}\rmd B_{t},\quad X_{\tau}^{\xi,s}=\xi\quad \text{ for all }\tau\leq s.
   \label{sT02}
\end{align}
For the case where the external force $f$ can be computed separately or it is known, then 
the McKean-Vlasov type mean field equation \eqref{sT02}, in combining the strong law of large number, can be used to design Monte-Carlo type numerical schemes for solving the velocity of \eqref{sT02} accordingly. The velocity $u(x,t)$ may be determined by $X^{\xi}$, so the previous
equation, which is the kind of ordinary stochastic differential equations
involving the law of the solution $X^{\xi}$, has an advantage for
numerically calculating the velocity $u(x,t)$, and hence provides a method
for numerically solving the Navier-Stokes equations. There is extensive literature 
on numerical solutions of both ordinary stochastic differential
equations and McKean-Vlasov type mean field equations, see \cite{Kloeden1992} for example. Analogous stochastic 
integral representations may be established for
a 3D flow in a domain with or without boundary in terms of 
the Taylor's diffusion \eqref{sT01} 
alone, as seen in \cite{QSZ3D} and \cite{Qian2022Stochastic}. In these works, new stochastic integral representation theorems were established by using the duality of conditional distributions of a class of diffusion processes and a forward-type Feynman-Kac formula 
for solutions of parabolic equations. 

There is however  a serious disadvantage, and indeed, it is an obstacle to implementing the numerical schemes for computing numerically the solutions of 
the nonlinear mean field equation \eqref{sT02} which 
is established  based on the classical representation
\eqref{crep-01}. In fact, numerical methods for solving the mean field equation \eqref{sT02} require numerically simulating Brownian fluid particles $X^{\xi, s}$ starting not only from any site $\xi$ in the region of fluid, but also for every instance $s\geq 0$. This becomes unavoidable (at least under the current technology) when the interaction force $f$ is not trivial -- 
unfortunately, it is the case for wall-bounded flows and also for fluid flows with internal interaction force applying to the underlying fluid.  The requirement for simulating diffusion paths starting at every instance 
substantially increases the computational cost for computing the solutions of the mean field equation (\ref{sT02}). 
In this paper, to overcome this obstacle for implementing the random vortex approach to the Oberbeck-Boussinesq flows where 
a conducting force is an essential feature, we utilise the divergence-free condition that $\nabla\cdot u=0$ and the duality of conditional laws established in the previous work \cite{QSZ3D}, and derive a new integral representation theorem for solutions of 
the linear parabolic equation \eqref{lpe}.


To the best knowledge of
the present authors, vortex methods have been studied for the
Navier-Stokes equations without the consideration of energy transfer
through heat or other fields. Indeed, substantial novel ideas need to
be introduced in order to extend the random vortex method to important fluid flows appearing in applications, such as fluid
flows with thermal conduction described by (\ref{NS-g1}, \ref{NS-g2},
\ref{NS-g3}). 

In this paper, our goal is to develop new numerical schemes by establishing stochastic functional integral representation theorems  for  solutions to Oberbeck-Boussinesq fluid flows with thermal conduction in unbounded and wall-bounded domains.  

Of course we must point out that the Computational Fluid Dynamics (CFD) is a huge subject, and 
there is a large volume of literature, see for example \cite{Fletcher1991, Wesseling2001} for a small sample.  CFD by default covers all aspects of computational techniques 
for calculating numerically solutions of all interesting fluid flows. In the past, however, CFD mainly concerns with 
the finite difference and finite element methods applying to
fluid flows in science and engineering. Numerical simulations for turbulent flows have become popular too, due to the increasing computational capability over recent years. Various simulation tools have been developed in recent years,
such as Direct Numerical Simulations (DNS), Large Eddy Simulations (LES), Probability Density Function (PDF) method and etc. The reader may refer to \cite{Fletcher1991, Lesieur2005, Pope2000, SenguptaBhaumik2019} for an overview of these aspects.

There is also good literature about
the numerical solutions of Oberbeck-Boussinseq flows, and in general about Monte-Carlo
simulations for fluid flows,
see for example \cite{alanko2016, AldaEtc1997}.


The paper is organised as follows: in Section 2, some preliminary results on Taylor's diffusions associated with solenoidal vector fields are recalled, and a new functional integral representation formula is derived for the solution of the parabolic equation associated with Taylor's diffusion. In section 3, we review the Biot-Savart law and introduce a variation of the classical Biot-Savart law to link the flow velocity field to its vorticity field for the latter sections. To handle the coupled equations governing the fluid flow with thermal conduction, we introduce the twin Brownian particles in Section 4, which serve as Taylor's diffusions in the representation formula discovered in section 2. We establish the results for the unbounded case first. In Section 5, by applying the representation formula with the twin Brownian particles, the probabilistic representations for the fluid dynamical variables in Oberbeck-Boussinesq flows in $\R^2$ and $\R^3$ are established, and this allows us to formulate a closed random vortex dynamical system for Oberbeck-Boussinesq flows on unbounded domains. In section 6, we handle the flows in wall-bounded domains. With the twin Brownian particles in bounded domains, we first derive the representation results for Oberbeck-Boussinesq flows in half-spaces using dynamical variables mollified in a thin layer adjacent to the boundary. Additionally, the random vortex dynamics are characterised using these representations for wall-bounded flows in dimensions two and three. Next, we send the thickness of the layer to zero and find limiting representations of these variables. Finally,  in section 7, we provide numerical schemes based on the representation formulae in Sections 5 and 6, for unbounded and wall-bounded Oberbeck-Boussinesq flows, along with some numerical experiment results with different Prandtl numbers. 

\emph{Notations and Conventions:} Unless otherwise specified, Einstein's summation convention over repeated indices is assumed throughout the paper. For two-dimensional vectors $a=(a_{1},a_{2})$ and $b=(b_{1},b_{2})$, $a\wedge b =a_{1}b_{2}-a_{2}b_{1}$. For a real number $c$,
$a\wedge c$ is identified with the vector $(a_{2}c,-a_{1}c)$.

\section{Divergence-free vector fields}

In this section, we recall several results on divergence-free vector
fields on $\mathbb{R}^{d}$, where the dimension $d\geq2$, although
we are only interested in the case where $d=2$ or $d=3$. 

Suppose that $b(x,t)$ is a time-dependent, bounded and Borel measurable vector field on $\mathbb{R}^{d}$,
such that $\nabla\cdot b=0$ on $\mathbb{R}^{d}$ in the sense of distribution for every $t$. 
Let $\lambda>0$ be a constant. Then 
\[
L_{\lambda,b}=\lambda\Delta+b\cdot\nabla
\]
is a second-order elliptic operator on $\mathbb{R}^{d}$. Furthermore, since $\nabla\cdot b=0$, the adjoint operator $L_{\lambda,b}^{\star}$
of $L_{\lambda,b}$ is $L_{\lambda,-b}$.

For every $\tau\geq0$ and $\xi\in\mathbb{R}^{d}$, there is a unique
weak solution, denoted by $X^{\xi,\tau}=(X_{t}^{\xi,\tau})_{t\geq0}$,
of the stochastic differential equation (SDE):
\[
\rmd X_{t}=b(X_{t},t)\rmd t+\sqrt{2\lambda}\rmd B_{t},\quad X_{\tau}=\xi,\quad\textrm{ for }t\geq\tau,
\]
where $B$ is a $d$-dimensional Brownian motion on some probability
space. If $\tau=0$, then $X^{\xi,0}$ will be denoted by $X^{\xi}$ for simplicity.
$X_{t}^{\xi,\tau}$ may be also denoted by $X(\xi,\tau;t)$ (for $t\geq0$).
The distribution of $X^{\xi,\tau}$ is denoted by $\mathbb{P}^{\xi,\tau}$
and by $\mathbb{P}^{\xi}$ if $\tau=0$, which are probability measures
on the path space $C([0,\infty);\mathbb{R}^{d})$. $X^{\xi,\tau}$
(also $\mathbb{P}^{\xi,\tau}$) is called the diffusion with infinitesimal
generator $L_{\lambda,b}$, or called the $L_{\lambda,b}$-diffusion.

It is known that the law of $X_{t}^{\xi,\tau}$ for every $t>\tau\geq0$
has a positive and continuous probability density function with respect
to the Lebesgue measure, denoted by $p_{\lambda,b}(\tau,\xi,t,x)$.
The function $p_{\lambda,b}(\tau,\xi,t,x)$ for $t>\tau\geq0$ and
$\xi,x\in\mathbb{R}^{d}$ is the transition probability density function
of the $L_{\lambda,b}$-diffusion. Since $b(x,t)$ is a Borel measurable,
bounded vector field, $p_{\lambda,b}(\tau,\xi,t,x)$ is jointly H\"older
continuous in $t>\tau$ and $\xi,x\in\mathbb{R}^{d}$. Moreover, if
$b(x,t)$ is smooth, so is $p_{\lambda,b}$. 

The conditional law $\mathbb{P}^{\xi,\tau}[\ \cdot \ |X_{T}^{\xi,\tau}=\eta]$
(also called the pinned measure) is denoted by $\mathbb{P}^{\xi,\tau\rightarrow\eta,T}$,
and by $\mathbb{P}^{\xi\rightarrow\eta}$ if $\tau=0$ and $T>0$
is given. For the construction of the conditional laws, cf. \citep{QSZ3D}
and \citep{Qian2022Stochastic}.

It is known that $p_{\lambda,b}(\tau,\xi,t,x)$ is the Green function
of the backward parabolic operator $L_{\lambda,b}+\frac{\partial}{\partial t}$
(cf. \citep{Stroock-VaradhanMDP}), so that $p_{\lambda,b}(\tau,\xi,t,x)$
coincides with the Green function of the forward parabolic operator
$L_{\lambda,b}^{\star}-\frac{\partial}{\partial t}$ (cf. \citep[Chapter 1, Theorem 15]{Friedman 1964}).
Since $b(x,t)$ is divergence-free on $\mathbb{R}^{d}$ in the distribution
sense, $L_{\lambda,b}^{\star}=L_{\lambda,-b}$, therefore $p_{\lambda,b}(\tau,\xi,t,x)$
is the Green function of the forward parabolic operator $L_{\lambda,-b}-\frac{\partial}{\partial t}$
on $\mathbb{R}^{d}$.

Suppose that $D\subset\mathbb{R}^{d}$ is a domain with a Lipschitz
continuous boundary $\partial D$. Let $p_{\lambda,b}^{D}(\tau,\xi,t,x)$
(for $t>\tau\geq0$, $\xi,x\in D$) be the transition density function
of the diffusion $X$ killed on leaving the domain $D$, that is
\[
\mathbb{E}\left[\varphi(X_{t}^{\xi,\tau})1_{\left\{ t<\zeta(X^{\xi,\tau})\right\} }\right]=\int_{D}p_{\lambda,b}^{D}(\tau,\xi,t,x)\varphi(x)\rmd x
\]
for any bounded and Borel measurable function $\varphi$, where $\zeta(\psi) = \inf\{t\geq 0: \psi(t)\notin D\}$. Since $\nabla\cdot b=0$
in the distribution sense in $\mathbb{R}^{d}$, $p_{\lambda,b}^{D}(\tau,\xi,t,x)$
is the Green function to the Dirichlet problem of the forward parabolic
equation operator 
\begin{equation}
\left(L_{\lambda,-b}-\frac{\partial}{\partial t}\right)w(x,t)=0\quad\textrm{ in }D\times(0,\infty)
\label{p-01}
\end{equation}
subject to the Dirichlet boundary condition that
\begin{equation}
w(x,t)=0\quad\textrm{ for }x\in\partial D.\label{p-02}
\end{equation}


Next, we establish the main technical tool for the present work. 

Let  $\varPsi(x,t)=(\varPsi^{1}(x,t),\cdots,\varPsi^{n}(x,t))$ be
a smooth solution of the parabolic equations
\begin{equation}
\left(\lambda\Delta-b\cdot\nabla-\frac{\partial}{\partial t}\right)\varPsi^{i}+q_{l}^{i}\varPsi^{l}+F^{i}=0\quad\textrm{ in }D\times(0,\infty),\label{Phi-01}
\end{equation}
\begin{equation}
\varPsi(x,t)=0\quad\textrm{ for }x\in\partial D,\label{B-Phi-01}
\end{equation}
for $i=1,\cdots,n$, where $q(x,t)=(q_{l}^{i}(x,t))_{1\leq i,l\leq n}$
is a bounded, Borel measurable $n\times n$ matrix-valued function
on $D$. $q_{l}^{i}(x,t)=0$ for $x\notin D$, or otherwise we replace
$q(x,t)$ by $1_{D}(x)q(x,t)$ instead. $F(x,t)=(F^{i}(x,t))_{1\leq i\leq n}$
is a family of functions which are $C^{2,1}(D\times[0,\infty))$.

Let $\Omega=C([0,\infty);\mathbb{R}^{d})$, and let $\tau_{T}$ denote
the time reversal at $T$ on $\Omega$. That is, $\psi\circ\tau_{T}(s)=\psi(T-s)$
for $s\in[0,T]$. Also, as defined before, $\zeta(\psi)=\inf\left\{ t:\psi(t)\notin D\right\}$, the first exit time from the region $D$, and $\gamma_{T}(\psi)=\sup\left\{ t\in(0,T):\psi(t)\notin D\right\}$, the last exit time before time $T$ from $D$.

\begin{thm}
\label{thm7.2-1} Let $T>0$ and $\eta\in\mathbb{R}^{d}$. Let $t\mapsto Q_{j}^{i}(\eta,T;t)$
(for $t\in[0,T]$) be the solutions to the ordinary differential equations
\begin{equation}
\frac{\rmd}{\rmd t}Q_{j}^{i}(\eta,T;t)=-Q_{k}^{i}(\eta,T;t)q_{j}^{k}\left(X_{t}^{\eta},t\right),\quad Q_{j}^{i}(\eta,T;T)=\delta_{ij}\label{Q-s01}
\end{equation}
for $i,j=1,\cdots,n$. Then
\begin{align}
\varPsi^{i}(\xi,T)  =&\int_{D}\mathbb{E}\left[\left.1_{\{T<\zeta(X^{\eta})\}}Q_{j}^{i}(\eta,T;0)\varPsi^{j}(\eta,0)\right|X_{T}^{\eta}=\xi\right]p_{\lambda,b}(0,\eta,T,\xi)\rmd \eta\nonumber \\
 & +\int_{0}^{T}\int_{D}\mathbb{E}\left[\left.1_{\left\{ t>\gamma_{T}(X^{\eta})\right\} }Q_{j}^{i}(\eta,T;t)F^{j}(X_{t}^{\eta},t)\right|X_{T}^{\eta}=\xi\right]p_{\lambda,b}(0,\eta,T,\xi)\rmd \eta\rmd t\label{rep-m02-1}
\end{align}
for every $\xi\in D$ and $T>0$, $i=1,\cdots,n$. 
\end{thm}

\begin{proof} The subscript $\lambda$ will be omitted in the proof. For each $T>0$, let $\tilde{X}^{\xi}$
be the solution to the stochastic differential equation
\[
\rmd \tilde{X}_{t}^{\xi}=-b(\tilde{X}_{t}^{\xi},T-t)\rmd t+\sqrt{2\lambda}\rmd B_{t},\quad\tilde{X}_{0}^{\xi}=\xi
\]
for every $\xi\in\mathbb{R}^{d}$ and $B$ is a standard Brownian
motion in $\mathbb{R}^{d}$ on some probability space, where $b(x,t)=0$
for $t<0$. Define
\[
\rmd\tilde{Q}_{j}^{i}(t)=\tilde{Q}_{k}^{i}(t)q_{j}^{k}(\tilde{X}_{t}^{\xi},T-t)\rmd t,\quad\tilde{Q}_{j}^{i}(0)=\delta_{ij}
\]
(note that we assume that $q_{j}^{i}(x,t)=0$ for $x\notin D$), where
$i,j=1,\cdots,n$. Let 
\[
Y_{t}=\varPsi(\tilde{X}_{t\wedge T_{\xi}}^{\xi},T-t)=1_{\{t<T_{\xi}\}}\varPsi(\tilde{X}_{t}^{\xi},T-t),
\]
where $T_{\xi}=\inf\left\{ t:\tilde{X}_{t}^{\xi}\notin D\right\} $.
Here the second equality is ensured as $\varPsi$ vanishes along the
boundary $\partial D$. Let $M^{i}=\tilde{Q}_{j}^{i}Y^{j}$. Then
by It\^o's formula and Eq. (\ref{Phi-01}, \ref{B-Phi-01}) we obtain
that
\begin{align*}
M_{t}^{i} =&Y_{0}^{i}+\sqrt{2\lambda}\int_{0}^{t}1_{\{s<T_{\xi}\}}\tilde{Q}_{j}^{i}(s)\nabla\varPsi^{j}(\tilde{X}_{s}^{\xi},T-s)\cdot\rmd B_{s}\\
 & -\int_{0}^{t}1_{\{s<T_{\xi}\}}\tilde{Q}_{j}^{i}(s)F^{j}(\tilde{X}_{s}^{\xi},T-s)\rmd s.
\end{align*}
Taking expectation on both sides and using the fact that $\varPsi$ vanishes
along the boundary, we obtain that 
\begin{align}
\varPsi^{i}(\xi,T) =& \mathbb{E}\left[\tilde{Q}_{j}^{i}(T)\varPsi^{j}(\tilde{X}_{T}^{\xi},0)1_{\{T<T_{\xi}\}}\right]\nonumber \\
 &+\int_{0}^{T}\mathbb{E}\left[1_{\{t<T_{\xi}\}}\tilde{Q}_{j}^{i}(t)F^{j}(\tilde{X}_{t}^{\xi},T-t)\right]\rmd t\nonumber \\
  \equiv& J_{1}^{i}+J_{2}^{i}.\label{Wei-11}
\end{align}
Note that $T_{\xi}$ is a stopping time with respect to the filtration
generated by $\tilde{X}$, so that $\tilde{Q}_{j}^{i}(t)1_{\{t<T_{\xi}\}}$
is therefore measurable with respect to $\tilde{X}$ running up to
time $t$, which allows us to take conditional expectation on giving
$\tilde{X}_{T}=\eta$, to obtain that
\begin{align*}
J_{2}^{i} & =\int_{0}^{T}\int_{\mathbb{R}^{d}}\mathbb{E}\left[\left.\tilde{Q}_{j}^{i}(t)1_{\{t<T_{\xi}\}}F^{j}(\tilde{X}_{t}^{\xi},T-t)\right|\tilde{X}_{T}^{\xi}=\eta\right]\mathbb{P}\left[\tilde{X}_{T}^{\xi}\in\rmd \eta\right]\rmd t\\
 & =\int_{0}^{T}\int_{D}\mathbb{E}\left[\left.\tilde{Q}_{j}^{i}(t)1_{\{t<T_{\xi}\}}F^{j}(\tilde{X}_{t}^{\xi},T-t)\right|\tilde{X}_{T}^{\xi}=\eta\right]p_{-b_{T}}(0,\xi,T,\eta)\rmd \eta\rmd t
\end{align*}
where $p_{-b_{T}}(\tau,\xi,t,\eta)$ is the transition probability
density function of the diffusion $\tilde{X}^{\xi}$, and the second equality follows from the fact that the conditional
expectation is zero if $\eta\notin D$, so we may restrict the integral
for $\eta\in D$ only. Similarly
\[
J_{1}^{i}=\int_{D}\mathbb{E}\left[\left.\tilde{Q}_{j}^{i}(T)1_{\{T<T_{\xi}\}}\right|\tilde{X}_{T}^{\xi}=\eta\right]\varPsi^{j}(\eta,0)p_{-b_{T}}(0,\xi,T,\eta)\rmd \eta.
\]

Since $\nabla\cdot b=0$, so that we can replace $p_{-b_{T}}(0,\xi,T,\eta)$
by $p_{b}(0,\eta,T,\xi)$. Hence
\begin{equation}
J_{1}^{i}=\int_{D}\mathbb{E}\left[\left.\tilde{Q}_{j}^{i}(T)1_{\{T<T_{\xi}\}}\right|\tilde{X}_{T}^{\xi}=\eta\right]\varPsi^{j}(\eta,0)p_{b}(0,\eta,T,\xi)\rmd \eta\label{J1-i01}
\end{equation}
and
\begin{equation}
J_{2}^{i}=\int_{0}^{T}\int_{D}\mathbb{E}\left[\left.\tilde{Q}_{j}^{i}(t)1_{\{t<T_{\xi}\}}F^{j}(\tilde{X}_{t}^{\xi},T-t)\right|\tilde{X}_{T}^{\xi}=\eta\right]p_{b}(0,\eta,T,\xi)\rmd\eta\rmd t\label{J2-i03}
\end{equation}
for $i=1,\cdots,n$. Next we utilise the approach in \citep{QSZ3D}
and rewrite the conditional expectations in terms of the diffusion
process with infinitesimal generator $\lambda\Delta+b\cdot\nabla$.
To this end, we first introduce a few notations. Let $\tilde{\mathbb{P}}^{\xi}$
denote the law of $\tilde{X}^{\xi}$ and $\tilde{\mathbb{P}}^{\xi\rightarrow\eta}$
denote the conditional law of $\tilde{X}^{\xi}$ given the terminal
value that $\tilde{X}_{T}^{\xi}=\eta$. Let $\tilde{Q}(\psi;t)$ denote
the solution to the linear system of ordinary differential equations
\[
\frac{\rmd}{\rmd t}\tilde{Q}_{j}^{i}(\psi;t)=\tilde{Q}_{k}^{i}(\psi;t)q_{j}^{k}(\psi(t),T-t),\quad\tilde{Q}_{j}^{i}(\psi;0)=\delta_{ij}
\]
for every $\psi\in C([0,T];\R^{d})$, where $i,j=1,\cdots,n$.
The representations (\ref{J1-i01}, \ref{J2-i03}) can be rewritten
as:
\[
J_{1}^{i}=\int_{D}\tilde{\mathbb{P}}^{\xi\rightarrow\eta}\left[\tilde{Q}_{j}^{i}(\psi;T)1_{\{T<\zeta(\psi)\}}\right]\varPsi^{j}(\eta,0)p_{b}(0,\eta,T,\xi)\rmd \eta
\]
and
\[
J_{2}^{i}=\int_{0}^{T}\int_{D}\tilde{\mathbb{P}}^{\xi\rightarrow\eta}\left[\tilde{Q}_{j}^{i}(\psi;t)F^{j}(\psi(t),T-t)1_{\{t<\zeta(\psi)\}}\right]p_{b}(0,\eta,T,\xi)\rmd\eta\rmd t.
\]
Since $\nabla\cdot b=0$, according to the duality for the conditional
laws (cf. \citep[Theorem 3.1]{QSZ3D}), the conditional law $\tilde{\mathbb{P}}^{\xi\rightarrow\eta}$
coincides with the conditional law $\mathbb{P}^{\eta\rightarrow\xi}\circ\tau_{T}$, where $\tau_T$ denotes the time reversal at $T$.
Therefore we may rewrite
\[
J_{1}^{i}=\int_{D}\P^{\eta\rightarrow\xi}\left[1_{\{T<\zeta(\psi\circ\tau_{T})\}}\tilde{Q}_{j}^{i}(\psi\circ\tau_{T};T)\right]\varPsi^{j}(\eta,0)p_{b}(0,\eta,T,\xi)\rmd \eta
\]
and
\[
J_{2}^{i}=\int_{0}^{T}\int_{D}\mathbb{P}^{\eta\rightarrow\xi}\left[1_{\{t<\zeta(\psi\circ\tau_{T})\}}\tilde{Q}_{j}^{i}(\psi\circ\tau_{T};t)F^{j}(\psi(T-t),T-t)\right]p_{b}(0,\eta,T,\xi)\rmd\eta\rmd t.
\]
It remains to identify the flow $\tilde{Q}_{j}^{i}(\psi\circ\tau_{T};t)$
for $t\in[0,T]$. Consider $Q_{j}^{i}(\psi,T;s)=\tilde{Q}_{j}^{i}(\psi\circ\tau_{T};T-s)$
for $0<s\leq T$. Then one can easily verify that $Q_{j}^{i}(\psi,T;s)$
is exactly the unique solution of Eq. (\ref{Q-s01}). Therefore
\[
J_{1}^{i}=\int_{D}\mathbb{P}^{\eta\rightarrow\xi}\left[1_{\{T<\zeta(\psi\circ\tau_{T})\}}Q_{j}^{i}(\psi,T;0)\right]\varPsi^{j}(\eta,0)p_{b}(0,\eta,T,\xi)\rmd \eta
\]
and
\begin{align*}
J_{2}^{i} & =\int_{0}^{T}\int_{D}\mathbb{P}^{\eta\rightarrow\xi}\left[1_{\{t<\zeta(\psi\circ\tau_{T})\}}Q_{j}^{i}(\psi,T;T-t)F^{j}(\psi(T-t),T-t)\right]p_{b}(0,\eta,T,\xi)\rmd \eta\rmd t\\
 & =\int_{0}^{T}\int_{D}\mathbb{P}^{\eta\rightarrow\xi}\left[1_{\{T-t<\zeta(\psi\circ\tau_{T})\}}Q_{j}^{i}(\psi,T;t)F^{j}(\psi(t),t)\right]p_{b}(0,\eta,T,\xi)\rmd \eta\rmd t.
\end{align*}
Hence, by using equalities in (\ref{Wei-11}), we obtain
\begin{align*}
\varPsi^{i}(\xi,T)  =&\int_{D}\mathbb{P}^{\eta\rightarrow\xi}\left[1_{\{T<\zeta(\psi\circ\tau_{T})\}}Q_{j}^{i}(\psi,T;0)\right]p_{b}(0,\eta,T,\xi)\varPsi^{j}(\eta,0)\rmd\eta\\
 & +\int_{0}^{T}\int_{D}\mathbb{P}^{\eta\rightarrow\xi}\left[1_{\{T-t<\zeta(\psi\circ\tau_{T})\}}Q_{j}^{i}(\psi,T;t)F^{j}(\psi(t),t)\right]p_{b}(0,\eta,T,\xi)\rmd \eta\rmd t.
\end{align*}
 Now we make the following observation. For $\xi,\eta\in D$, then
for a path $\psi\in C([0,T];\mathbb{R}^{d})$ with $\psi(0)=\eta$
and $\psi(T)=\xi$, then $T<\zeta(\psi\circ\tau_{T})$ is equivalent
to say $\psi(T-t)\in D$ for all $t\in(0,T)$, and therefore the condition
that $T<\zeta(\psi\circ\tau_{T})$ almost surely under the conditional
law $\mathbb{P}^{\eta\rightarrow\xi}$ is equivalent to that $T<\zeta(\psi)$
almost surely w.r.t. $\mathbb{P}^{\eta\rightarrow\xi}$. Therefore
the indicator function $1_{\{T<\zeta(\psi\circ\tau_{T})\}}$ can
be replaced by $1_{\{T<\zeta(\psi)\}}$. The treatment for the random function $1_{\{T-t<\zeta(\psi\circ\tau_{T})\}}$
(for $t\in(0,T)$) in the second term is more subtle. Under the conditional law $\mathbb{P}^{\eta\rightarrow\xi}$
for $\xi,\eta\in D$, we may replace 
\begin{align*}
\zeta(\psi\circ\tau_{T}) & =\inf\left\{ s>0:\psi(T-s)\in\partial D\right\} \wedge T\\
 & =\inf\left\{ T-s>0:\psi(s)\in\partial D\right\} \wedge T\\
 & =T-\sup\left\{ s\in(0,T):\psi(s)\in\partial D\right\} 
\end{align*}
with the convention that $\sup\emptyset=0$. Thus $T-t<\zeta(\psi\circ\tau_{T})$
is equivalent to 
\[t>\sup\left\{ s\in(0,T):\psi(s)\in\partial D\right\} .
\]
Therefore it is useful to introduce the notation that
\[
\gamma_{T}(\psi)=\sup\left\{ s\in(0,T):\psi(s)\in\partial D\right\} 
\]
for every path $\psi$ and $T>0$. Then $1_{\{T-t<\zeta(\psi\circ\tau_{T})\}}$
can be replaced by $1_{\left\{ t>\gamma_{T}(\psi)\right\} }$. Hence
(\ref{rep-m02-1}) follows immediately.
\end{proof}

\begin{rem}
In handling the term $J_{2}^{i}$ ($i=1,\cdots,n$), one may use the
conditional law given $\tilde{X}_{t}=\eta$ in place of $\tilde{X}_{T}=\eta$,
so that
\begin{align*}
J_{2}^{i} & =\int_{0}^{T}\int_{\mathbb{R}^{d}}\mathbb{E}\left[\left.\tilde{Q}_{j}^{i}(t)1_{\{t<T_{\xi}\}}F^{j}(\tilde{X}_{t}^{\xi},T-t)\right|\tilde{X}_{t}^{\xi}=\eta\right]\mathbb{P}\left[\tilde{X}_{t}^{\xi}\in\rmd \eta\right]\rmd t\\
 & =\int_{0}^{T}\int_{D}\mathbb{E}\left[\left.\tilde{Q}_{j}^{i}(t)1_{\{t<T_{\xi}\}}\right|\tilde{X}_{t}^{\xi}=\eta\right]F^{j}(\eta,T-t)p_{-b_{T}}(0,\xi,t,\eta)\rmd \eta\rmd t\\
 & =\int_{0}^{T}\int_{D}\mathbb{E}\left[\left.\tilde{Q}_{j}^{i}(t)1_{\{t<T_{\xi}\}}\right|\tilde{X}_{t}^{\xi}=\eta\right]F^{j}(\eta,T-t)p_{b}(T-t,\eta,T,\xi)\rmd \eta\rmd t.
\end{align*}
This leads to the similar formula in the case where $D=\mathbb{R}^{d}$
and $q_{j}^{i}\equiv0$. Indeed, for this case $T_{\xi}=\infty$ and
$\tilde{Q}_{j}^{i}=\delta_{ij}$ so that 
\[
\mathbb{E}\left[\left.\tilde{Q}_{j}^{i}(t)1_{\{t<T_{\xi}\}}\right|\tilde{X}_{t}^{\xi}=\eta\right]=\delta_{ij}
\]
and therefore 
\begin{align*}
J_{2}^{i} & =\int_{0}^{T}\int_{\mathbb{R}^{d}}F^{i}(\eta,T-t)p_{b}(T-t,\eta,T,\xi)\rmd\eta\rmd t\\
 & =\int_{0}^{T}\int_{\mathbb{R}^{d}}F^{i}(\eta,t)p_{b}(t,\eta,T,\xi)\rmd \eta\rmd t.
\end{align*}
\end{rem}

\begin{rem}
    The duality of conditional laws among 
    certain diffusion processes used in the proof of the preceding theorem is the main tool developed
    in \cite{QSZ3D} in order to formulate a random vortex
    method for three dimensional incompressible fluid flows by using only the forward Taylor diffusion of the flow velocity.  The conditional
    law techniques for diffusions have their origin from the study of
    symmetric diffusion processes and Dirichlet forms. The
    notion of duality of diffusion distributions was certainly developed from the well-known concept of self-adjoint operators, and for diffusion semigroups, the path-space
    version of the self-adjointness was first discovered in
    a seminal work by Lyons and Zheng \cite{LyonsZheng1988}. In fact, Lyons-Zheng \cite{LyonsZheng1990} has utilised the conditional
    laws of symmetric diffusions to the study of heat kernels for a class of non-symmetric diffusion processes. The reader may find a detailed account in \cite{Qian2022Stochastic} about
    conditional law duality and its applications 
    to forward type Feynman-Kac formulas. 
\end{rem}

While in the presence of a non-trivial gauge function $q(x,t)$, or
the presence of non-empty boundary, it seems that the formulation
is more useful by conditioning on $\tilde{X}_{T}$. Due to these
are important cases, we may formulate the following theorem, which
is indeed an extension of the theorem proved in \citep{QSZ3D}.

\begin{thm}
\label{thm7.2-1-1} Under the same assumptions and notations as in
Theorem \ref{thm7.2-1}, suppose that $\varPsi=(\varPsi^{i})_{1\leq i\leq n}$ is a solution
to Eq. (\ref{Phi-01}) on $\mathbb{R}^{d}$, then
\begin{align*}
\varPsi^{i}(\xi,T)  =&\int_{\mathbb{R}^{d}}\mathbb{E}\left[\left.Q_{j}^{i}(\eta,T;0)\varPsi^{j}(\eta,0)\right|X_{T}^{\eta}=\xi\right]p_{b}(0,\eta,T,\xi)\rmd \eta\nonumber \\
 & +\int_{0}^{T}\int_{\mathbb{R}^{d}}\mathbb{E}\left[\left.Q_{j}^{i}(\eta,T;t)F^{j}(X_{t}^{\eta},t)\right|X_{T}^{\eta}=\xi\right]p_{b}(0,\eta,T,\xi)\rmd \eta\rmd t
 \end{align*}
for every $\xi\in\mathbb{R}^{d}$, $T>0$ and $i=1,\cdots,n$. 
\end{thm}

\section{The Biot-Savart laws}

In this section we formulate the Biot-Savart laws we need in this
work for velocity, vorticity and temperature and its gradient.

Recall that the Green function in $\mathbb{R}^{d}$ is given by the following
formula
\[
\Gamma_{d}(y,x)=
\begin{cases}
-\frac{1}{(d-2)s_{d-1}}\frac{1}{|y-x|^{d-2}}, & \textrm{ if }d>2,\\
\frac{1}{2\pi}\ln|y-x|,\quad & \textrm{ if }d=2,
\end{cases}
\]
where $s_{d-1}$ is the surface area of a unit sphere in $\mathbb{R}^{d}$,
so $s_{2}=4\pi$. We are only interested in the cases where $d=2$
and $d=3$. The Biot-Savart singular integral kernels $K_{d}(y,x)=\nabla_{y}\Gamma_{d}(y,x)$,
so that
\[
K_{2}(y,x)=\frac{1}{2\pi}\frac{y-x}{|y-x|^{2}}\quad\textrm{ for }y\neq x
\]
and
\[
K_{3}(y,x)=\frac{1}{4\pi}\frac{y-x}{|y-x|^{3}}\quad\textrm{ for }y\neq x.
\]
 
Recall that under our convention for two-dimensional vectors, $\omega=\nabla\wedge u$ is identified
with the scalar function $\frac{\partial}{\partial x_{1}}u^{2}-\frac{\partial}{\partial x_{2}}u^{1}$,
and for $a=(a_1,a_2)$, $a\wedge\omega$ is identified with $(a_{2}\omega,-a_{1}\omega)$. 

\begin{lem}[The Biot-Savart law]
\label{lem4.9}
Let $d=2$ or $d=3$. 
\begin{enumerate}[{(1)}]
    \item If $u$ is a vector field on $\mathbb{R}^{d}$ such that $\nabla\cdot u=0$ and $\omega=\nabla\wedge u$, then 
    \begin{equation}
    u(x)=\int_{\mathbb{R}^{d}}K_{d}(y,x)\wedge\omega(y)\rmd y,\quad\forall x\in\mathbb{R}^{d}.\label{t0t-BS01}
    \end{equation}
    \item If $\theta$ is a function on $\mathbb{R}^{d}$ and $\varTheta=\nabla\cdot\theta$,
    then 
    \begin{equation}
    \theta(x)=-\int_{\mathbb{R}^{d}}K_{d}(y,x)\cdot\varTheta(y)\rmd y,\quad\forall x\in\mathbb{R}^{d}.\label{t0t-BS01-1}
\end{equation}
\end{enumerate}

\end{lem}

These formulae follow immediately from the Green function and integration
by parts.

The Green function for the upper half space $\mathbb{R}_{+}^{d} = \{x=(x_1,x_2,\cdots, x_d)\in \mathbb{R}^d: x_d>0\}$
is given by
\[
G_{d}(y,x)=\Gamma_{d}(y,x)-\Gamma_{d}(y,\overline{x})
\]
for every $x,y\in\mathbb{R}_{+}^{d}$, and $\overline{x}=(x_{1},\cdots,x_{d-1},-x_{d})$
is the reflection of $x$ about the hyperplane  $\{x_{d}=0\}$. Then the
Green formula for $\mathbb{R}_{+}^{d}$ implies that
\begin{equation}
\varphi(x)=\int_{\mathbb{R}_{+}^{d}}G_{d}(y,x)\Delta\varphi(y)\rmd y-\int_{\{y_{d}=0\}}\varphi(y)\frac{\partial}{\partial y_{d}}G_{d}(y,x)\rmd y_{1}\cdots\rmd y_{d-1}\label{Gr-03}
\end{equation}
for $x\in\mathbb{R}_{+}^{d}$ , where $\varphi$ is $C^{2}$, continuous
up to the boundary where $x_{d}=0$, and vanishes at the infinity.

Similarly, we define $\varLambda_{d}(y,x)=\nabla_{y}G_{d}(y,x)$, which
may be called the Biot-Savart singular kernel on $\mathbb{R}_{+}^{d}$.
Then 
\begin{equation}
\varLambda_{2}(y,x)=\frac{1}{2\pi}\left(\frac{y-x}{|y-x|^{2}}-\frac{y-\overline{x}}{|y-\overline{x}|^{2}}\right)\quad\textrm{ for }y\neq x\textrm{ or }\overline{x}\label{2D-H-01}
\end{equation}
and 
\[
\varLambda_{3}(y,x)=\frac{1}{4\pi}\left(\frac{y-x}{|y-x|^{3}}-\frac{y-\overline{x}}{|y-\overline{x}|^{3}}\right)\quad\textrm{ for }y\neq x\textrm{ or }\overline{x}.
\]
The Biot-Savart laws we need in this paper follow from the Green formula
(\ref{Gr-03}), and are stated as two lemmas below. 
\begin{lem}
\label{lem211} Let $d=2$ or $d=3$. If $u$ is a vector field on
$\mathbb{R}_{+}^{d}$ (continuous up to the boundary, decays to zero
at infinity) such that $\nabla\cdot u=0$ in $\mathbb{R}_{+}^{d}$
and $u(x)=0$ when $x_{d}=0$. Let $\omega=\nabla\wedge u$. Then
\begin{equation}
u(x)=\int_{\mathbb{R}_{+}^{d}}\varLambda_{d}(y,x)\wedge\omega(y)\rmd y,\quad\textrm{ for }x\in\mathbb{R}_{+}^{d}.\label{3D-BS-01}
\end{equation}
\end{lem}

\begin{proof}
Since $\nabla\cdot u=0$ and $\nabla\wedge u=\omega$, 
\[
\Delta u=-\nabla\wedge\nabla\wedge u=-\nabla\wedge\omega.
\]
Now as $u(x)=0$ when $x_{d}=0$, by Green's formula we obtain that
\[
u(x)=\int_{\mathbb{R}_{+}^{d}}G_{d}(y,x)\Delta u(y)\rmd y=-\int_{\mathbb{R}_{+}^{d}}G_{d}(y,x)\nabla\wedge\omega(y)\rmd y
\]
and the claim follows immediately after applying integration by parts.
\end{proof}
Next, we formulate a similar law of Biot-Savart's for the temperature. 
\begin{lem}
\label{lem212}Let $\theta$ be a scalar function on $\mathbb{R}_{+}^{d}$
and $\varTheta=\nabla\theta$ be its gradient with its components
$\varTheta_{i}=\frac{\partial}{\partial x_{i}}\theta$ for $i=1,\cdots,d$.
Suppose 
\[\theta(x)=\theta_{0}(x_{1},\cdots,x_{d-1}), \quad\forall x=(x_{1},\cdots,x_{d-1},0)
\]
is the trace of $\theta$ along the boundary $\{x_{d}=0\}$.
\begin{enumerate}[{(1)}]
    \item If $d=3$, then 
    \[
    \theta(x)=-\int_{\mathbb{R}_{+}^{3}}\varLambda_{3}(y,x)\cdot\varTheta(y)\rmd y+\frac{x_{3}}{2\pi}\int_{\mathbb{R}^{2}}\frac{\theta_{0}(y_{1},y_{2})}{|y_{1}-x_{2}|^{2}+|y_{2}-x_{2}|^{2}+|x_{3}|^{2}}\rmd y_{1}\rmd y_{2}\]
    for $x=(x_{1},x_{2},x_{3})$ with $x_{3}>0$.
    \item If $d=2$, then
    \[
    \theta(x)=-\int_{\mathbb{R}_{+}^{2}}\varLambda_{2}(y,x)\cdot\varTheta(y)\rmd y+\frac{x_{2}}{\pi}\int_{-\infty}^{\infty}\frac{\theta_{0}(y_{1})}{|y_{1}-x_{1}|^{2}+x_{2}^{2}}\rmd y_{1}
    \]
    for $x=(x_{1},x_{2})$ with $x_{2}>0$.
\end{enumerate}
\end{lem}

\begin{proof}
By definition, $\Delta\theta=\nabla\cdot\varTheta$, so according
to the Green formula,
\[
\theta(x)=-\int_{\mathbb{R}_{+}^{d}}\nabla_{y}G_{d}(y,x)\cdot\varTheta(y)\rmd y-\int_{\{y_{d}=0\}}\theta_{0}(y)\frac{\partial}{\partial y_{d}}G_{d}(y,x)\rmd y_{1}\cdots\rmd y_{d-1}.
\]
The claims follow immediately.
\end{proof}

\section{Twin Brownian particles}

In the remainder of the paper, $u(x,t)$ is a time-dependent vector
field on $\mathbb{R}^{d}$ (where $d=2$ or $d=3$). $u(x,t)$ may
be the velocity of an incompressible fluid flow in $\mathbb{R}^{d}$
or the appropriate (divergence-free) extension of the velocity of
an incompressible fluid flow in $D\subset\mathbb{R}^{d}$. Let $\nu>0$
be the kinematic viscosity constant, and $\kappa>0$ the thermal diffusivity constant.
We then introduce two families of random particles $X$ and $Y$,
defined by the following stochastic differential equations:
\begin{equation}
\rmd X_{t}^{\xi}=u(X_{t}^{\xi},t)\rmd t+\sqrt{2\nu}\rmd B_{t}^{\nu},\quad X_{0}^{\xi}=\xi\label{X-P-01}
\end{equation}
and
\begin{equation}
\rmd Y_{t}^{\xi}=u(Y_{t}^{\xi},t)\rmd t+\sqrt{2\kappa}\rmd B_{t}^{\kappa},\quad Y_{0}^{\xi}=\xi\label{Y-P-01}
\end{equation}
for every $\xi\in\mathbb{R}^{d}$. Here $B^{\nu}$ and $B^{\kappa}$
are two independent standard $d$-dimensional Brownian motions on
some probability space. 

The transition probability density functions for $X$ (resp. for $Y$)
are denoted by $p_{\nu}(s,x,t,y)$ and $p_{\kappa}(s,x,t,y)$ respectively. 

Let $A_{j}^{i}=\frac{\partial}{\partial x_{j}}u^{i}$ be the entries of Jacobian matrix of $u(x,t)$. Given a domain $D\subset\mathbb{R}^{d}$,
we introduce two gauge functionals $Q=\left(Q_{j}^{i}(\eta,t;s)\right)$
and $R=\left(R_{j}^{i}(\eta,t;s)\right)$ (for $0\leq s\leq t$ and
$\eta\in\mathbb{R}^{d}$), where $i,j=1,\cdots,d$, defined by the
following ordinary differential equations:
\begin{equation}
\frac{\rmd}{\rmd s}Q_{j}^{i}(\eta,t;s)=-Q_{k}^{i}(\eta,t;s)1_{D}(X_{s}^{\eta})A_{j}^{k}(X_{s}^{\eta},s),\quad Q_{j}^{i}(\eta,t;t)=\delta_{ij}\label{Ga-01}
\end{equation}
and
\begin{equation}
\frac{\rmd}{\rmd s}R_{i}^{j}(\eta,t;s)=R_{i}^{l}(\eta,t;s)1_{D}(Y_{s}^{\eta})A_{l}^{j}(Y_{s}^{\eta},s),\quad R_{i}^{j}(\eta,t;t)=\delta_{ij}\label{Ga-02}
\end{equation}
respectively. 

The diffusions $X^{\xi,\tau}$ and $Y^{\eta,\tau}$ (where $\xi,\eta\in\mathbb{R}^{d}$
and $\tau\geq0$) are defined as the (weak) solutions of the stochastic
differential equations: 
\begin{equation}
\rmd X^{\xi,\tau}_t=u(X^{\xi,\tau}_t,t)\rmd t+\sqrt{2\nu}\rmd B^{\nu}_t,\quad X_{s}^{\xi,\tau}=\xi \ \ \text{ for } s\leq \tau\label{eq:T-01-1}
\end{equation}
and
\begin{equation}
\rmd Y_{t}^{\eta,\tau}=u(Y_{t}^{\eta,\tau},t)\rmd t+\sqrt{2\kappa}\rmd B^{\kappa}_t,\quad Y_{s}^{\eta,\tau}=\eta \ \ \text{ for } s\leq\tau\label{eq:T-01-2}
\end{equation}
for $t\geq 0$.

\section{Unbounded fluid flows}

In this section, we consider an ideal fluid flow in $\R^d$, $d=2$ or $d=3$, with an external source that supplies heat to the fluid. 
\subsection{The Oberbeck-Boussinesq equations in \texorpdfstring{$\R^d$}{Rd}}
In the Oberbeck-Boussinesq model, the fluid flow has no space constraint,
so such a model serves as a model of homogeneous turbulent flows
with heat transfer. Suppose the temperature gradient is small so
that the density of the fluid is nearly constant. Therefore the model
is described by the Boussinesq equations on the whole space $\mathbb{R}^{d}$.
Suppose the heat source is represented by temperature $\theta(x)$,
so that the flow is described by (\ref{NS-g1}, \ref{NS-g2}, \ref{NS-g3})
in $\mathbb{R}^{d}$. The Navier-Stokes equations may be written as
\begin{equation}
\left(\nu\Delta-u\cdot\nabla-\frac{\partial}{\partial t}\right)u-\nabla P+f(\theta)=0\quad\textrm{ in }\mathbb{R}^{d}\label{M-NS01}
\end{equation}
where the dimension $d=2$ or $3$, and the interaction force $f=(f_{1},\cdots,f_{d})$.
$u$ is divergence-free, i.e. $\nabla\cdot u=0$, and the temperature transport
equation may be written as
\begin{equation}
\left(\kappa\Delta-u\cdot\nabla-\frac{\partial}{\partial t}\right)\theta=0\quad\textrm{ in }\mathbb{R}^{d}.\label{M-TNS-01}
\end{equation}

Let $\omega=\nabla\wedge u$ and $\varTheta=\nabla\theta$. If $d=3$,
the vorticity transport equation for $\omega$ is the following
PDE
\[
\left(\nu\Delta-u\cdot\nabla-\frac{\partial}{\partial t}\right)\omega+A\omega+F=0\quad\textrm{ in }\mathbb{R}^{3},
\]
where $(A\omega)^{i}=A_{l}^{i}\omega^{l}$ and
\[
F(x,t)=\varTheta(x,t)\wedge f'(\theta(x,t)),\quad\textrm{ with }f'=(f_{1}',f_{2}',f_{3}').
\]
For example, we may take $f_{i}(\theta)=(\theta-\theta_{0})\delta_{i3}$, $i=1,2,3$, then 
\[
F^{i}=\delta_{i1}\varTheta_{2}-\delta_{i2}\varTheta_{1}-\delta_{i1}\frac{\partial\theta_{0}}{\partial x_{2}}+\delta_{i2}\frac{\partial\theta_{0}}{\partial x_{1}}.
\]

For two-dimensional flows,
\[
\left(\nu\Delta-u\cdot\nabla-\frac{\partial}{\partial t}\right)\omega+F=0\quad\textrm{ in }\mathbb{R}^{2}
\]
where $F=\varTheta_{2}f_{1}'(\theta)-\varTheta_{1}f_{2}'(\theta)$.

The temperature gradient $\varTheta$ satisfies the following transport
equations
\begin{equation}
\left(\kappa\Delta-u\cdot\nabla-\frac{\partial}{\partial t}\right)\varTheta-A\varTheta=0\quad\textrm{ in }\mathbb{R}^{d}\label{t-32}
\end{equation}
where $(A\varTheta)_{j}=A_{j}^{l}\varTheta_{l}$. 

\subsection{Representations of the vorticity and the temperature gradient}

In this part, we shall work out the functional integral representations for
the vorticity and the temperature gradient. We simplify our notation and omit $u$ in the notation of transition probability densities, using $p_\nu(s,\eta, t,\xi)$ and $p_\kappa(s,\eta, t,\xi)$ instead of $p_{\nu,u}(s,\eta, t,\xi)$ and $p_{\kappa,u}(s,\eta, t,\xi)$.
\begin{lem}
\label{lem6.2}The temperature gradient $\varTheta=\nabla\theta$
has the following integral representation:
\[
\varTheta(\xi,t)=\int_{\mathbb{R}^{d}}\mathbb{E}\left[\left.R(\eta,t;0)\varTheta(\eta,0)\right|Y_{t}^{\eta}=\xi\right]p_{\kappa}(0,\eta,t,\xi)\rmd \eta
\]
for $\xi\in\mathbb{R}^{d}$ and $t>0$, where $R$ are defined by
(\ref{Ga-02}) with $D=\mathbb{R}^{d}$ ($d=2$
or $d=3$). Here we have used the convention that $(R\varTheta)_{j}=R_{j}^{l}\varTheta_{l}$. 
\end{lem}

\begin{proof}
Apply Theorem \ref{thm7.2-1-1} to the parabolic equations (\ref{t-32}). 
\end{proof}
%
\begin{lem}
\begin{enumerate}[{(1)}]
    \item Suppose $d=3$. The vorticity $\omega=\nabla\wedge u$
    possesses the following integral representation:
    \begin{align}
    \omega(\xi,t)  =&\int_{\mathbb{R}^{3}}\mathbb{E}\left[\left.Q(\eta,t;0)\omega(\eta,0)\right|X_{t}^{\eta}=\xi\right]p_{\nu}(0,\eta,t,\xi)\rmd \eta\nonumber \\
     & +\int_{0}^{t}\int_{\mathbb{R}^{3}}\mathbb{E}\left[\left.Q(\eta,t;s)F(X_{s}^{\eta},s)\right|X_{t}^{\eta}=\xi\right]p_{\nu}(0,\eta,t,\xi)\rmd \eta\rmd s\label{om-re51}
    \end{align}
    for $\xi\in\mathbb{R}^{3}$ and $t>0$, where $Q$ is defined by
    (\ref{Ga-01}) with $D=\mathbb{R}^{3}$. Here $(Q\omega)^{i}=Q_{l}^{i}\omega^{l}$ and $(QF)^{i}=Q_{l}^{i}F^{l}$.
    
    \item Suppose $d=2$, so that $\omega=\frac{\partial}{\partial x_{1}}u^{2}-\frac{\partial}{\partial x_{2}}u^{1}$,
    and $F=\varTheta_{1}f_{2}'(\theta)-\varTheta_{2}f_{1}'(\theta)$.
    $\omega$ has the following representation:
    \begin{align}
    \omega(\xi,t)  =&\int_{\mathbb{R}^{2}}p_{\nu}(0,\eta,t,\xi)\omega(\eta,0)\rmd \eta\nonumber \\
     & +\int_{0}^{t}\int_{\mathbb{R}^{2}}\mathbb{E}\left[\left.F(X_{s}^{\eta},s)\right|X_{t}^{\eta}=\xi\right]p_{\nu}(0,\eta,t,\xi)\rmd\eta\rmd s\label{2D-new-91}
    \end{align}
    for $\xi\in\mathbb{R}^{2}$ and $t>0$.
    \end{enumerate}
\end{lem}

The proof also follows from Theorem \ref{thm7.2-1-1} immediately.

\subsection{Random vortex dynamics}
Using the integral representations derived above, we are able to identify the random vortex dynamical systems associated with the heat-conducting incompressible fluid flows in $\R^2$ and $\R^3$.
\subsubsection{\label{subsec52}Random vortex dynamics in two-dimensional case}

Now we are in a position to formulate the vortex dynamics for the
model. Let us start with the two-dimensional model.
As we have pointed out that for two-dimensional flows the vorticity
$\omega$ is a scalar function, and $A\omega$ vanishes, so that one
does not need the gauge functional $Q$. 

Let $u(x,t)$ and $\theta(x,t)$ be the regular solutions to the model
(\ref{M-NS01}, \ref{M-TNS-01}) when $d=2$. 
\begin{thm}
\label{thm6.4} The following representation holds:
\begin{align*}
u(x,t)  =&\int_{\mathbb{R}^{2}}\mathbb{E}\left[K_{2}(X_{t}^{\eta},x)\wedge\omega(\eta,0)\right]\rmd \eta\nonumber \\
 & +\int_{0}^{t}\int_{\mathbb{R}^{2}}\mathbb{E}\left[K_{2}(X_{t}^{\eta},x)\wedge F(X_{s}^{\eta},s)\right]\rmd\eta\rmd s
\end{align*}
for $x\in\mathbb{R}^{2}$ and $t>0$.
\end{thm}

\begin{proof}
By using the Biot-Savart law (see Lemma \ref{lem4.9}, Eq. (\ref{2D-new-91})),
we obtain that
\begin{align*}
u(x,t)  =&\int_{\mathbb{R}^{2}}K_{2}(\xi,x)\wedge\omega(\xi,t)\rmd\xi\\
  =&\int_{\mathbb{R}^{2}}K_{2}(\xi,x)\wedge\left( \int_{\mathbb{R}^{2}}p_{\nu}(0,\eta,t,\xi)\omega(\eta,0)\rmd\eta\right) \rmd \xi\\
 & +\int_{\mathbb{R}^{2}}K_{2}(\xi,x)\wedge\left( \int_{0}^{t}\int_{\mathbb{R}^{2}}\mathbb{E}\left[\left.F(X_{s}^{\eta},s)\right|X_{t}^{\eta}=\xi\right]p_{\nu}(0,\eta,t,\xi)\rmd\eta\rmd s\right) \rmd\xi
\end{align*}
and the conclusion then follows from an application of Fubini's theorem.
\end{proof}
Similarly, we have the following theorem.
\begin{thm}
\label{thm6.5} The following representation for the temperature $\theta(x,t)$
holds:
\[
\theta(x,t)=-\int_{\mathbb{R}^{2}}\mathbb{E}\left[K_{2}(Y_{t}^{\eta},x)\cdot R(\eta,t;0)\varTheta(\eta,0)\right]\rmd \eta
\]
for $x\in\mathbb{R}^{2}$ and $t>0$, where $R=(R_{i}^{j})$ is defined
by Eq. (\ref{Ga-02}), and $R\varTheta=\left(R_{i}^{l}\varTheta_{l}\right)$.
\end{thm}

\begin{proof}
This follows from (\ref{lem4.9}) and Lemma \ref{lem6.2} immediately.
\end{proof}
Theorem \ref{thm6.4}, Theorem \ref{thm6.5}, the $X$ particle equation
(\ref{X-P-01}), the $Y$ particle equation (\ref{Y-P-01}), the gauge
equation (\ref{Ga-02}) (with $D=\mathbb{R}^{2}$), together with
the equations
\begin{equation}
A_{j}^{i}=\frac{\partial}{\partial x_{j}}u^{i},\quad\varTheta_{i}=\frac{\partial}{\partial x_{i}}\theta\label{S-01}
\end{equation}
and
\[
F(x,t)=\varTheta_{2}(x,t)f_{1}'(\theta(x,t))-\varTheta_{1}(x,t)f_{2}'(\theta(x,t))
\]
give rise to a closed random vortex system for the two-dimensional
model. 

\subsubsection{Random vortex dynamics in three-dimensional case}

In this part, we work out the random vortex system for three-dimensional flows defined by the equations of motion (\ref{M-NS01},
\ref{M-TNS-01}) where $d=3$, and the details of computations will
be omitted.

Let $u(x,t)$ and $\theta(x,t)$ be the regular solutions to (\ref{M-NS01},
\ref{M-TNS-01}) with $d=3$.
\begin{thm}
\label{thm66} The following representation formula for the velocity
$u(x,t)$ holds:
\begin{align*}
u(x,t)  =&\int_{\mathbb{R}^{3}}\mathbb{E}\left[K_{3}(X_{t}^{\eta},x)\wedge Q(\eta,t;0)\omega(\eta,0)\right]\rmd \eta\nonumber \\
 & +\int_{\mathbb{R}^{3}}\int_{0}^{t}\mathbb{E}\left[K_{3}(X_{t}^{\eta},x)\wedge Q(\eta,t;s)F(X_{s}^{\eta},s)\right]\rmd\eta\rmd s
\end{align*}
for $x\in\mathbb{R}^{3}$ and $t>0$, where $Q\omega=(Q_{l}^{i}\omega^{l})$
and $QF=(Q_{l}^{i}F^{l})$, $Q=(Q_{j}^{i})$ is defined by Eq. (\ref{Ga-01})
with $D=\mathbb{R}^{3}$.
\end{thm}

This formula follows from Eq. (\ref{t0t-BS01-1}) and Eq. (\ref{om-re51}).
\begin{thm}
\label{thm67}The temperature $\theta(x,t)$ possesses the following
representation:
\[
\theta(x,t)=-\int_{\mathbb{R}^{3}}\mathbb{E}\left[K_{3}(Y_{t}^{\eta},x)\cdot R(\eta,t;0)\varTheta(\eta,0)\right]\rmd \eta
\]
for every $x\in\mathbb{R}^{3}$ and $t>0$, where $R=(R_{i}^{j})$
is defined by (\ref{Ga-02}) with $D=\mathbb{R}^{3}$, and $R\varTheta=(R_{i}^{l}\varTheta_{l})$.
\end{thm}

Theorem \ref{thm66}, Theorem \ref{thm67}, the particle defining
equations (\ref{X-P-01}, \ref{Y-P-01}), the gauge functional defining
equations (\ref{Ga-01}, \ref{Ga-02}), together with the relations
\[
A_{j}^{i}=\frac{\partial}{\partial x_{j}}u^{i},\quad\varTheta_{i}=\frac{\partial}{\partial x_{i}}\theta,\quad F=\varTheta\wedge f'(\theta)
\]
give rise to the random vortex dynamic system for three-dimensional
flows.

\section{Wall-bounded flows }

In this section, we study the random vortex dynamics associated with
the Oberbeck-Boussinesq flows along a flat plate, which is very important
examples of wall-bounded fluid flows, in particular, wall-bounded turbulent
flows with thermal conduction from the solid wall.
\subsection{The Oberbeck-Boussinesq equations in \texorpdfstring{$\R_+^d$}{Rd}}

In the Oberbeck-Boussinesq model, it is assumed that the fluid density
$\rho$ diverges little from the fluid density $\rho_{0}$ at
the fluid temperature $\theta_{0}$:
\[
\rho=\rho_{0}\left(1-\alpha\left(\theta-\theta_{0}\right)\right),
\]
where $\theta$ is the fluid temperature and $\alpha$ is a very small
constant. Therefore the density $\rho$ is considered a constant.
The Oberbeck-Boussinesq model is composed of the following partial differential
equations on $\mathbb{R}_{+}^{d}$ (where the dimension $d=3$ or
$d=2$). The Navier-Stokes equations with interaction from the
heat supply on the boundary
\[
\left(\nu\Delta-u\cdot\nabla-\frac{\partial}{\partial t}\right)u-\nabla\left(\frac{P}{\rho_{0}}+gx_{d}\right)+f(\theta-\theta_{0})=0\quad\textrm{ in }\mathbb{R}_{+}^{d},
\]
where the interaction force $f(\theta-\theta_{0})$, with its components
\[f_{i}(\theta-\theta_{0})=\alpha g(\theta-\theta_{0})\delta_{id},\quad 
i=1,\cdots,d.
\] 
Here $g$ is the gravity constant, and $P$ denotes the
pressure. The continuity equation
\[
\nabla\cdot u=0\quad\textrm{ in }\mathbb{R}_{+}^{d},
\]
and the equation of energy in terms of the temperature $\theta$:
\begin{equation}
\left(\kappa\Delta-u\cdot\nabla-\frac{\partial}{\partial t}\right)\theta=0\quad\textrm{ in }\mathbb{R}_{+}^{d},\label{eq:B-heat-01}
\end{equation}
where $\kappa=k/(\rho_{0}c_{P})$ (perfect gas) or $k/(\rho_{0}c)$
(for liquid) is a positive constant. 

The velocity $u(x,t)$ satisfies the no-slip condition, i.e. $u(x,t)=0$
when $x_{d}=0$, so that $u(x,t)$ is extended via reflection about
$\{x_{d}=0\}$ to a time-dependent vector field on $\mathbb{R}^{d}$:
\[
u^{i}(x,t)=u^{i}(\bar{x},t)\quad\textrm{ for }i=1,\cdots,d-1,\quad u^{d}(x,t)=-u^{d}(\bar{x},t),
\]
so that $u(\cdot,t)$ is divergence-free on the whole space $\mathbb{R}^{d}$
in the distribution sense.

Let $\omega=\nabla\wedge u$ and $\varTheta=\nabla\theta$. We assume
that they are continuous up to the boundary $\{x_{d}=0\}$.

The vorticity transport equation for three-dimensional flows is
given as the following: 
\[
\left(\nu\Delta-u\cdot\nabla-\frac{\partial}{\partial t}\right)\omega+A\omega+F=0\quad\textrm{ in }\mathbb{R}_{+}^{3},
\]
where $A=(A_{l}^{i})$ with $A_{l}^{i}=\frac{\partial}{\partial x_{l}}u^{i}$
for $i,l=1,2,3$, $(A\omega)^{i}=A_{l}^{i}\omega^{l}$ and 
\[
F=(\varTheta-\nabla\theta_{0})\wedge f'(\theta-\theta_{0}).
\]

For two-dimensional flows, the non-linear vorticity stretching term
$A\omega$ vanishes identically. Therefore
\[
\left(\nu\Delta-u\cdot\nabla-\frac{\partial}{\partial t}\right)\omega+F=0\quad\textrm{ in }\mathbb{R}_{+}^{2},
\]
where
\[
F=(\varTheta_{1}-\frac{\partial}{\partial x_{1}}\theta_{0})f_{2}'(\theta-\theta_{0})-(\varTheta_{2}-\frac{\partial}{\partial x_{2}}\theta_{0})f_{1}'(\theta-\theta_{0}).
\]

Similarly, by differentiating the heat equation (\ref{eq:B-heat-01})
one obtains the evolution equations for the temperature gradient
\[
\left(\kappa\Delta-u\cdot\nabla-\frac{\partial}{\partial t}\right)\varTheta-A\varTheta=0\quad\textrm{ in }\mathbb{R}_{+}^{d}.
\]
Here we have used the convention that $(A\varTheta)_{j}=A_{j}^{l}\varTheta_{l}$.

An essential difference exists between this bounded case and the case discussed in Section
4, where there is no physical boundary for the fluid flow. For wall-bounded
flows, it is impossible to determine the boundary value of the vorticity $\omega$ and the temperature
gradient $\varTheta$, although the boundary
vorticity can be identified as the normal stress at the boundary, cf. Anderson \cite{Anderson 1986}. Nevertheless, the boundary
the temperature gradient can not be specified either, while by the definition
of the Oberbeck-Boussinesq motion, the boundary temperature gradient
is considered small, so it can be treated as zero in numerical
experiments.

\subsection{Diffusions in \texorpdfstring{$\mathbb{R}^d_+$}{Rd}}

Let $D=\mathbb{R}_{+}^{d}$. The reflection $\mathscr{R}$ about the
hyperplane $\{x_{d}=0\}$ on $\mathbb{R}^{d}$ is the linear map which
sends $x=(x_{1},\cdots,x_{d})$ to $\mathscr{R}x=\overline{x}=(x_{1},\cdots,x_{d-1},-x_{d})$
for every $x\in\mathbb{R}^{d}$.

The velocity $u(x,t)$ is extended to a time-dependent and divergence-free
vector field on $\mathbb{R}^{d}$ such that $\mathscr{R}(b(x,t))=b(\mathscr{R}x,t)$
for every $x\in\mathbb{R}^{d}$ and $t\geq0$. That is
\[
u^{i}(\overline{x},t)=u^{i}(x,t)\quad\textrm{ for }i=1,\ldots,d-1\quad\textrm{ and }u^{d}(\bar{x},t)=-u^{d}(x,t))
\]
for all $x\in\mathbb{R}^{d}$. Then $\overline{X^{\xi,\tau}}$ and
$X^{\overline{\xi},\tau}$ have the same distribution, so that $p_{\lambda,u}(\tau,\bar{\xi},t,\overline{x})=p_{\lambda,u}(\tau,\xi,t,x)$,
and 
\begin{equation}
p_{\lambda,u}^{D}(\tau,\xi,t,x)=p_{\lambda,u}(\tau,\xi,t,x)-p_{\lambda,u}(\tau,\overline{\xi},t,x)\label{eq:ST-pdf-01}
\end{equation}
for any $\xi,x\in\mathbb{R}_{+}^{d}$. The last equality may be verified
by checking that the right-hand side is indeed the Green function
of the backward parabolic operator $\lambda\Delta+b\cdot\nabla+\frac{\partial}{\partial t}$.
In particular
\begin{equation}
\mathbb{E}\left[\varphi(X_{t}^{\xi,\tau})1_{\left\{ t<\zeta(X^{\xi,\tau})\right\} }\right]=\mathbb{E}\left[1_{D}(X_{t}^{\xi,\tau})\varphi(X_{t}^{\xi,\tau})\right]-\mathbb{E}\left[1_{D}(X_{t}^{\overline{\xi},\tau})\varphi(X_{t}^{\overline{\xi},\tau})\right]\label{stop-for01}
\end{equation}
for any bounded and Borel measurable function $\varphi$. This approach has been put forward in \cite{QQZW2022} for simulating
the solutions of the Navier-Stokes equations within the boundary layer. 

\subsection{Representations of the vorticity and the temperature gradient}

We observe that the vorticity $\omega$, the temperature $\theta$
and the temperature gradient $\varTheta$ possess non-homogeneous
Dirichlet boundary conditions, i.e. 
\[
    \omega(x,t) = \sigma(x_{1},\cdots,x_{d-1},t) \quad\text{on } \{x_d=0\},
\]
\[
    \varTheta(x,t) = \gamma(x_{1},\cdots,x_{d-1},t) \quad\text{on } \{x_d=0\}
\]
and 
\[
\theta(x,t) = \theta_0(x_{1},\cdots,x_{d-1}) \quad\text{on } \{x_d=0\}
\]
for all $t\geq 0$, where $\theta_0$ is considered the external heat supply on the boundary. 
so we use the cutting-off technique
to reduce their boundary conditions to homogeneous ones.


Let $\phi:[0,\infty)\rightarrow[0,1]$ be a smooth cut-off function
such that $\phi(r)=1$ for $r\in[0,1/3)$ and $\phi(r)=0$ for $r\geq2/3$.
Let
\[
\omega^{\varepsilon}(x,t)=\omega(x,t)-\sigma(x_{1},\cdots,x_{d-1},t)\phi\left(\frac{x_{d}}{\varepsilon}\right)
\]
and
\[
\theta^{\varepsilon}(x,t)=\theta(x,t)-\theta_{0}(x_{1},\cdots,x_{d-1})\phi\left(\frac{x_{d}}{\varepsilon}\right)
\]
for every $\varepsilon>0$. It is clear from the definition that
\[
\lim_{\varepsilon\rightarrow\infty}\omega^{\varepsilon}(x,t)=\omega(x,t)-1_{\{x_{d}\geq0\}}\sigma\left(x_{1},\cdots,x_{d-1},t\right),
\]
\[
\lim_{\varepsilon\rightarrow\infty}\theta^{\varepsilon}(x,t)=\theta(x,t)-1_{\{x_{d}\geq0\}}\theta_{0}\left(x_{1},\cdots,x_{d-1}\right),
\]
\[
\lim_{\varepsilon\rightarrow0+}\omega^{\varepsilon}(x,t)=\omega(x,t)-1_{\{x_{d}=0\}}\sigma(x_{1},\cdots,x_{d-1},t)
\]
and
\[
\lim_{\varepsilon\rightarrow0+}\theta^{\varepsilon}(x,t)=\theta(x,t)-1_{\{x_{d}=0\}}\theta_{0}(x_{1},\cdots,x_{d-1}).
\]
Then $\omega^{\varepsilon}$ and $\theta^{\varepsilon}$ satisfy the
homogeneous Dirichlet boundary conditions $\omega^{\varepsilon}(x,t)=0$
and $\theta^{\varepsilon}(x,t)=0$ on $\{x_{d}=0\}$. $\omega^{\varepsilon}$
and $\theta^{\varepsilon}$, for every $\varepsilon>0$, evolve according
to the following parabolic equations:
\begin{equation}
\left(\nu\Delta-u\cdot\nabla-\frac{\partial}{\partial t}\right)\omega^{\varepsilon}+A\omega^{\varepsilon}+F+\chi_{\varepsilon}=0\quad\textrm{ in }\mathbb{R}_{+}^{d},\label{eq:ome-01}
\end{equation}
(here if the dimension $d=2$, then the non-linear stretching term
$A\omega^{\varepsilon}$ may be dropped),

\begin{equation}
\left(\kappa\Delta-u\cdot\nabla-\frac{\partial}{\partial t}\right)\varTheta^{\varepsilon}-A\varTheta^{\varepsilon}+\alpha_{\varepsilon}=0\quad\textrm{ in }\mathbb{R}_{+}^{d}\label{Th-51}
\end{equation}
and
\begin{equation}
\left(\kappa\Delta-u\cdot\nabla-\frac{\partial}{\partial t}\right)\theta^{\varepsilon}+\beta_{\varepsilon}=0\quad\textrm{ in }\mathbb{R}_{+}^{d},\label{eq:Tem-e-01}
\end{equation}
where the correction terms are given by the following formulae:
\begin{align}
\chi_{\varepsilon}(x,t)  =&\phi\left(\frac{x_{d}}{\varepsilon}\right)\left(\nu\Delta-\sum_{j=1}^{d-1}u^{j}(x,t)\frac{\partial}{\partial x_{j}}-\frac{\partial}{\partial t}+A\right)\sigma(x_{1},\cdots,x_{d-1},t)\nonumber \\
 & -\frac{1}{\varepsilon}\phi'\left(\frac{x_{d}}{\varepsilon}\right)\sigma(x_{1},\cdots,x_{d-1},t)u^{d}(x,t)+\nu\frac{1}{\varepsilon^{2}}\phi''\left(\frac{x_{d}}{\varepsilon}\right)\sigma(x_{1},\cdots,x_{d-1},t),\label{eq:oe-er-01}
\end{align}
(if the dimension $d=2$, then the term involving $A$ may be dropped),
\begin{align*}
\alpha_{\varepsilon}(x,t)  =&\phi\left(\frac{x_{d}}{\varepsilon}\right)\left(\kappa\Delta-\sum_{j=1}^{d-1}u^{j}(x,t)\frac{\partial}{\partial x_{j}}-\frac{\partial}{\partial t}-A\right)\gamma(x_{1},\cdots,x_{d-1},t)\nonumber \\
 & -\frac{1}{\varepsilon}\phi'\left(\frac{x_{d}}{\varepsilon}\right)\gamma(x_{1},\cdots,x_{d-1},t)u^{d}(x,t)+\kappa\frac{1}{\varepsilon^{2}}\phi''\left(\frac{x_{d}}{\varepsilon}\right)\gamma(x_{1},\cdots,x_{d-1},t),
\end{align*}
and
\begin{align*}
\beta_{\varepsilon}(x,t)  =&\phi\left(\frac{x_{d}}{\varepsilon}\right)\left(\kappa\Delta-\sum_{j=1}^{d-1}u^{j}\frac{\partial}{\partial x_{j}}\right)\theta_{0}(x_{1},\cdots,x_{d-1})\nonumber \\
 & -\frac{1}{\varepsilon}\phi'\left(\frac{x_{d}}{\varepsilon}\right)\theta_{0}(x_{1},\cdots,x_{d-1})u^{d}(x,t)+\kappa\frac{1}{\varepsilon^{2}}\phi''\left(\frac{x_{d}}{\varepsilon}\right)\theta_{0}(x_{1},\cdots,x_{d-1}).
\end{align*}
Here $x=(x_{1},\cdots,x_{d})$
and the Laplacian $\Delta$ on $\sigma$ and $\theta_{0}$ is the
(boundary) $(d-1)$-dimensional Laplacian $\sum_{j=1}^{d-1}\frac{\partial^{2}}{\partial x_{j}^{2}}$. 

It is clear from the definition that the following limits
\[
\lim_{\varepsilon\rightarrow\infty}\chi_{\varepsilon}(x,t)=1_{\{x_{d}\geq0\}}\left(\nu\Delta-\sum_{j=1}^{d-1}u^{j}(x,t)\frac{\partial}{\partial x_{j}}-\frac{\partial}{\partial t}+A\right)\sigma(x_{1},\cdots,x_{d-1},t),
\]
\[
\lim_{\varepsilon\rightarrow\infty}\alpha_{\varepsilon}(x,t)=1_{\{x_{d}\geq0\}}\left(\kappa\Delta-\sum_{j=1}^{d-1}u^{j}(x,t)\frac{\partial}{\partial x_{j}}-\frac{\partial}{\partial t}-A\right)\gamma(x_{1},\cdots,x_{d-1},t),
\]
and
\[
\lim_{\varepsilon\rightarrow\infty}\beta_{\varepsilon}(x,t)=1_{\{x_{d}\geq0\}}\left(\kappa\Delta-\sum_{j=1}^{d-1}u^{j}(x,t)\frac{\partial}{\partial x_{j}}\right)\theta_{0}(x_{1},\cdots,x_{d-1})
\]
exist, where $x=(x_{1},\cdots, x_d)$. Unfortunately, these limits involve
the velocity of the outer layer flow, and the limits as $\varepsilon\rightarrow0+$
do not exist in the ordinary sense, so we need to evaluate them under
integration. However, since $u(x,t)$ satisfies the no-slip condition,
so that 
\begin{equation}
\lim_{\varepsilon\rightarrow0+}\phi\left(\frac{x_{d}}{\varepsilon}\right)\left(\nu\Delta-\sum_{j=1}^{d-1}u^{j}\frac{\partial}{\partial x_{j}}-\frac{\partial}{\partial t}+A\right)\sigma=1_{\{x_{d}=0\}}\left(\nu\Delta-\frac{\partial}{\partial t}+A\right)\sigma,\label{eq:oe-lim0p}
\end{equation}
\[
\lim_{\varepsilon\rightarrow0+}\phi\left(\frac{x_{d}}{\varepsilon}\right)\left(\kappa\Delta-\sum_{j=1}^{d-1}u^{j}\frac{\partial}{\partial x_{j}}-\frac{\partial}{\partial t}-A\right)\gamma=1_{\{x_{d}=0\}}\left(\kappa\Delta-\frac{\partial}{\partial t}-A\right)\gamma,
\]
and
\[
\lim_{\varepsilon\rightarrow0+}\phi\left(\frac{x_{2}}{\varepsilon}\right)\left(\kappa\Delta-\sum_{j=1}^{d-1}u^{j}\frac{\partial}{\partial x_{j}}\right)\theta_{0}=\kappa1_{\{x_{d}=0\}}\Delta\theta_{0},
\]
which vanish almost surely in $x=(x_{1},\cdots,x_{d})$ with respect
to the Lebesgue measure on $\mathbb{R}^{d}$. 


By using the forward Feynman-Kac formula, we may work out the functional
integral representations for $\theta^{\varepsilon}$, $\varTheta^{\varepsilon}$
and $\omega^{\varepsilon}$ respectively. Again, as in previous sections, we drop $u$ in the notation of the transition probability density functions. 
\begin{prop}
For every $\varepsilon>0$, the temperature $\theta$
has the following representation
\begin{align*}
\theta(\xi,t) =& \phi\left(\frac{\xi_{d}}{\varepsilon}\right)\theta_{0}(\xi_{1},\cdots,\xi_{d-1})+\int_{\mathbb{R}_{+}^{d}}\mathbb{E}\left[\left.1_{\{t<\zeta(Y^{\eta})\}}\right|Y_{t}^{\eta}=\xi\right]p_{\kappa}(0,\eta,t,\xi)\theta^{\varepsilon}(\eta,0)\rmd\eta\nonumber \\
 & +\int_{0}^{t}\int_{\mathbb{R}_{+}^{d}}\mathbb{E}\left[\left.1_{\left\{ s>\gamma_{t}(Y^{\eta})\right\} }\beta_{\varepsilon}\left(Y_{s}^{\eta},s\right)\right|Y_{t}^{\eta}=\xi\right]p_{\kappa}(0,\eta,t,\xi)\rmd\eta\rmd t
\end{align*}
for every $\xi\in\mathbb{R}_{+}^{d}$ and $t>0$, where $d=2$ or
$d=3$. 
\end{prop}

\begin{proof}
It follows from the PDE (\ref{eq:Tem-e-01}) and Theorem \ref{thm7.2-1}
immediately.
\end{proof}
There is a similar representation for the temperature gradient.
\begin{prop}
\label{prop53}For every $\varepsilon>0$, the following representation
holds:
\begin{align*}
\varTheta(\xi,t)  =&\phi\left(\frac{\xi_{d}}{\varepsilon}\right)\varTheta(\xi_{1},\cdots,\xi_{d-1},0,t) \\
 & +\int_{\mathbb{R}_{+}^{d}}\mathbb{E}\left[\left.1_{\{t<\zeta(Y^{\eta})\}}R(\eta,t;0)\varTheta^{\varepsilon}(\eta,0)\right|Y_{t}^{\eta}=\xi\right]p_{\kappa}(0,\eta,t,\xi)\rmd \eta\nonumber \\
 & +\int_{0}^{t}\int_{\mathbb{R}_{+}^{d}}\mathbb{E}\left[\left.1_{\left\{ s>\gamma_{t}(Y^{\eta})\right\} }R(\eta,t;s)\alpha_{\varepsilon}\left(Y_{s}^{\eta},s\right)\right|Y_{t}^{\eta}=\xi\right]p_{\kappa}(0,\eta,t,\xi)\rmd \eta\rmd s
\end{align*}
for all $\xi\in\mathbb{R}_{+}^{d}$ and $t>0$, where $d=2$ and $d=3$, and $(R\varTheta^\varepsilon)_{j}=R_{j}^{l}\varTheta^\varepsilon_{l}$. 
\end{prop}

The representation for the temperature follows immediately from (\ref{Th-51})
and Theorem \ref{thm7.2-1}.

Finally, the functional integral representation for the vorticity in dimension
two takes a simpler form, so we state it separately.
\begin{prop}
\label{prop54}Suppose $d=2$. Then the vorticity $\omega$
possesses the following integral representation:
\begin{align*}
\omega(\xi,t) =& \phi\left(\frac{\xi_{2}}{\varepsilon}\right)\sigma(\xi_{1},t)+\int_{\mathbb{R}_{+}^{2}}\mathbb{E}\left[\left.1_{\{t<\zeta(X^{\eta})\}}\right|X_{t}^{\eta}=\xi\right]p_{\nu}(0,\eta,t,\xi)\omega^{\varepsilon}(\eta,0)\rmd\eta \\
 & +\int_{0}^{t}\int_{\mathbb{R}_{+}^{2}}\mathbb{E}\left[\left.1_{\left\{ s>\gamma_{t}(X^{\eta})\right\} }F(X_{s}^{\eta},s)\right|X_{t}^{\eta}=\xi\right]p_{\nu}(0,\eta,t,\xi)\rmd\eta\rmd s \\
 & +\int_{0}^{t}\int_{\mathbb{R}_{+}^{2}}\mathbb{E}\left[\left.1_{\left\{ s>\gamma_{t}(X^{\eta})\right\} }\chi_{\varepsilon}(X_{s}^{\eta},s)\right|X_{t}^{\eta}=\xi\right]p_{\nu}(0,\eta,t,\xi)\rmd\eta\rmd s
\end{align*}
for every $\xi=(\xi_{1},\xi_{2})\in\mathbb{R}_{+}^{2}$ and $t>0$.
\end{prop}

In particular, for two-dimensional flow, we do not need to introduce
the gauge functional $Q$. However, for three-dimensional flows, we
need both gauge functionals $Q$ and $R$, defined by (\ref{Ga-01})
and (\ref{Ga-02}) respectively with $D=\mathbb{R}_{+}^{3}$. 
\begin{prop}
\label{prop55} Suppose $d=3$. Then the vorticity $\omega$ has the
functional integral representation given by
\begin{align*}
\omega(\xi,t)  =&\phi\left(\frac{\xi_{3}}{\varepsilon}\right)\sigma(\xi_{1},\xi_{2},t)+\int_{\mathbb{R}_{+}^{3}}\mathbb{E}\left[\left.1_{\{t<\zeta(X^{\eta})\}}Q(\eta,t;0)\omega^{\varepsilon}(\eta,0)\right|X_{t}^{\eta}=\xi\right]p_{\nu}(0,\eta,t,\xi)\rmd\eta\nonumber \\
 & +\int_{0}^{t}\int_{\mathbb{R}_{+}^{3}}\mathbb{E}\left[\left.1_{\left\{ s>\gamma_{t}(X^{\eta})\right\} }Q(\eta,t;s)F(X_{s}^{\eta},s)\right|X_{t}^{\eta}=\xi\right]p_{\nu}(0,\eta,t,\xi)\rmd\eta\rmd s\nonumber \\
 & +\int_{0}^{t}\int_{\mathbb{R}_{+}^{3}}\mathbb{E}\left[\left.1_{\left\{ s>\gamma_{t}(X^{\eta})\right\} }Q(\eta,t;s)\chi_{\varepsilon}(X_{s}^{\eta},s)\right|X_{t}^{\eta}=\xi\right]p_{\nu}(0,\eta,t,\xi)\rmd \eta\rmd s
\end{align*}
for every $\xi=(\xi_{1},\xi_{2},\xi_{3})\in\mathbb{R}_{+}^{3}$ and
$t>0$, where $(Q\omega^\varepsilon)^{i}=Q_{j}^{i}(\omega^\varepsilon)^{j}$ and $(QF)^{i}=Q_{j}^{i}F^{j}$.
\end{prop}

These functional integral representations are consequences of the vorticity
transport equation (\ref{eq:ome-01}) and Theorem \ref{thm7.2-1}
immediately.

\subsection{Representations of the velocity and the temperature}
To avoid repetition, unlike in the unbounded domain, we state the representations of velocity and temperature for $d=2$ and $d=3$ together and establish random vortex systems in the next section separately.

Combining with the Biot-Savart laws, we may establish various integral
functional representations for the velocity and the temperature.

Firstly, using Proposition \ref{prop53} and the Biot-Savart law (cf. Lemma
\ref{lem212}), we may deduce the following.
\begin{thm}
 For the temperature $\theta$, the following functional
integral representation holds:
\begin{enumerate}[{(1)}]
\item When $d=2$,
\begin{align}
\theta(x,t) =& \frac{x_{2}}{\pi}\int_{-\infty}^{\infty}\frac{\theta_{0}(\xi_{1})}{|\xi_{1}-x_{1}|^{2}+x_{2}^{2}}\rmd\xi_{1}\nonumber \\
 & -\int_{0}^{\varepsilon}\left[\phi\left(\frac{\xi_{2}}{\varepsilon}\right)\int_{\mathbb{R}}\varLambda_{2}(\xi,x)\cdot\varTheta((\xi_{1},0),t)\rmd\xi_{1}\right]\rmd \xi_{2}\nonumber \\
 & -\int_{\mathbb{R}_{+}^{2}}\mathbb{E}\left[1_{\{t<\zeta(Y^{\eta})\}}1_{\mathbb{R}_{+}^{2}}(Y_{t}^{\eta})\varLambda_{2}(Y_{t}^{\eta},x)\cdot R(\eta,t;0)\varTheta^{\varepsilon}(\eta,0)\right]\rmd\eta\nonumber \\
 & -\int_{0}^{t}\int_{\mathbb{R}_{+}^{2}}\mathbb{E}\left[1_{\left\{ s>\gamma_{t}(Y^{\eta})\right\} }1_{\mathbb{R}_{+}^{2}}(Y_{t}^{\eta})\varLambda_{2}(Y_{t}^{\eta},x)\cdot R(\eta,t;s)\alpha_{\varepsilon}\left(Y_{s}^{\eta},s\right)\right]\rmd\eta\rmd s\label{2D-t-211}
\end{align}
for $x\in \R^2_+$ and $t>0$, where $(R\varTheta^\varepsilon)_{j}=R_{j}^{l}\varTheta^\varepsilon_{l}$.
\item When $d=3$,
\begin{align}
\theta(x,t) = &\frac{x_{3}}{2\pi}\int_{\mathbb{R}^{2}}\frac{\theta_{0}(\xi_{1},\xi_{2})}{|\xi_{1}-x_{2}|^{2}+|\xi_{2}-x_{2}|^{2}+|x_{3}|^{2}}\rmd\xi_{1}\rmd\xi_{2}\nonumber \\
 & -\int_{0}^{\varepsilon}\left[\phi\left(\frac{\xi_{3}}{\varepsilon}\right)\int_{\mathbb{R}^{2}}\varLambda_{3}(\xi,x)\cdot\varTheta(\xi_{1},\xi_{2},0,t)\rmd\xi_{1}\rmd\xi_{2}\right]\rmd\xi_{3}\nonumber \\
 & -\int_{\mathbb{R}_{+}^{3}}\mathbb{E}\left[1_{\{t<\zeta(Y^{\eta})\}}1_{\mathbb{R}_{+}^{3}}(Y_{t}^{\eta})\varLambda_{3}(Y_{t}^{\eta},x)\cdot R(\eta,t;0)\varTheta^{\varepsilon}(\eta,0)\right]\rmd\eta\nonumber \\
 & -\int_{0}^{t}\int_{\mathbb{R}_{+}^{3}}\mathbb{E}\left[1_{\left\{ s>\gamma_{t}(Y^{\eta})\right\} }1_{\mathbb{R}_{+}^{3}}(Y_{t}^{\eta})\varLambda_{3}(Y_{t}^{\eta},x)\cdot R(\eta,t;s)\alpha_{\varepsilon}\left(Y_{s}^{\eta},s\right)\right]\rmd\eta\rmd s\label{3D-t-211}
\end{align}
for $x\in \R^3_+$ and $t>0$, where $(R\varTheta^\varepsilon)_{j}=R_{j}^{l}\varTheta^\varepsilon_{l}$. 

\end{enumerate}
\end{thm}

Similarly, by using the Biot-Savart law (cf. Lemma \ref{lem211}),
Proposition \ref{prop54} and Proposition \ref{prop55}, we may establish the
following functional integral representations.
\begin{thm}
Suppose $d=2$. Then for every $\varepsilon>0$
\begin{align}
u(x,t)  =&\int_{0}^{\varepsilon}\left[\phi\left(\frac{\xi_{2}}{\varepsilon}\right)\int_{-\infty}^{\infty}\varLambda_{2}(\xi,x)\wedge\sigma(\xi_{1},t)\rmd \xi_1\right]\rmd \xi_{2}\nonumber \\
 & +\int_{\mathbb{R}_{+}^{2}}\mathbb{E}\left[1_{\{t<\zeta(X^{\eta})\}}1_{\mathbb{R}_{+}^{2}}(X_{t}^{\eta})\varLambda_{2}(X_{t}^{\eta},x)\wedge\omega^{\varepsilon}(\eta,0)\right]\rmd \eta\nonumber \\
 & +\int_{0}^{t}\int_{\mathbb{R}_{+}^{2}}\mathbb{E}\left[1_{\left\{ s>\gamma_{t}(X^{\eta})\right\} }1_{\mathbb{R}_{+}^{2}}(X_{t}^{\eta})\varLambda_{2}(X_{t}^{\eta},x)\wedge F(X_{s}^{\eta},s)\right]\rmd\eta\rmd s\nonumber \\
 & +\int_{0}^{t}\int_{\mathbb{R}_{+}^{2}}\mathbb{E}\left[1_{\left\{ s>\gamma_{t}(X^{\eta})\right\} }1_{\mathbb{R}_{+}^{2}}(X_{t}^{\eta})\varLambda_{2}(X_{t}^{\eta},x)\wedge\chi_{\varepsilon}(X_{s}^{\eta},s)\right]\rmd \eta\rmd s\label{2Du01}
\end{align}
for $x\in\mathbb{R}_{+}^{2}$ and $t>0$, and $u(x,t)=\mathscr{R}(u(\overline{x},t))$
if $x_{2}<0$.
\end{thm}

In dimension 3, the representation of the velocity is much more
complicated. 
\begin{thm}
Suppose $d=3$. For every $\varepsilon>0$, the velocity $u(x,t)$ has the following
functional integral representation: 
\begin{align}
u(x,t)  =&\int_{0}^{\varepsilon}\left(\phi\left(\frac{\xi_{3}}{\varepsilon}\right)\int_{\mathbb{R}^{2}}\varLambda_{3}(\xi,x)\wedge\sigma(\xi_{1},\xi_{2},t)\rmd \xi_{1}\rmd \xi_{2}\right)\rmd \xi_{3}\nonumber \\
 & +\int_{\mathbb{R}_{+}^{3}}\mathbb{E}\left[1_{\{t<\zeta(X^{\eta})\}}1_{\mathbb{R}_{+}^{3}}(X_{t}^{\eta})\varLambda_{3}(X_{t}^{\eta},x)\wedge Q(\eta,t;0)\omega^{\varepsilon}(\eta,0)\right]\rmd \eta\nonumber \\
 & +\int_{0}^{t}\int_{\mathbb{R}_{+}^{3}}\mathbb{E}\left[1_{\left\{ s>\gamma_{t}(X^{\eta})\right\} }1_{\mathbb{R}_{+}^{3}}(X_{t}^{\eta})\varLambda_{3}(X_{t}^{\eta},x)\wedge Q(\eta,t;s)F(X_{s}^{\eta},s)\right]\rmd \eta\rmd s\nonumber \\
 & +\int_{0}^{t}\int_{\mathbb{R}_{+}^{3}}\mathbb{E}\left[1_{\left\{ s>\gamma_{t}(X^{\eta})\right\} }1_{\mathbb{R}_{+}^{3}}(X_{t}^{\eta})\varLambda_{3}(X_{t}^{\eta},x)\wedge Q(\eta,t;s)\chi_{\varepsilon}(X_{s}^{\eta},s)\right]\rmd\eta\rmd s.\label{3D-ua1}
\end{align}
for $x\in\mathbb{R}_{+}^{3}$ and $t>0$, and $u(x,t)=\mathscr{R}(u(\overline{x},t))$
if $x_{3}<0$, where the gauge functional $Q$ is defined by (\ref{Ga-01})
with $D=\mathbb{R}_{+}^{3}$, $(Q\omega^\varepsilon)^{i}=Q_{j}^{i}(\omega^\varepsilon)^{j}$ and $(QF)^{i}=Q_{j}^{i}F^{j}$.
\end{thm}

\subsection{Random vortex dynamics}

With the stochastic representation formulae derived in the previous subsections, we are ready to establish a system of random vortex dynamics for wall-bounded
flows. 

\subsubsection{\label{subsec651}2D random vortex for wall-bounded flows}

In this part, let us work out the simpler scenario, the two-dimensional case.
When $d=2$, the interaction force due to the heat conduction is given by 
\[
f_{1}(x,t)=0,\quad\textrm{ and }f_{2}(x,t)=\theta(x,t)-\theta_{0}(x_{1})
\]
(the constant factor $\alpha g$ is absorbed into the temperature $\theta$
for simplicity), so that
\begin{equation}
F(x,t)=\varTheta_{1}(x,t)-\frac{\partial}{\partial x_{1}}\theta_{0}(x_{1}).\label{2D-F01}
\end{equation}

The random vortex dynamics may be defined as the following:
\begin{itemize}
    \item Eq. (\ref{X-P-01}) for $X$ particles and Eq. (\ref{Y-P-01}) for
    $Y$ particles, where in both equations, the velocity is extended by reflection
    so that $u(x,t)=0$ when $x_{2}=0$, and $u^{1}(x,t)=u^{1}(\overline{x},t)$,
    $u^{2}(x,t)=-u^{2}(\overline{x},t)$. 
    \item Eq. (\ref{Ga-02}) with
    $D=\mathbb{R}_{+}^{2}$ that defines the gauge functional $R$.
    \item The
    representation \eqref{2Du01} for the velocity and the representation
    (\ref{2D-t-211}) for the temperature $\theta(x,t)$.
    \item Other dynamical
    variables given by
    \begin{equation}
    A_{j}^{l}(x,t)=\frac{\partial}{\partial x_{j}}u^{l}(x,t),\quad F(x,t)=\frac{\partial}{\partial x_{1}}\theta(x,t)-\frac{\partial}{\partial x_{1}}\theta_{0}(x_{1}),\label{2D-others01}
    \end{equation}
    which lead to the stochastic representations for $A_{j}^{l}$ and $F$, obtained by differentiating the functional integral representations
    for $u^{l}$ and $\theta$ respectively. 
\end{itemize}

During numerical experiments, we may choose $\varepsilon>0$ sufficiently small, so that the contribution from 
\[
-\int_{0}^{\varepsilon}\left[\phi\left(\frac{\xi_{2}}{\varepsilon}\right)\int_{-\infty}^\infty\varLambda_{2}(\xi,x)\cdot\varTheta((\xi_{1},0),t)\rmd \xi_{1}\right]\rmd \xi_{2}
\]
is negligible. Also, by the definition of the Oberbeck-Boussinesq
model, the temperature gradient at the boundary is small, so it is reasonable to neglect
the contribution from 
\[
-\int_{0}^{t}\int_{\mathbb{R}_{+}^{2}}\mathbb{E}\left[1_{\left\{ s>\gamma_{t}(Y^{\eta})\right\} }1_{\mathbb{R}_{+}^{2}}(Y_{t}^{\eta})\varLambda_{2}(Y_{t}^{\eta},x)\cdot R(\eta,t;s)\alpha_{\varepsilon}\left(Y_{s}^{\eta},s\right)\right]\rmd \eta\rmd s.
\]
 Therefore, by sending $\varepsilon\downarrow0$,
we may calculate the temperature via the formula:
\begin{align}
\theta(x,t) \sim& \frac{x_{2}}{\pi}\int_{-\infty}^{\infty}\frac{\theta_{0}(\xi_{1})}{|\xi_{1}-x_{1}|^{2}+x_{2}^{2}}\rmd \xi_{1}\nonumber \\
 & -\int_{\mathbb{R}_{+}^{2}}\mathbb{E}\left[1_{\{t<\zeta(Y^{\eta})\}}1_{\mathbb{R}_{+}^{2}}(Y_{t}^{\eta})\varLambda_{2}(Y_{t}^{\eta},x)\cdot R(\eta,t;0)\varTheta(\eta,0)\right]\rmd\eta.\label{t-app-01}
\end{align}
Similarly, for the velocity $u(x,t)$, the contribution, when $\varepsilon>0$
is small enough, we may ignore the minor contribution from 
\[
\int_{0}^{\varepsilon}\left[\phi\left(\frac{\xi_{2}}{\varepsilon}\right)\int_{-\infty}^{\infty}\varLambda_{2}(\xi,x)\wedge\sigma(\xi_{1},t)\rmd \xi\right]\rmd\xi_{2}.
\]
Furthermore, by sending $\varepsilon\downarrow0$, we obtained the limiting representation for $u(x,t)$:
\begin{align}
u(x,t)  =&\int_{\mathbb{R}_{+}^{2}}\mathbb{E}\left[1_{\{t<\zeta(X^{\eta})\}}1_{\mathbb{R}_{+}^{2}}(X_{t}^{\eta})\varLambda_{2}(X_{t}^{\eta},x)\wedge\omega(\eta,0)\right]\rmd \eta\nonumber \\
 & +\int_{0}^{t}\int_{\mathbb{R}_{+}^{2}}\mathbb{E}\left[1_{\left\{ s>\gamma_{t}(X^{\eta})\right\} }1_{\mathbb{R}_{+}^{2}}(X_{t}^{\eta})\varLambda_{2}(X_{t}^{\eta},x)\wedge F(X_{s}^{\eta},s)\right]\rmd \eta\rmd s\nonumber \\
 & +\lim_{\varepsilon\downarrow0}\int_{0}^{t}\int_{\mathbb{R}_{+}^{2}}\mathbb{E}\left[1_{\left\{ s>\gamma_{t}(X^{\eta})\right\} }1_{\mathbb{R}_{+}^{2}}(X_{t}^{\eta})\varLambda_{2}(X_{t}^{\eta},x)\wedge\chi_{\varepsilon}(X_{s}^{\eta},s)\right]\rmd \eta\rmd s,\label{u-app-001}
\end{align}
where $\chi_{\varepsilon}$ can be replaced by
\begin{equation}
\tilde{\chi}_{\varepsilon}(x,t)=\nu\frac{1}{\varepsilon^{2}}\phi''\left(\frac{x_{2}}{\varepsilon}\right)\sigma(x_{1},t),\label{error-app-001}
\end{equation}
and $\sigma$ is the trace of $\omega$ on the boundary $\{x_2=0\}$.

We will justify this procedure in Subsection \ref{subsec66} below
by showing that the limit of each term on the right-hand side of (\ref{u-app-001})
exists.

\subsubsection{3D random vortex for wall-bounded flows}

The computation is analogous to the previous part, except we need to include the gauge functional
$Q$ now. 

Hence the random vortex dynamics in dimension three can be characterised by the following system of equations:
\begin{itemize}
    \item The equation for $X$ particle by (\ref{X-P-01}), the equation for $Y$ particle by (\ref{Y-P-01}),
    where in both equations, the velocity is extended by reflection so that $u(x,t)=0$
    when $x_{3}=0$, and $u^{1}(x,t)=u^{1}(\overline{x},t)$, $u^{1}(x,t)=u^{1}(\overline{x},t)$,
    $u^{3}(x,t)=-u^{3}(\overline{x},t)$.
    \item  The equation that defines the
    gauge functional $Q$ and $R$, (\ref{Ga-01}) and (\ref{Ga-02})
    with $D=\mathbb{R}_{+}^{3}$.
    \item The representation \eqref{3D-ua1} for the velocity
    and the representation \eqref{3D-t-211} for the temperature
    $\theta(x,t)$.
    \item  Other dynamical variables given by
    \begin{align*}    
    A_{j}^{l}(x,t) &=\frac{\partial}{\partial x_{j}}u^{l}(x,t), 
    &F^{1}(x,t)&=\frac{\partial}{\partial x_{2}}\theta(x,t)-\frac{\partial}{\partial x_{2}}\theta_{0}(x_{1},x_{2}),\\
    F^{2}(x,t)&=-\frac{\partial}{\partial x_{1}}\theta(x,t)+\frac{\partial}{\partial x_{1}}\theta_{0}(x_{1},x_{2}),
    &F^{3}(x,t)&=0,
    \end{align*}
    which give rise to the representations for $A_{j}^{l}$ and $F$ by differentiating the functional integral representations for $u^{l}$
    and $\theta$. 
\end{itemize}


\subsection{\label{subsec66}Limiting representations}

In the previous parts, we have used the families of perturbations
of the vorticity, the temperature and the temperature gradient to
establish the random vortex dynamical systems. For the purpose of numerical
simulations, we need to choose the parameter $\varepsilon>0$ small
enough so that the outer layer velocity of the fluid flow can be avoided.
In this part, we demonstrate that the limiting representations as $\varepsilon\downarrow0$
 exist, which will justify the approximations \eqref{t-app-01} and \eqref{u-app-001}. 

Let us consider the two-dimensional case. Recall that the velocity
$u(x,t)$ of the fluid flow is extended for all $x\in\mathbb{R}^{2}$,
so that $\mathscr{R}(u(x,t))=u(\overline{x},t)$. Hence 
\[
p_{\lambda,u}^{D}(s,\xi,t,x)=p_{\lambda,u}(\tau,\xi,t,x)-p_{\lambda,u}(\tau,\overline{\xi},t,x)
\]
is the Green function to the Dirichlet problem of the forward parabolic
operator $L_{\lambda,-u}-\frac{\partial}{\partial t}$, where $D=\mathbb{R}_{+}^{2}$. In the computation below, we drop $u$ in the notation of transition probabilities for simplicity. 
By using the vorticity transport equation (\ref{eq:ome-01})
for $\omega^{\varepsilon}(x,t)$ and  the temperature transport
equation (\ref{eq:Tem-e-01}) for $\theta^{\varepsilon}(x,t)$, according to \citep[Theorem 12 on page 25]{Friedman 1964}, we obtain that
\begin{align}
\omega^{\varepsilon}(x,t)  =&\int_{D}p_{\nu}^{D}(0,\eta,t,x)\omega^{\varepsilon}(\eta,0)\rmd\eta+\int_{0}^{t}\int_{D}p_{\nu}^{D}(s,\eta,t,x)F(\eta,s)\rmd \eta\rmd s\nonumber \\
 & +\int_{0}^{t}\int_{D}p_{\nu}^{D}(s,\eta,t,x)\chi_{\varepsilon}(\eta,s)\rmd \eta\rmd s\label{oe-81}
\end{align}
and
\begin{equation}
\theta^{\varepsilon}(x,t)=\int_{D}p_{\kappa}^{D}(0,\eta,t,x)\theta^{\varepsilon}(\eta,0)\rmd \eta+\int_{0}^{t}\int_{D}p_{\kappa}^{D}(s,\eta,t,x)\beta_{\varepsilon}(\eta,s)\rmd \eta\rmd s\label{eq:t-rep-01}
\end{equation}
for every $x\in D$ and $t>0$. 
Before we present the proof, let us first introduce a new notation and a helpful result that apply to all dimensions. 
\begin{defn}
If $\varphi$ is a Borel measurable function on $\mathbb{R}_{+}^{d}$,
then define 
\begin{equation}
\hat{\varphi}(x)=1_{\{x_{d}>0\}}\varphi(x)-1_{\{x_{d}<0\}}\varphi(\overline{x})\label{eq:def-hat}
\end{equation}
for every $x\in\mathbb{R}^{d}$.
\end{defn}

\begin{lem}
Let $D=\mathbb{R}_{+}^{d}$ and $\lambda>0$. Then
it holds that
\[
\int_{D}p_{\lambda}^{D}(s,\eta,t,x)\varphi(\eta)\rmd\eta=\int_{\mathbb{R}^{d}}p_{\lambda,u}(s,\eta,t,x)\hat{\varphi}(\eta)\rmd\eta
\]
for every $x\in D$ and $t>s\geq0$.
\end{lem}

\begin{proof}
Since $\nabla\cdot u=0$ on $\mathbb{R}^{d}$, by Eq. (\ref{eq:ST-pdf-01}),
\begin{align*}
\int_{D}p_{\lambda}^{D}(s,\eta,t,x)\varphi(\eta)\rmd\eta =&\int_{\mathbb{R}^{d}}p_{\lambda}(s,\eta,t,x)1_{\{\eta_{d}>0\}}\varphi(\eta)\rmd\eta
  -\int_{\mathbb{R}^{d}}p_{\lambda}(s,\overline{\eta},t,x)1_{\{\eta_{d}>0\}}\varphi(\eta)\rmd\eta\\
 = &\int_{\mathbb{R}^{d}}p_{\lambda}(s,\eta,t,x)1_{\{\eta_{d}>0\}}\varphi(\eta)\rmd\eta
  -\int_{\mathbb{R}^{d}}p_{\lambda}(s,\eta,t,x)1_{\{\eta_{d}<0\}}\varphi(\overline{\eta})\rmd\eta,
\end{align*}
which completes the proof.
\end{proof}
Therefore we may rewrite the representations (\ref{oe-81}) and (\ref{eq:t-rep-01})
and obtain the following lemma.
\begin{lem}
For every $\varepsilon>0$ we have
\begin{align}
\omega^{\varepsilon}(x,t)  =&\int_{\mathbb{R}^{2}}p_{\nu}(0,\eta,t,x)\hat{\omega^{\varepsilon}}(\eta,0)\rmd\eta+\int_{0}^{t}\int_{\mathbb{R}^{2}}p_{\nu}(s,\eta,t,x)\hat{F}(\eta,s)\rmd\eta\rmd s\nonumber \\
 & +\int_{0}^{t}\int_{\mathbb{R}^{2}}p_{\nu}(s,\eta,t,x)\hat{\chi_{\varepsilon}}(\eta,s)\rmd\xi\rmd s\label{oe-81-1}
\end{align}
and
\[
\theta^{\varepsilon}(x,t)=\int_{\mathbb{R}^{2}}p_{\kappa}(0,\eta,t,x)\hat{\theta^{\varepsilon}}(\eta,0)\rmd\eta+\int_{0}^{t}\int_{\mathbb{R}^{2}}p_{\kappa}(s,\eta,t,x)\hat{\beta_{\varepsilon}}(\eta,s)\rmd\eta\rmd s
\]
for every $x\in\mathbb{R}_{+}^{2}$ and $t>0$. 
\end{lem}

In order to avoid calculating the outer layer velocity, which appears
in the error terms $\hat{\chi}_{\varepsilon}(\xi,s)$ and $\hat{\beta}_{\varepsilon}(\xi,s)$,
we consider their limiting representations as $\varepsilon\downarrow0$.
We may apply the same technique used in \citep{QQZW2022}. 
\begin{prop}
 The following integral representations hold:
\begin{align}
\omega(x,t)  =&1_{\{x_{2}=0\}}\sigma(x_{1},t)+\int_{\mathbb{R}^{2}}p_{\nu}(0,\eta,t,x)\hat{\omega}(\eta,0)\rmd\eta\nonumber \\
 & +\int_{0}^{t}\int_{\mathbb{R}^{2}}p_{\nu}(s,\eta,t,x)\hat{F}(\eta,s)\rmd\eta\rmd s\nonumber \\
 & +2\nu\int_{0}^{t}\int_{-\infty}^{\infty}\sigma(\eta_{1},s)\left.\frac{\partial}{\partial\eta_{2}}\right|_{\eta_{2}=0}p_{\nu}(s,(\eta_{1},\eta_{2}),t,x)\rmd\eta_{1}\rmd s\label{2D-0-N1}
\end{align}
and
\begin{align}
\theta(x,t) = &1_{\{x_{2}=0\}}\theta_{0}(x_{1})+\int_{\mathbb{R}^{2}}p_{\kappa}(0,\eta,t,x)\hat{\theta}(\eta,0)\rmd\eta\nonumber \\
 & +2\kappa\int_{-\infty}^{\infty}\theta_{0}(\eta_{1})\left.\frac{\partial}{\partial\eta_{2}}\right|_{\eta_{2}=0}\left(\int_{0}^{t}p_{\kappa}(s,\eta,t,x)\rmd s\right)\rmd\eta_{1}\label{2D-t-N1}
\end{align}
for $x\in\mathbb{R}_{+}^{2}$.
\end{prop}

\begin{proof}
Let us prove (\ref{2D-0-N1}) in detail, and the argument for $\theta^{\varepsilon}$
is similar. Since $\lim_{\varepsilon\rightarrow0+}\omega^{\varepsilon}(x,t)=\omega(x,t)$
for $x_{2}>0$, it follows that
\[
\lim_{\varepsilon\rightarrow0+}\int_{\mathbb{R}^{2}}p_{\nu}(0,\eta,t,x)\hat{\omega^{\varepsilon}}(\eta,0)\rmd\eta=\int_{\mathbb{R}^{2}}p_{\nu}(0,\eta,t,x)\hat{\omega}(\eta,0)\rmd \eta.
\]
Hence we only need to deal with the limit of the error term
\[
J(x,\varepsilon)=\int_{0}^{t}\int_{\mathbb{R}^{2}}p_{\nu}(s,\xi,t,x)\hat{\chi}_{\varepsilon}(\xi,s)\rmd\xi\rmd s
\]
appearing on the right-hand side of (\ref{oe-81-1}), where $\chi_{\varepsilon}$ is given in (\ref{eq:oe-er-01}). According
to (\ref{eq:oe-lim0p}), there is no contribution as $\varepsilon\downarrow0$
towards the limit of $J(x,\varepsilon)$ from the first term in $\chi_{\varepsilon}$.
Therefore we only need to calculate the contributions from the error
terms
\[
E_{1}^{\varepsilon}(x,t)=-\frac{1}{\varepsilon}\phi'\left(\frac{x_{2}}{\varepsilon}\right)\sigma(x_{1},t)u^{2}(x,t)
\]
and
\[
E_{2}^{\varepsilon}(x,t)=\nu\frac{1}{\varepsilon^{2}}\phi''\left(\frac{x_{2}}{\varepsilon}\right)\sigma(x_{1},t).
\]
To this end, we may choose a concrete cut-off function $\phi$. Let us set
\begin{equation}
\phi(r)=\begin{cases}
1 & \textrm{ for \ensuremath{r\in[0,\frac{1}{3})},}\\
\frac{1}{2}+54\left(r-\frac{1}{2}\right)^{3}-\frac{9}{2}\left(r-\frac{1}{2}\right) & \textrm{ for }r\in[\frac{1}{3},\frac{2}{3}],\\
0 & \textrm{ for }r\geq \frac{2}{3}
\end{cases}\label{phi-def}
\end{equation}
Then $-54\leq\phi''\leq54$, $-\frac{9}{2}\leq\phi'\leq0$ on $[\frac{1}{3},\frac{2}{3}]$
and $\phi'=0$ for $r\leq\frac{1}{3}$ or $r\geq
\frac{2}{3}$. In fact, 
\[
\phi'(r)=\begin{cases}
162\left(r-\frac{1}{2}\right)^{2}-\frac{9}{2}, & \textrm{ for }r\in[\frac{1}{3},\frac{2}{3}],\\
0, & \textrm{ otherwise. }
\end{cases}
\]
and 
\[
\phi''(r)=\begin{cases}
324\left(r-\frac{1}{2}\right), & \textrm{ for }r\in[\frac{1}{3},\frac{2}{3}],\\
0, & \textrm{ otherwise. }
\end{cases}
\]
Let
\[
J_{i}(x,\varepsilon)=\int_{0}^{t}\int_{\mathbb{R}^{2}}p_{\nu}(s,\xi,t,x)\hat{E}_{i}^{\varepsilon}(\xi,s)\rmd \xi\rmd s
\]
where $i=1,2$. Then 
\begin{align*}
J_{i}(x,\varepsilon) & =\int_{0}^{t}\int_{-\infty}^{\infty}\left[\int_{0}^{\infty}p_{\nu}\left(s,(\xi_{1},\xi_{2}),t,x\right)E_{i}^{\varepsilon}(\xi,s)\rmd\xi_{2}\right]\rmd \xi_{1}\rmd s\\
 & -\int_{0}^{t}\int_{-\infty}^{\infty}\left[\int_{0}^{\infty}p_{\nu}\left(s,(\xi_{1},-\xi_{2}),t,x\right)E_{i}^{\varepsilon}(\xi,s)\rmd\xi_{2}\right]\rmd \xi_{1}\rmd s.
\end{align*}
Let us consider the integral
\begin{align*}
I_{1}(\varepsilon) & :=\int_{0}^{\infty}p_{\nu}\left(s,(\xi_{1},\xi_{2}),t,x\right)E_{1}^{\varepsilon}(\xi,s)\rmd\xi_{2}\\
 & =-\sigma(\xi_{1},s)\int_{0}^{\varepsilon}u^{2}(\xi,s)p_{\nu}\left(s,(\xi_{1},\xi_{2}),t,x\right)\frac{1}{\varepsilon}\phi'\left(\frac{\xi_{2}}{\varepsilon}\right)\rmd\xi_{2}\\
 & =\sigma(\xi_{1},s)\int_{0}^{\varepsilon}\phi\left(\frac{\xi_{2}}{\varepsilon}\right)\frac{\partial}{\partial\xi_{2}}\left(u^{2}(\xi,t)p_{\nu}\left(s,(\xi_{1},\xi_{2}),t,x\right)\right)\rmd \xi_{2}\\
 & \rightarrow0\quad\textrm{ as }\varepsilon\downarrow0,
\end{align*}
where for the third equality, we have used the fact that $u(x,t)$
vanishes when $x_{2}=0$. Hence
\[
\lim_{\varepsilon\rightarrow0+}\int_{0}^{t}\int_{\mathbb{R}^{2}}p_{\nu}(s,\xi,t,x)\hat{E}_{1}^{\varepsilon}(\xi,s)\rmd \xi\rmd s=0.
\]
Now we consider $J_{2}(x,\varepsilon)$. To this end, we observe that
\[
J_{2}(x,\varepsilon)=\int_{0}^{t}\int_{-\infty}^{\infty}\left[\int_{0}^{\infty}\left(p_{\nu}\left(s,(\xi_{1},\xi_{2}),t,x\right)-p_{\nu}\left(s,(\xi_{1},-\xi_{2}),t,x\right)\right)E_{2}^{\varepsilon}(\xi,s)\rmd \xi_{2}\right]\rmd \xi_{1}\rmd s
\]
and integrate by parts twice to deduce that
\begin{align*}
I_{2}(\varepsilon)  :=&\int_{0}^{\infty}\left(p_{\nu}\left(s,(\xi_{1},\xi_{2}),t,x\right)-p_{\nu}\left(s,(\xi_{1},-\xi_{2}),t,x\right)\right)E_{2}^{\varepsilon}(\xi,s)\rmd\xi_{2}\\
  =&\nu\sigma(\xi_{1},t)\int_{0}^{\varepsilon}\left(p_{\nu}\left(s,(\xi_{1},\xi_{2}),t,x\right)-p_{\nu}\left(s,(\xi_{1},-\xi_{2}),t,x\right)\right)\frac{1}{\varepsilon^{2}}\phi''\left(\frac{\xi_{2}}{\varepsilon}\right)\rmd\xi_{2}\\
  =&-\nu\sigma(\xi_{1},t)\int_{0}^{\varepsilon}\frac{\partial}{\partial\xi_{2}}\left(p_{\nu}\left(s,(\xi_{1},\xi_{2}),t,x\right)-p_{\nu}\left(s,(\xi_{1},-\xi_{2}),t,x\right)\right)\frac{1}{\varepsilon}\phi'\left(\frac{\xi_{2}}{\varepsilon}\right)\rmd\xi_{2}\\
  =&2\nu\sigma(\xi_{1},t)\left.\frac{\partial}{\partial\xi_{2}}\right|_{\xi_{2}=0}p_{\nu}\left(s,(\xi_{1},\xi_{2}),t,x\right)\\
 & +\nu\sigma(\xi_{1},t)\int_{0}^{\varepsilon}\phi\left(\frac{\xi_{2}}{\varepsilon}\right)\frac{\partial^{2}}{\partial\xi_{2}^2}\left(p_{\nu}\left(s,(\xi_{1},\xi_{2}),t,x\right)-p_{\nu}\left(s,(\xi_{1},-\xi_{2}),t,x\right)\right)\rmd\xi_{2}.
\end{align*}
As a consequence, we have
\[
I_{2}(\varepsilon)\rightarrow2\nu\sigma(\xi_{1},t)\left.\frac{\partial}{\partial\xi_{2}}\right|_{\xi_{2}=0}p_{\nu}\left(s,(\xi_{1},\xi_{2}),t,x\right)\quad\textrm{ as }\varepsilon\downarrow0
\]
and therefore
\[
\lim_{\varepsilon\rightarrow0+}J_{2}(x,\varepsilon)=2\nu\int_{0}^{t}\int_{-\infty}^{\infty}\sigma(\xi_{1},t)\left.\frac{\partial}{\partial\xi_{2}}\right|_{\xi_{2}=0}p_{\nu}\left(s,(\xi_{1},\xi_{2}),t,x\right)\rmd\xi_{1}\rmd s
\]
which completes the proof.
\end{proof}
By using the Biot-Savart law, we may then deduce the following functional
integral representation theorem. 
\begin{thm}
Let $X$ and $Y$ be defined as in \eqref{eq:T-01-1} and \eqref{eq:T-01-2} respectively. The following functional integral representations hold:
\begin{align*}
u(x,t)  =&\int_{\mathbb{R}^{2}}\mathbb{E}\left[1_{\mathbb{R}_{+}^{2}}(X_{t}^{\eta})\varLambda_{2}(X_{t}^{\eta},x)\wedge\hat{\omega}(\eta,0)\right]\rmd\eta\\
 & +\int_{0}^{t}\int_{\mathbb{R}^{2}}\mathbb{E}\left[1_{\mathbb{R}_{+}^{2}}(X_{t}^{\eta,s})\varLambda_{2}(X_{t}^{\eta,s},x)\wedge\hat{F}(\eta,s)\right]\rmd\eta\rmd s\\
 & +2\nu\int_{0}^{t}\int_{-\infty}^{\infty}\left.\frac{\partial}{\partial\eta_{2}}\right|_{\eta_{2}=0}\mathbb{E}\left[1_{\mathbb{R}_{+}^{2}}(X_{t}^{\eta,s})\varLambda_{2}(X_{t}^{\eta,s},x)\wedge\sigma(\eta,s)\right]\rmd \eta_{1}\rmd s
\end{align*}
and 
\begin{align*}
u(x,t) = &I_{2,1}(x,t)+\int_{\mathbb{R}^{2}}\mathbb{E}\left[1_{\mathbb{R}_{+}^{2}}(X_{t}^{\eta})\varLambda_{2}(X_{t}^{\eta},x)\wedge\hat{\omega}(\eta,0)\right]\rmd \eta\nonumber \\
 & -\int_{0}^{t}\int_{\mathbb{R}^{2}}\mathbb{E}\left[1_{\mathbb{R}_{+}^{2}}(X_{t}^{\eta,s})\varLambda_{2}(X_{t}^{\eta,s},x)\wedge\widehat{\frac{\partial\theta_{0}}{\partial\eta_{1}}}(\eta_{1},s)\right]\rmd \eta\rmd s\nonumber \\
 & +2\nu\int_{0}^{t}\int_{-\infty}^{\infty}\left.\frac{\partial}{\partial\eta_{2}}\right|_{\eta_{2}=0}\mathbb{E}\left[1_{\mathbb{R}_{+}^{2}}(X_{t}^{\eta,s})\varLambda_{2}(X_{t}^{\eta,s},x)\wedge\sigma(\eta,s)\right]\rmd\eta_{1}\rmd s
\end{align*}
for every $x\in\mathbb{R}_{+}^{2}$ and $t>0$, where
\begin{align*}
I_{2,1}(x,t)  =&-\int_{0}^{t}\int_{\mathbb{R}^{2}}\mathbb{E}\left[1_{\mathbb{R}_{+}^{2}}\left(Y_{s}^{\xi}\right)H\left(Y_{s}^{\xi},s;t,x\right)\wedge\hat{\theta}(\xi,0)\right]\rmd \xi\rmd s\\
 & +\int_{0}^{t}\int_{\mathbb{R}^{2}}\mathbb{E}\left[1_{\mathbb{R}_{+}^{2}}\left(\overline{Y_{s}^{\overline{\xi}}}\right)H\left(Y_{s}^{\overline{\xi}},s;t,x\right)\wedge\hat{\theta}(\xi,0)\right]\rmd\xi\rmd s\\
 & -2\kappa\int_{0}^{t}\int_{-\infty}^{\infty}\left.\frac{\partial}{\partial\xi_{2}}\right|_{\xi_{2}=0}\left(\int_{0}^{s}\mathbb{E}\left[1_{\mathbb{R}_{+}^{2}}\left(Y_{s}^{\xi,\tau}\right)H\left(Y_{s}^{\xi,\tau},s;t,x\right)\wedge\theta_{0}(\xi_{1})\right]\rmd\tau\right)\rmd\xi_{1}\rmd s\\
 & +2\kappa\int_{0}^{t}\int_{-\infty}^{\infty}\left.\frac{\partial}{\partial\xi_{2}}\right|_{\xi_{2}=0}\left(\int_{0}^{s}\mathbb{E}\left[1_{\mathbb{R}_{+}^{2}}\left(\overline{Y_{s}^{\overline{\xi},\tau}}\right)H\left(Y_{s}^{\overline{\xi},\tau},s;t,x\right)\wedge\theta_{0}(\xi_{1})\right]\rmd\tau\right)\rmd\xi_{1}\rmd s
\end{align*}
and
\[
H(\eta,s;t,x)=\frac{\partial}{\partial\eta_{1}}\mathbb{E}\left[1_{\mathbb{R}_{+}^{2}}(X_{t}^{\eta,s})\varLambda_{2}(X_{t}^{\eta,s},x)\right]
\]
for $x\in\mathbb{R}_{+}^{2}$ and $t>s\geq0$. 
\end{thm}

\begin{proof}
By (\ref{3D-BS-01}) and (\ref{2D-0-N1}), we obtain that
\begin{align*}
u(x,t)  =&\int_{\mathbb{R}_{+}^{2}}\varLambda_{2}(y,x)\wedge\omega(y,t)\rmd y\\
 = &\int_{\mathbb{R}_{+}^{2}}\varLambda_{2}(y,x)\wedge1_{\{y_{2}=0\}}\sigma(y_{1},t)\rmd y\\
 & +\int_{\mathbb{R}^{2}}\int_{\mathbb{R}_{+}^{2}}\varLambda_{2}(y,x)\wedge\hat{\omega}(\eta,0)p_{\nu}(0,\eta,t,y)\rmd y\rmd\eta\\
 & +\int_{0}^{t}\int_{\mathbb{R}^{2}}\int_{\mathbb{R}_{+}^{2}}\varLambda_{2}(y,x)\wedge\hat{F}(\eta,s)p_{\nu}(s,\eta,t,y)\rmd y\rmd\eta\rmd s\\
 & +2\nu\int_{0}^{t}\int_{-\infty}^{\infty}\left.\frac{\partial}{\partial\eta_{2}}\right|_{\eta_{2}=0}\int_{\mathbb{R}_{+}^{2}}\varLambda_{2}(y,x)\wedge\sigma(\eta_{1},s)p_{\nu}(s,(\eta_{1},\eta_{2}),t,y)\rmd y\rmd\eta_{1}\rmd s.
\end{align*}
Hence 
\begin{align*}
u(x,t)  =&\int_{\mathbb{R}^{2}}\mathbb{E}\left[1_{\mathbb{R}_{+}^{2}}(X_{t}^{\eta})\varLambda_{2}(X_{t}^{\eta},x)\wedge\hat{\omega}(\eta,0)\right]\rmd\eta\nonumber \\
 & +\int_{0}^{t}\int_{\mathbb{R}^{2}}\mathbb{E}\left[1_{\mathbb{R}_{+}^{2}}(X_{t}^{\eta,s})\varLambda_{2}(X_{t}^{\eta,s},x)\wedge\hat{F}(\eta,s)\right]\rmd\eta\rmd s\nonumber \\
 & +2\nu\int_{0}^{t}\int_{-\infty}^{\infty}\left.\frac{\partial}{\partial\eta_{2}}\right|_{\eta_{2}=0}\mathbb{E}\left[1_{\mathbb{R}_{+}^{2}}(X_{t}^{\eta,s})\varLambda_{2}(X_{t}^{\eta,s},x)\wedge\sigma(\eta,s)\right]\rmd\eta_{1}\rmd s,
\end{align*}
where
\[
F(x,t)=\frac{\partial}{\partial x_{1}}\theta(x,t)-\frac{\partial}{\partial x_{1}}\theta_{0}(x_{1}),
\]
which yields in particular the first representation for $u(x,t)$.
To prove the second representation, we need to handle the second term
\begin{align*}
I_{2}  =&\int_{0}^{t}\int_{\mathbb{R}^{2}}\mathbb{E}\left[1_{\mathbb{R}_{+}^{2}}(X_{t}^{\eta,s})\varLambda_{2}(X_{t}^{\eta,s},x)\wedge\hat{F}(\eta,s)\right]\rmd\eta\rmd s\\
 = &\int_{0}^{t}\int_{\mathbb{R}^{2}}\mathbb{E}\left[1_{\mathbb{R}_{+}^{2}}(X_{t}^{\eta,s})\varLambda_{2}(X_{t}^{\eta,s},x)\wedge\widehat{\frac{\partial\theta}{\partial\eta_{1}}}(\eta,s)\right]\rmd\eta\rmd s\\
 & -\int_{0}^{t}\int_{\mathbb{R}^{2}}\mathbb{E}\left[1_{\mathbb{R}_{+}^{2}}(X_{t}^{\eta,s})\varLambda_{2}(X_{t}^{\eta,s},x)\wedge\widehat{\frac{\partial\theta_{0}}{\partial\eta_{1}}}(\eta_{1},s)\right]\rmd\eta\rmd s,
\end{align*}
where the first term on the right-hand side is denoted by
\[
I_{2,1}=\int_{0}^{t}\int_{\mathbb{R}^{2}}\mathbb{E}\left[1_{\mathbb{R}_{+}^{2}}(X_{t}^{\eta,s})\varLambda_{2}(X_{t}^{\eta,s},x)\right]\wedge\widehat{\frac{\partial\theta}{\partial\eta_{1}}}(\eta,s)\rmd\eta\rmd s.
\]
Next we notice that by definition 
\[
\widehat{\frac{\partial\theta}{\partial\eta_{1}}}(\eta,s)=\frac{\partial}{\partial\eta_{1}}\hat{\theta}(\eta,s)
\]
so that we may rewrite 
\begin{equation}
I_{2,1}=-\int_{0}^{t}\int_{\mathbb{R}^{2}}H(\eta,s;t,x)\wedge\hat{\theta}(\eta,s)\rmd\eta\rmd s,\label{I21-01}
\end{equation}
where for simplicity, we have introduced the following kernel
\[
H(\eta,s;t,x)=\frac{\partial}{\partial\eta_{1}}\mathbb{E}\left[1_{\mathbb{R}_{+}^{2}}(X_{t}^{\eta,s})\varLambda_{2}(X_{t}^{\eta,s},x)\right].
\]
Using the representation (\ref{2D-t-N1}), we have
\begin{align*}
\theta(\eta,s) =& 1_{\{\eta_{2}=0\}}\theta_{0}(\eta_{1})+\int_{\mathbb{R}^{2}}p_{\kappa}(0,\xi,s,\eta)\hat{\theta}(\xi,0)\rmd\xi\nonumber \\
 & +2\kappa\int_{-\infty}^{\infty}\theta_{0}(\xi_{1})\left.\frac{\partial}{\partial\xi_{2}}\right|_{\xi_{2}=0}\left(\int_{0}^{s}p_{\kappa}(\tau,\xi,s,\eta)\rmd\tau\right)\rmd\xi_{1},
\end{align*}
so that
\begin{align*}
\hat{\theta}(\eta,s)  =&\int_{\mathbb{R}^{2}}1_{\mathbb{R}_{+}^{2}}(\eta)p_{\kappa}(0,\xi,s,\eta)\hat{\theta}(\xi,0)\rmd\xi-\int_{\mathbb{R}^{2}}1_{\mathbb{R}_{+}^{2}}(\overline{\eta})p_{\kappa}(0,\overline{\xi},s,\eta)\hat{\theta}(\xi,0)\rmd\xi\\
 & +2\kappa\int_{-\infty}^{\infty}\theta_{0}(\xi_{1})\left.\frac{\partial}{\partial\xi_{2}}\right|_{\xi_{2}=0}\left(\int_{0}^{s}1_{\mathbb{R}_{+}^{2}}(\eta)p_{\kappa}(\tau,\xi,s,\eta)\rmd\tau\right)\rmd\xi_{1}\\
 & -2\kappa\int_{-\infty}^{\infty}\theta_{0}(\xi_{1})\left.\frac{\partial}{\partial\xi_{2}}\right|_{\xi_{2}=0}\left(\int_{0}^{s}1_{\mathbb{R}_{+}^{2}}(\overline{\eta})p_{\kappa}(\tau,\overline{\xi},s,\eta)\rmd \tau\right)\rmd\xi_{1}.
\end{align*}
Substituting this expression into (\ref{I21-01}) we then deduce that
\begin{align*}
I_{2,1}  =&-\int_{0}^{t}\int_{\mathbb{R}^{2}}\mathbb{E}\left[1_{\mathbb{R}_{+}^{2}}\left(Y_{s}^{\xi}\right)H\left(Y_{s}^{\xi},s;t,x\right)\wedge\hat{\theta}(\xi,0)\right]\rmd\xi\rmd s\\
 & +\int_{0}^{t}\int_{\mathbb{R}^{2}}\mathbb{E}\left[1_{\mathbb{R}_{+}^{2}}\left(\overline{Y_{s}^{\overline{\xi}}}\right)H\left(Y_{s}^{\overline{\xi}},s;t,x\right)\wedge\hat{\theta}(\xi,0)\right]\rmd\xi\rmd s\\
 & -2\kappa\int_{0}^{t}\int_{-\infty}^{\infty}\left.\frac{\partial}{\partial\xi_{2}}\right|_{\xi_{2}=0}\left(\int_{0}^{s}\mathbb{E}\left[1_{\mathbb{R}_{+}^{2}}\left(Y_{s}^{\xi,\tau}\right)H\left(Y_{s}^{\xi,\tau},s;t,x\right)\wedge\theta_{0}(\xi_{1})\right]\rmd\tau\right)\rmd\xi_{1}\rmd s\\
 & +2\kappa\int_{0}^{t}\int_{-\infty}^{\infty}\left.\frac{\partial}{\partial\xi_{2}}\right|_{\xi_{2}=0}\left(\int_{0}^{s}\mathbb{E}\left[1_{\mathbb{R}_{+}^{2}}\left(\overline{Y_{s}^{\overline{\xi},\tau}}\right)H\left(Y_{s}^{\overline{\xi},\tau},s;t,x\right)\wedge\theta_{0}(\xi_{1})\right]\rmd\tau\right)\rmd\xi_{1}\rmd s.
\end{align*}
Putting together, we obtain the second representation.
\end{proof}

\section{Numerical schemes and experiment results}

In this section, we formulate several numerical schemes based on the representations for flows in two dimensional space established in the previous sections, and demonstrate the numerical  results. 

Let us review our notations for the sake of comprehensibility. There
are four fluid dynamical variables we are going to calculate by means
of numerical simulations: the velocity
$u(x,t)$, the vorticity $\omega(x,t)$, the temperature $\theta(x,t)$
and the temperature gradient $\varTheta(x,t)$. We are given the initial velocity
$u(x,0)$, and hence the initial vorticity $\omega(x,0)$ as well as the initial
temperature $\theta(x,0)$, which has a small gradient, i.e., the magnitude
of the gradient $\varTheta(x,0)=\nabla\theta(x,0)$ is small so can be ignored in numerical schemes. The
kinematic viscosity $\nu>0$ and the heat diffusivity constant $\kappa>0$
depend on the nature of the fluid. The fluid density $\rho$ is almost
a constant, so we may choose it to be the unit. 

\subsection{Oberbeck-Boussinesq flows in \texorpdfstring{$\mathbb{R}^{2}$}{R2}}

The numerical experiments in this part are carried out for two-dimensional fluid flows on the whole space,
and therefore, the Biot-Savart singular kernel is a vector kernel given
by $K_{2}(y,x)=(2\pi)^{-1}(y-x)/|y-x|^{2}$, and the interaction force
$f(\theta)(x,t)=(0,\alpha(\theta(x,t)-\theta(x,0)))$ where $\alpha>0$ is a
constant. Thus
\[
F(x,t)=\alpha\frac{\partial}{\partial x_{1}}\theta(x,t)-\alpha\frac{\partial}{\partial x_{1}}\theta(x,0),
\]
where $\theta(x,0)$ is the given initial temperature distribution
whose gradient is small. 

Choose a lattice mesh $h>0$. For $i_1,i_2\in \mathbb{Z}$, denote $x^{i_{1},i_{2}}=(i_{1},i_{2})h$, the lattice points,  $\omega_{i_{1},i_{2}}=\omega\left(x^{i_{1},i_{2}},0\right)$,
and $\varTheta_{i_{1},i_{2}}=\varTheta(x^{i_{1},i_{2}},0)$. Let $\delta>0$ be the step length of the time variable
and $t_{i}=i\delta$, $i=0,1,2,\cdots$.

Note that in the numerical schemes based on the functional integral
representations, the interaction term $F$ involves the temperature
gradient, therefore, we have to calculate the derivatives of $\theta(x,t)$.
Since the derivative of the Biot-Savart kernel is no longer locally
integrable, we need to replace $K_{2}$ with its regularisation
measured via a positive parameter $\epsilon>0$, and introduce
\[
K_{2,\epsilon}(y,x)=\left(1-e^{-\frac{|y-x|^{2}}{\epsilon}}\right)K_{2}(y,x).
\]
 Let us describe the numerical scheme by adopting the functional integral
representations in Subsection \ref{subsec52}.

\subsubsection{One copy scheme}

In this numerical scheme, we drop the expectation using independent copies of Brownian motion. 
We discretise the stochastic differential equations using the Euler scheme: for $i_1,i_2\in \mathbb{Z}$, $k=1,2,\cdots$,
\begin{equation}
X_{0}^{i_{1},i_{2}}=x^{i_{1},i_{2}},\quad X_{t_{k}}^{i_{1},i_{2}}=X_{t_{k-1}}^{i_{1},i_{2}}+\delta u\left(X_{t_{k-1}}^{i_{1},i_{2}},t_{k-1}\right)+\sqrt{2\nu}(B_{t_{k}}^{1}-B_{t_{k-1}}^{1}),\label{XB}
\end{equation}
\begin{equation}
Y_{0}^{i_{1},i_{2}}=x^{i_{1},i_{2}},\quad Y_{t_{k}}^{i_{1},i_{2}}=Y_{t_{k-1}}^{i_{1},i_{2}}+\delta u\left(Y_{t_{k-1}}^{i_{1},i_{2}},t_{k-1}\right)+\sqrt{2\kappa}(B_{t_{k}}^{2}-B_{t_{k-1}}^{2}),\label{YB}
\end{equation}
where we use $X_{t}^{i_{1},i_{2}}$ to denote $X_{t}^{x^{i_{1},i_{2}}}$ to simplify our notation, and $B^{1}$, $B^{2}$ are two independent two-dimensional Brownian
motions.

The integral representations in Theorems \ref{thm6.4} and \ref{thm6.5} are approximated by the following discretisation:
\begin{align}
u(x,t_{k})=&\sum_{i_{1},i_{2}}h^{2}K_{2,\epsilon}\left(X_{t_{k}}^{i_{1},i_{2}},x\right)\wedge\omega_{i_{1},i_{2}}
\nonumber\\
&+\sum_{i_1,i_2}\sum_{j=1}^{k}\delta h^{2}K_{2,\epsilon}\left(X_{t_{k}}^{i_{1},i_{2}},x\right)\wedge F\left(X_{t_{j-1}}^{i_{1},i_{2}},t_{j-1}\right),\label{du-01}
\end{align}
and
\begin{equation}
\theta(x,t_{k})=-\sum_{i_{1},i_{2}}h^{2}K_{2,\epsilon}\left(Y_{t_{k}}^{i_{1},i_{2}},x\right)\cdot R(x^{i_{1},i_{2}},t_{k};0)\varTheta_{i_{1},i_{2}},\label{dt-01}
\end{equation}
where $R(x^{i_{1},i_{2}},t_{k};t_{k})=I$, and
\[
R(x^{i_{1},i_{2}},t_{k};t_{i})=I-\sum_{l=i+1}^{k}\delta A\left(Y_{t_{l}}^{i_{1},i_{2}},t_{l}\right)R(x^{i_{1},i_{2}},t_{k};t_{l})
\]
for $i=0,\cdots,k-1$. As for $F(x,t)$ and $A(x,t)$,  we update them in each iteration by formally differentiating the equations (\ref{du-01})
and (\ref{dt-01}) respectively, so that
\begin{align}
A(x,t_{k})=&\sum_{i_{1},i_{2}}h^{2}\nabla_{x}K_{2,\epsilon}\left(X_{t_{k}}^{i_{1},i_{2}},x\right)\wedge\omega_{i_{1},i_{2}}\nonumber\\
    &+\sum_{i_1,i_2}\sum_{j=1}^{k}\delta h^{2}\nabla_{x}K_{2,\epsilon}\left(X_{t_{k}}^{i_{1},i_{2}},x\right)\wedge F\left(X_{t_{j-1}}^{i_{1},i_{2}},t_{j-1}\right)\label{A-0},
\end{align}
\[
\varTheta(x,t_{k})=-\sum_{i_{1},i_{2}}h^{2}\nabla_{x}K_{2,\epsilon}\left(Y_{t_{k}}^{i_{1},i_{2}},x\right)\cdot R(x^{i_{1},i_{2}},t_{k};0)\varTheta_{i_{1},i_{2}},
\]
and 
\[
F(X^{i_1,i_2}_{t_k}, t_k) = \alpha \varTheta_1(X^{i_1,i_2}_{t_k}, t_k) -\alpha \frac{\partial \theta}{\partial x_1} (X^{i_1,i_2}_{t_k},0),
\]
where the gradient of the Biot-Savart kernel is replaced by
\begin{equation}
\nabla_{x}K_{2,\epsilon}(y,x)=\nabla_{x}\left[\left(1-e^{-\frac{|y-x|^{2}}{\epsilon}}\right)K_2(y,x)\right].\label{D-K}
\end{equation}

\begin{rem}
This scheme is still quite computationally expensive as it requires
storing the values of $A$ at all times. Since in any simulation
for non-linear dynamics, the time duration cannot be long, and one can use their approximations of the time integral. That is, equations (\ref{du-01}) and (\ref{A-0}) can be substituted with the following iterations:
\begin{align*}
u(x,t_{k})=&\sum_{i_{1},i_{2}}h^{2}K_{2,\epsilon}\left(X_{t_{k}}^{i_{1},i_{2}},x\right)\wedge\omega_{i_{1},i_{2}}
    \nonumber\\
    &+ \sum_{i_{1},i_{2}} k\delta h^{2}K_{2,\epsilon}\left(X_{t_{k}}^{i_{1},i_{2}},x\right)\wedge F\left(X_{t_{k-1}}^{i_{1},i_{2}},t_{k-1}\right)
\end{align*}
and
\begin{align*}
A(x,t_{k})=&\sum_{i_{1},i_{2}}h^{2}\nabla_{x}K_{2,\epsilon}\left(X_{t_{k}}^{i_{1},i_{2}},x\right)\wedge\omega_{i_{1},i_{2}}\nonumber\\
&+ \sum_{i_{1},i_{2}} k\delta h^{2}\nabla_{x}K_{2,\epsilon}\left(X_{t_{k}}^{i_{1},i_{2}},x\right)\wedge F\left(X_{t_{k-1}}^{i_{1},i_{2}},t_{k-1}\right)
\end{align*}
respectively.
\end{rem}

\subsubsection{Multi-copy scheme}
We introduce the second numerical scheme, where the expectation is substituted with the empirical mean based on the strong law of large numbers. Take $2N$ independent copies of Brownian motions $(B^{1,m}_t)$ and $(B^{2,m}_t)$, $m=1,2,\cdots, N$. 

We repeat the diffusion processes of the twin particle $N$ times by running $2N$ independent copies of Brownian motion and replacing the expectations with their averages.  That is, 
for $1\leq m\leq N$, define
\[
X_{0}^{m,i_{1},i_{2}}=x^{i_{1},i_{2}},\quad X_{t_{k}}^{m,i_{1},i_{2}}=X_{t_{k-1}}^{m,i_{1},i_{2}}+\delta u\left(X_{t_{k-1}}^{m,i_{1},i_{2}},t_{k-1}\right)+\sqrt{2\nu}(B_{t_{k}}^{1,m}-B_{t_{k-1}}^{1,m})
\]
and
\[
Y_{0}^{m,i_{1},i_{2}}=x^{i_{1},i_{2}},\quad Y_{t_{k}}^{m,i_{1},i_{2}}=Y_{t_{k-1}}^{m,i_{1},i_{2}}+\delta u\left(Y_{t_{k-1}}^{m,i_{1},i_{2}},t_{k-1}\right)+\sqrt{2\kappa}(B_{t_{k}}^{2,m}-B_{t_{k-1}}^{2,m}).
\]
Then the velocity and temperature are approximated by 
\begin{align*}
u(x,t_{k})=&\frac{1}{N}\sum_{m=1}^N\sum_{i_{1},i_{2}}h^{2}K_{2,\epsilon}\left(X_{t_{k}}^{m,i_{1},i_{2}},x\right)\wedge\omega_{i_{1},i_{2}}\\
&+\frac{1}{N}\sum_{m=1}^N \sum_{i_1,i_2}\sum_{j=1}^{k}\delta h^{2}K_{2,\epsilon}\left(X_{t_{k}}^{m,i_{1},i_{2}},x\right)\wedge F\left(X_{t_{j-1}}^{m,i_{1},i_{2}},t_{j-1}\right), \\
\theta(x,t_{k})=&-\frac{1}{N}\sum_{m=1}^N\sum_{i_{1},i_{2}}h^{2}K_{2,\epsilon}\left(Y_{t_{k}}^{m,i_{1},i_{2}},x\right)\cdot R(x^{i_{1},i_{2}},t_{k};0)\varTheta_{i_{1},i_{2}},
\end{align*}
where similar to the first scheme, 
\begin{align*} 
A(x,t_{k})= &\frac{1}{N}\sum_{m=1}^N \sum_{i_{1},i_{2}}h^{2}\nabla_{x}K_{2,\epsilon}\left(X_{t_{k}}^{m,i_{1},i_{2}},x\right)\wedge\omega_{i_{1},i_{2}}\\
&+\frac{1}{N}\sum_{m=1}^N \sum_{i_1,i_2}\sum_{j=1}^{k}\delta h^{2}\nabla_{x}K_{2,\epsilon}\left(X_{t_{k}}^{m,i_{1},i_{2}},x\right)\wedge F\left(X_{t_{j-1}}^{m,i_{1},i_{2}},t_{j-1}\right), \\
 \varTheta(x,t_{k})=&-\frac{1}{N}\sum_{m=1}^N \sum_{i_{1},i_{2}}h^{2}\nabla_{x}K_{2,\epsilon}\left(Y_{t_{k}}^{m,i_{1},i_{2}},x\right)\cdot R(x^{i_{1},i_{2}},t_{k};0)\varTheta_{i_{1},i_{2}},
\end{align*}
where $R$ depends on $m$ such that for each $m$, $R(x^{i_{1},i_{2}},t_{k};t_{k})=I$, and
\[
R(x^{i_{1},i_{2}},t_{k};t_{i})=I-\sum_{l=i+1}^{k}\delta A\left(Y_{t_{l}}^{m,i_{1},i_{2}},t_{l}\right)R(x^{i_{1},i_{2}},t_{k};t_{l}).
\]
The rest of the scheme remains the same as in the one-copy scheme.

\subsubsection{Numerical experiments}
Following the one-copy numerical scheme described above, we carried out several simple numerical experiments. Here, we present the results of experiments with different Prandtl numbers. In the first experiment, we set $\nu = 1$ and $\kappa = 0.15$, so that the Prandtl number $\mathrm{Pr}= 0.15$. In the second experiment, we swap these two values and consider the case when $\mathrm{Pr}=6.67$. 

We choose the typical length scale $L= 2\pi$, and assume that $\alpha = 0.0005$. The parameter we use to smooth out the Biot-Savart kernel is chosen to be $\epsilon = 0.1$. 

In the experiment presented, we set the initial velocity to be of the form
\[
u(x,0) = (-10\sin(x_2),0),
\]
and the initial temperature is given by 
\[
\theta_0(x) = \theta(x,0) = 0.01(8\pi^2-{x_1}^2-{x_2}^2).
\]

Thus, $\omega(x,0) = 10\cos(x_2)$, and 
\[
\omega_{i_1,i_2} = 10\cos(i_2 h).
\]
The time step is $\delta = 0.01$ with mesh size $h = 2\pi/40$. The numerical experiment results at times $t = 0.6$, $t= 1.2$, $t=1.8$ 
 with $\Pr = 0.15$ and $\Pr = 6.67$ are shown in the Figure \ref{OB-R2-sim-015} and Figure \ref{OB-R2-sim-667}, respectively.

 \begin{figure}[!ht]
     \centering
     \begin{subfigure}[b]{0.3\textwidth}
     \includegraphics[width=\textwidth]{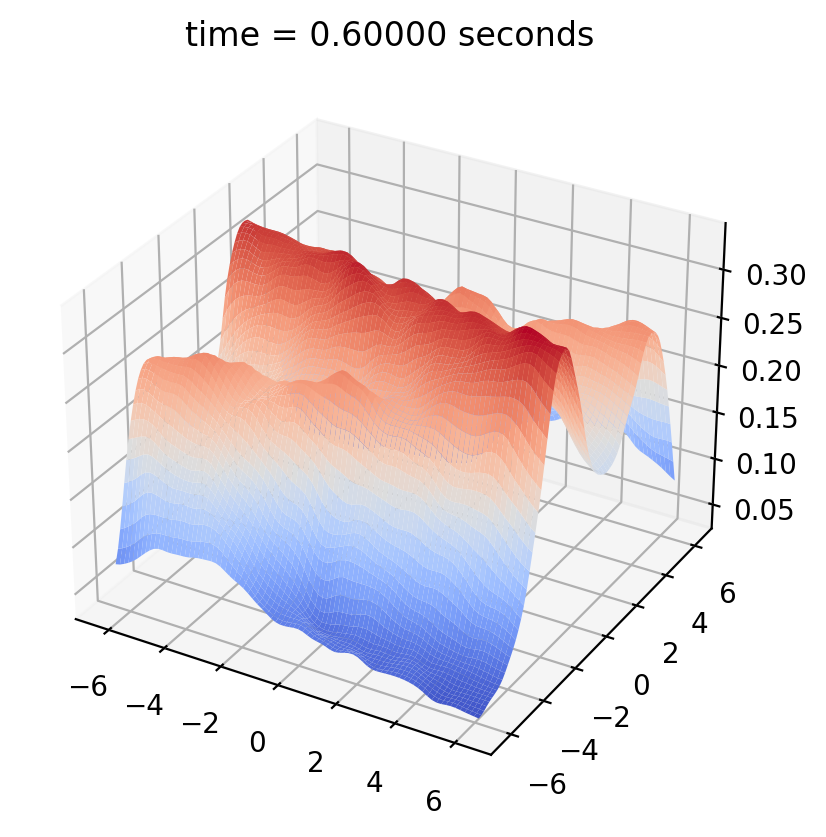}
     \caption{Temperature at $t=0.6$.}
     \end{subfigure}
     \hspace{0.5cm}
     \begin{subfigure}[b]{0.3\textwidth}
     \includegraphics[width=\textwidth]{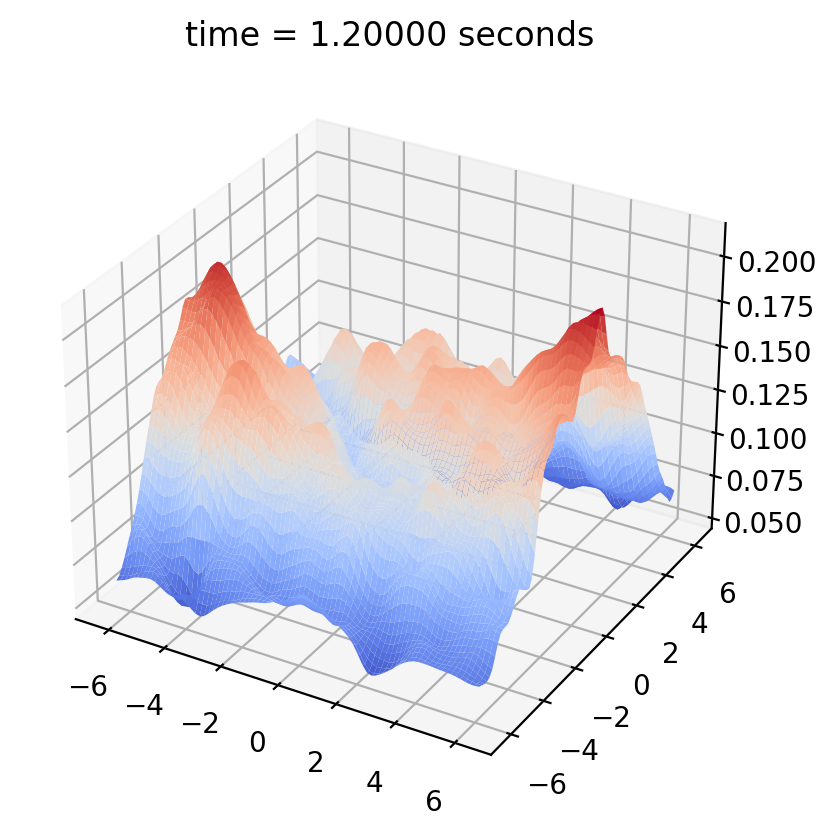}
     \caption{Temperature at $t=1.2$.}
     \end{subfigure}
     \hspace{0.5cm}
     \begin{subfigure}[b]{0.3\textwidth}
     \includegraphics[width=\textwidth]{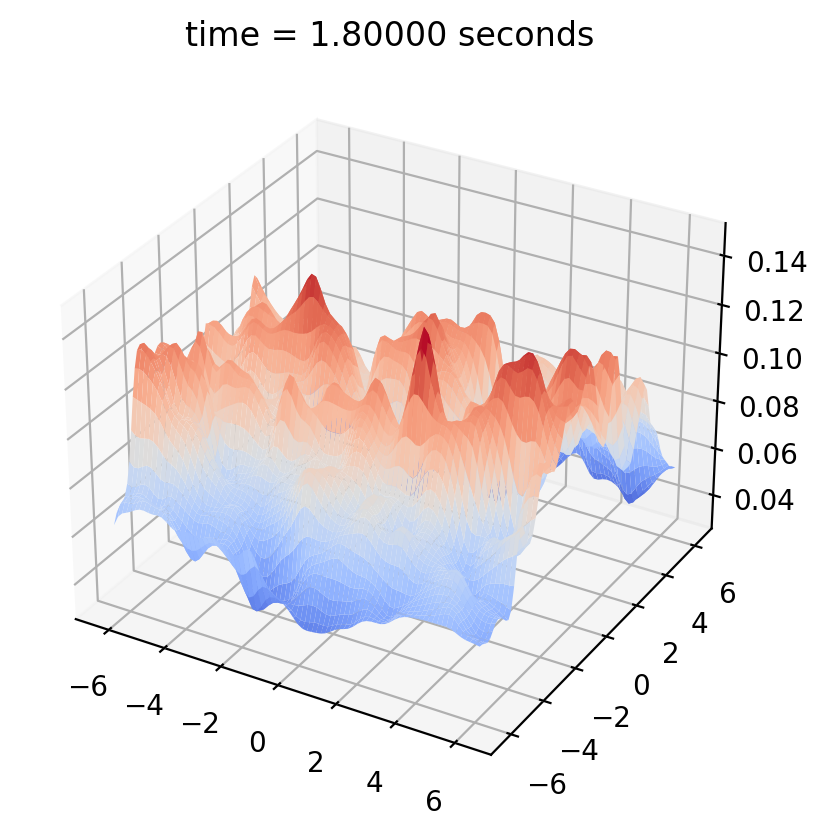}
     \caption{Temperature at $t=1.8$.}
     \end{subfigure}
     
     \vspace{0.5cm}
     
     \begin{subfigure}[b]{0.3\textwidth}
     \includegraphics[width=\textwidth]{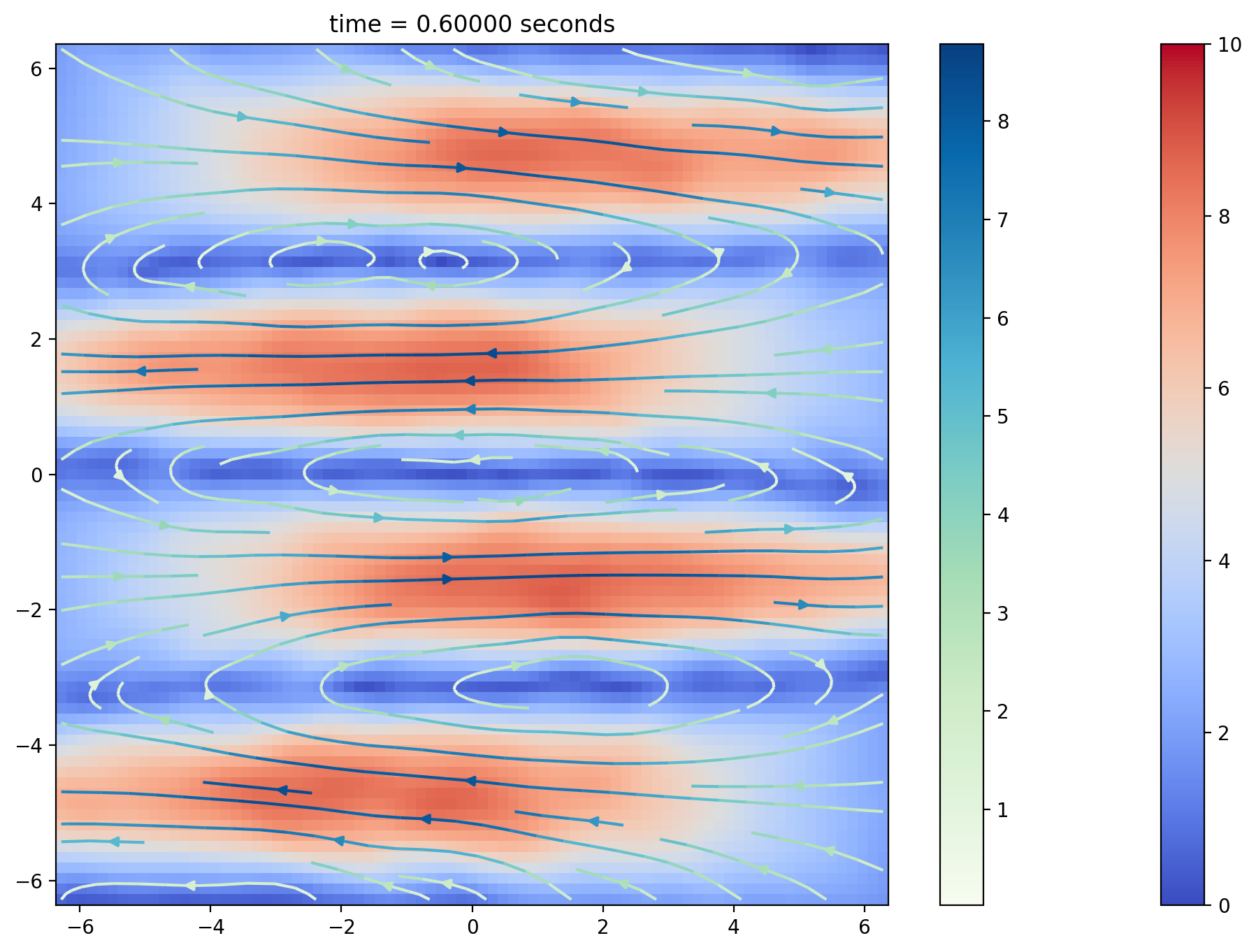}
     \caption{Velocity at $t=0.6$.}
     \end{subfigure}     
     \hspace{0.5cm}
     \begin{subfigure}[b]{0.3\textwidth}
     \includegraphics[width=\textwidth]{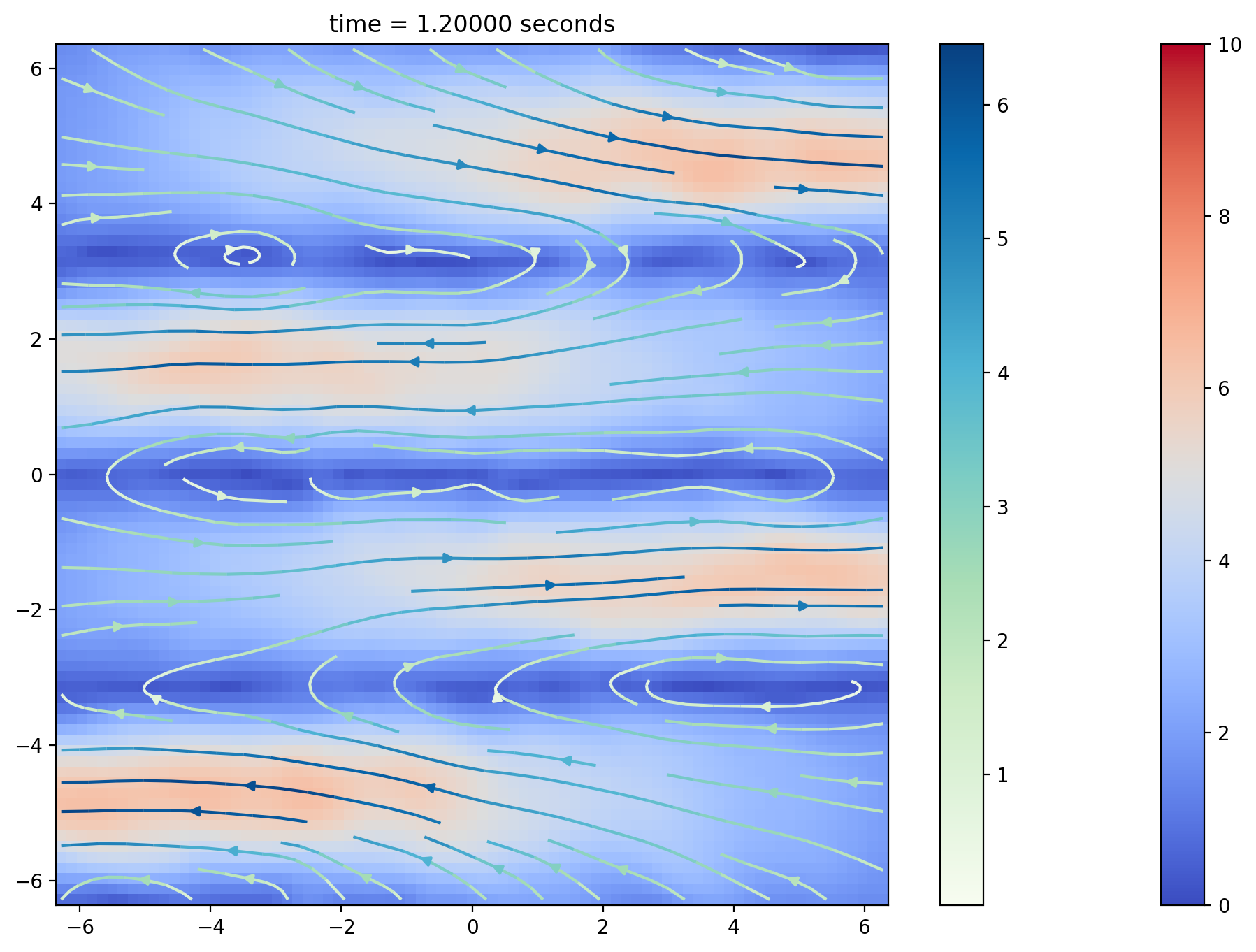}
     \caption{Velocity field at $t=1.2$.}
     \end{subfigure}    
     \hspace{0.5cm}
     \begin{subfigure}[b]{0.3\textwidth}
     \includegraphics[width=\textwidth]{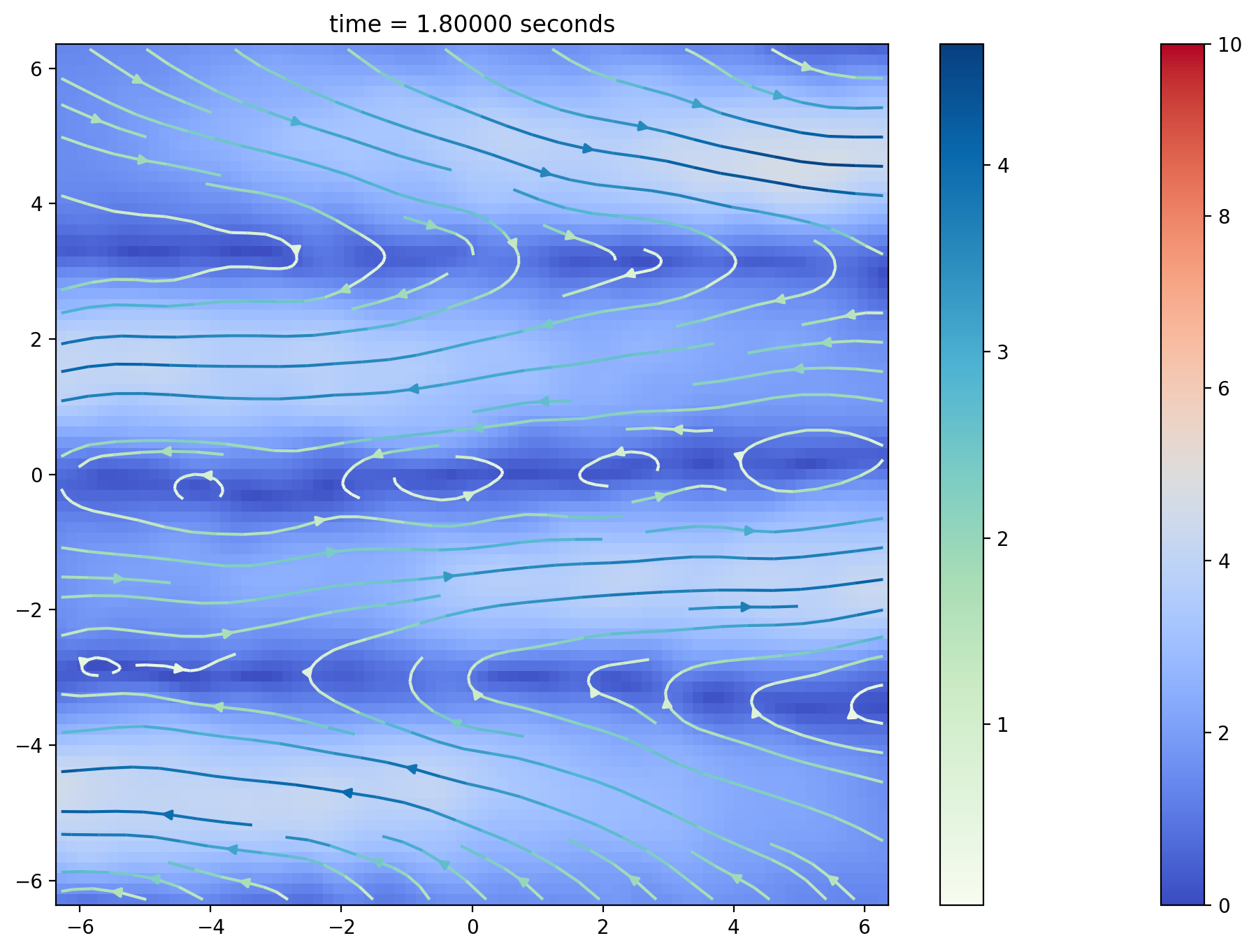}
     \caption{Velocity field at $t=1.8$.}
     \end{subfigure}
 \caption{Temperature and velocity fields of Oberbeck-Boussinesq flows on $\mathbb{R}^2$ with Prandtl number $\mathrm{Pr} = 0.15$.}
 \label{OB-R2-sim-015}
 \end{figure}

 \begin{figure}[!ht]
     \centering
     \begin{subfigure}[b]{0.3\textwidth}
     \includegraphics[width=\textwidth]{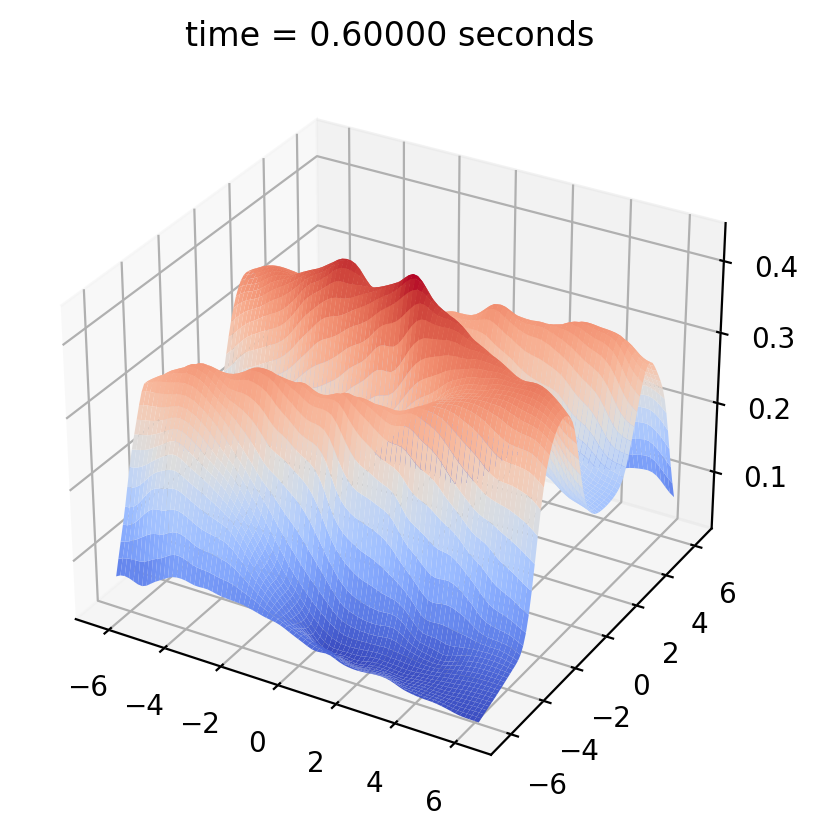}
     \caption{Temperature at $t=0.6$.}
     \end{subfigure}
     \hspace{0.5cm}
     \begin{subfigure}[b]{0.3\textwidth}
     \includegraphics[width=\textwidth]{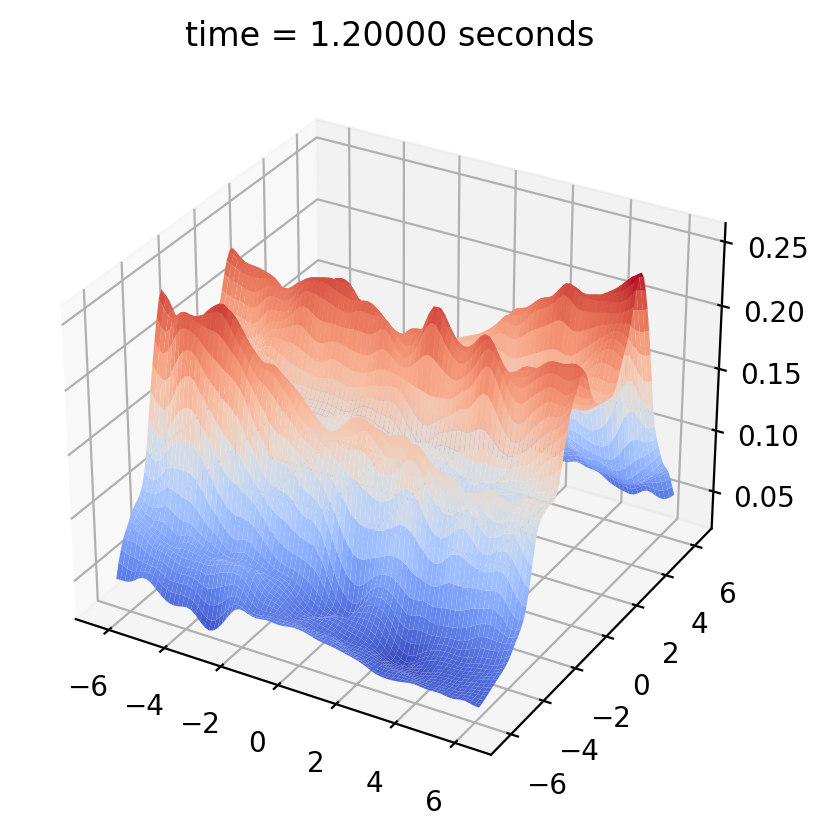}
     \caption{Temperature at $t=1.2$.}
     \end{subfigure}
     \hspace{0.5cm}
     \begin{subfigure}[b]{0.3\textwidth}
     \includegraphics[width=\textwidth]{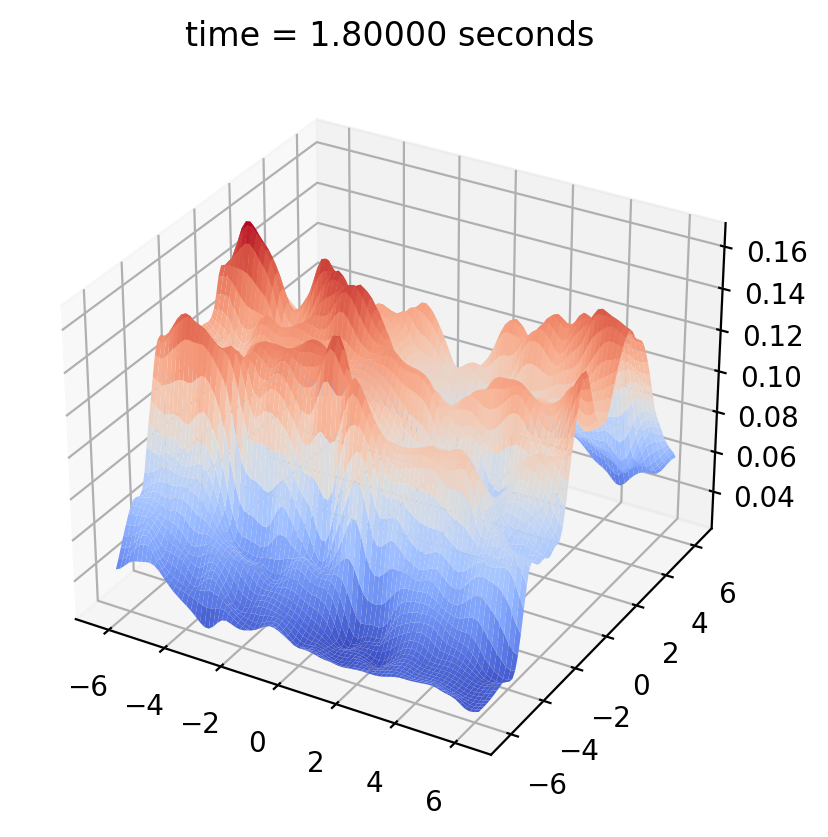}
     \caption{Temperature at $t=1.8$.}
     \end{subfigure}
     
     \vspace{0.5cm}
     
     \begin{subfigure}[b]{0.3\textwidth}
     \includegraphics[width=\textwidth]{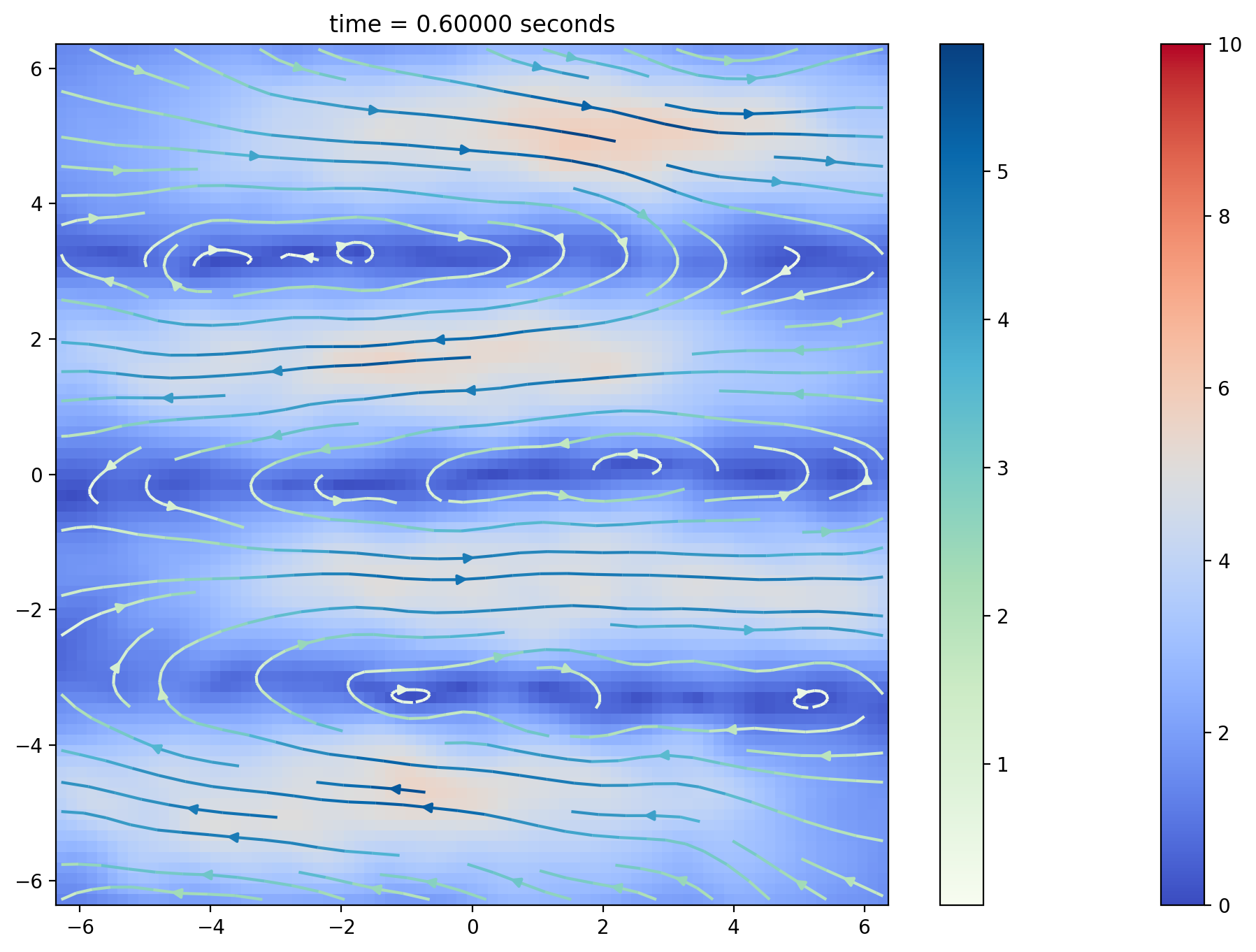}
     \caption{Velocity at $t=0.6$.}
     \end{subfigure}     
     \hspace{0.5cm}
     \begin{subfigure}[b]{0.3\textwidth}
     \includegraphics[width=\textwidth]{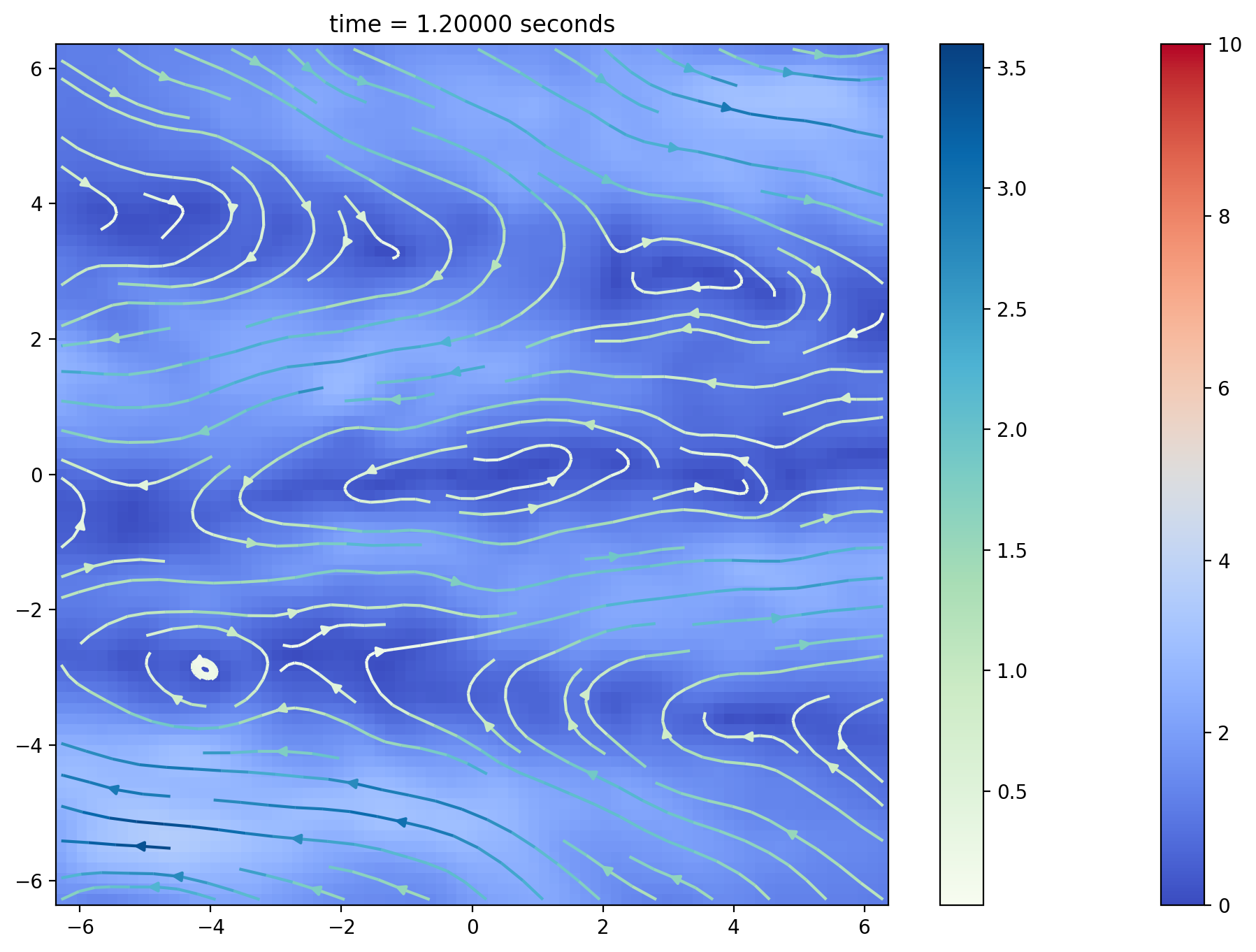}
     \caption{Velocity field at $t=1.2$.}
     \end{subfigure}    
     \hspace{0.5cm}
     \begin{subfigure}[b]{0.3\textwidth}
     \includegraphics[width=\textwidth]{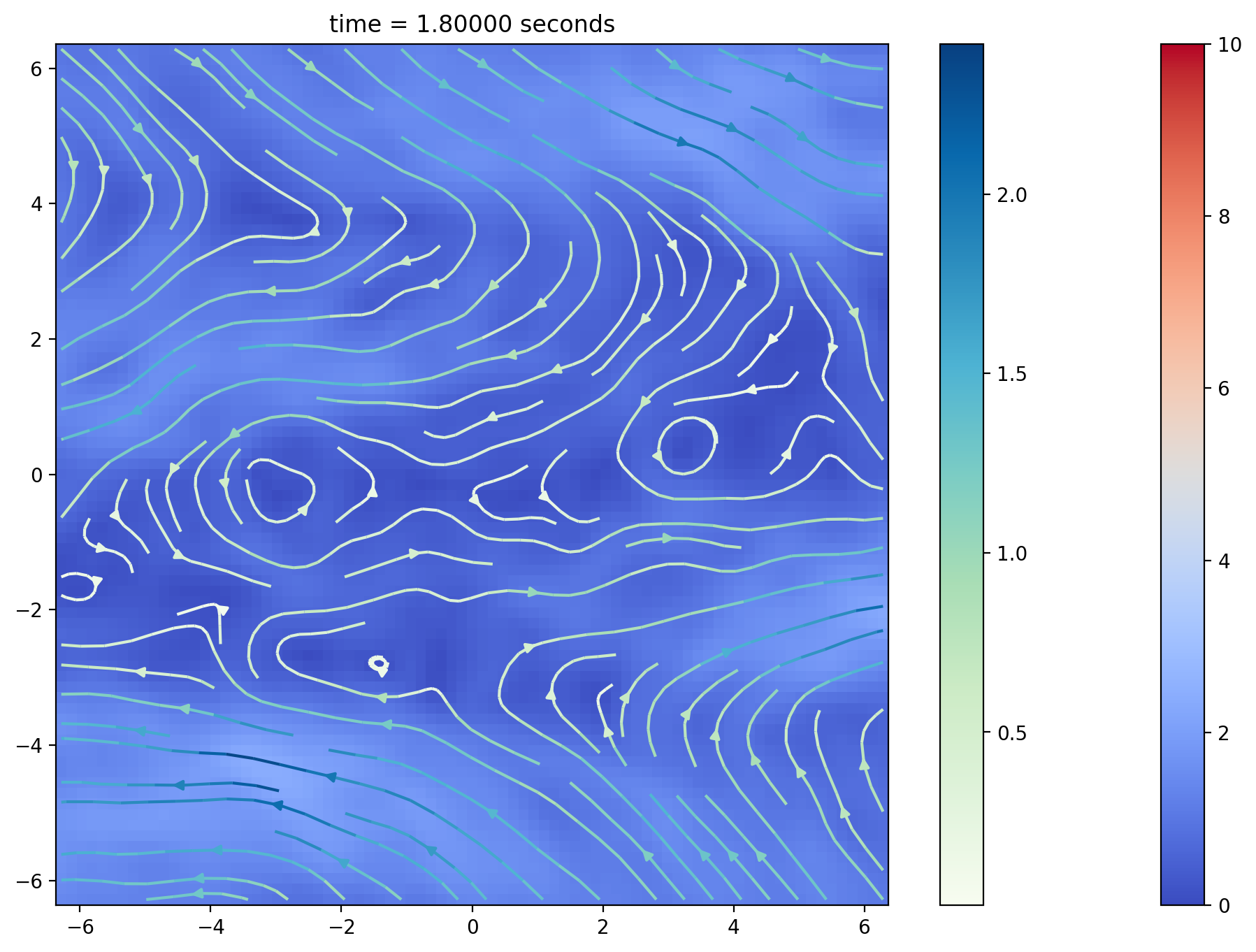}
     \caption{Velocity field at $t=1.8$.}
     \end{subfigure}
 \caption{Temperature and velocity fields of Oberbeck-Boussinesq flows on $\mathbb{R}^2$ with Prandtl number $\mathrm{Pr} = 6.67$.}
 \label{OB-R2-sim-667}
 \end{figure}
    
The figures show how the buoyancy from small temperature variations accelerates the flow. Besides, the growth of temperature is different from the result of a linear parabolic equation - the nonlinearity revealed in the temperature evolution reflects the velocity in the drift term is, in turn, driven by the thermal convection.

\subsection{Oberbeck-Boussinesq flows in wall-bounded domains}

This series of numerical experiments are based on the wall-bounded fluid flows
in Subsection \ref{subsec651}. Therefore the fluid is heated from the bottom with an external source of heat at temperature
$\theta_{0}(x)=\theta_{0}(x_{1})$ , which depends only on the first coordinate. 

The Biot-Savart kernel for this case is given by
\[
\varLambda_{2}(y,x)=\frac{1}{2\pi}\left(\frac{y-x}{|y-x|^{2}}-\frac{y-\overline{x}}{|y-\overline{x}|^{2}}\right)\quad\textrm{ for }y\neq x\textrm{ or }\overline{x}
\]
and for the same reason as in the previous case, we take its regularisation
via a real parameter $\epsilon>0$, and replace the singular integral
kernel by 
\[
\varLambda_{2,\epsilon}(y,x)=\left(1-e^{-\frac{|y-x|^{2}}{\epsilon}}\right)\varLambda_{2}(y,x).
\]

The numerical scheme is based on the functional integral representations
(\ref{u-app-001}, \ref{t-app-01}). Choose $\varepsilon>0$ small.
The approximated representations we will use here are then given by
\begin{align}
u(x,t)  \sim&\int_{\mathbb{R}_{+}^{2}}\mathbb{E}\left[1_{\{t<\zeta(X^{\eta})\}}1_{\mathbb{R}_{+}^{2}}(X_{t}^{\eta})\varLambda_{2}(X_{t}^{\eta},x)\wedge\omega(\eta,0)\right]\rmd\eta\nonumber \\
 & +\int_{0}^{t}\int_{\mathbb{R}_{+}^{2}}\mathbb{E}\left[1_{\left\{ s>\gamma_{t}(X^{\eta})\right\} }1_{\mathbb{R}_{+}^{2}}(X_{t}^{\eta})\varLambda_{2}(X_{t}^{\eta},x)\wedge F(X_{s}^{\eta},s)\right]\rmd\eta\rmd s\nonumber \\
 & +\int_{0}^{t}\int_{\mathbb{R}_{+}^{2}}\mathbb{E}\left[1_{\left\{ s>\gamma_{t}(X^{\eta})\right\} }1_{\mathbb{R}_{+}^{2}}(X_{t}^{\eta})\varLambda_{2}(X_{t}^{\eta},x)\wedge\tilde{\chi}_{\varepsilon}(X_{s}^{\eta},s)\right]\rmd\eta\rmd s\label{u-app-001-1}
\end{align}
and
\begin{align}
\theta(x,t) \sim &\frac{x_{2}}{\pi}\int_{-\infty}^{\infty}\frac{\theta_{0}(\xi_{1})}{|\xi_{1}-x_{1}|^{2}+x_{2}^{2}}\rmd \xi_{1}\nonumber \\
 & -\int_{\mathbb{R}_{+}^{2}}\mathbb{E}\left[1_{\{t<\zeta(Y^{\eta})\}}1_{\mathbb{R}_{+}^{2}}(Y_{t}^{\eta})\varLambda_{2}(Y_{t}^{\eta},x)\cdot R(\eta,t;0)\varTheta(\eta,0)\right]\rmd \eta.\label{t-S-01}
\end{align}
for $x\in\mathbb{R}_{+}^{2}$ and $t>0$. According to (\ref{stop-for01}),
the first term on the right-hand side of (\ref{u-app-001-1}) can
be written as
\[
\int_{\mathbb{R}^{2}}\mathbb{E}\left[1_{\mathbb{R}_{+}^{2}}(X_{t}^{\eta})\varLambda_{2}(X_{t}^{\eta},x)\wedge\hat{\omega}(\eta,0)\right]\rmd\eta, 
\]
where $\hat{\omega}$ can be computed using \eqref{eq:def-hat}
so that
\begin{align}
u(x,t) & \sim\int_{\mathbb{R}_{+}^{2}}\mathbb{E}\left[1_{\mathbb{R}_{+}^{2}}(X_{t}^{\eta})\varLambda_{2}(X_{t}^{\eta},x)\wedge\hat{\omega}(\eta,0)\right]\rmd\eta\nonumber \\
 & +\int_{0}^{t}\int_{\mathbb{R}_{+}^{2}}\mathbb{E}\left[1_{\left\{ s>\gamma_{t}(X^{\eta})\right\} }1_{\mathbb{R}_{+}^{2}}(X_{t}^{\eta})\varLambda_{2}(X_{t}^{\eta},x)\wedge F(X_{s}^{\eta},s)\right]\rmd\eta\rmd s\nonumber \\
 & +\int_{0}^{t}\int_{\mathbb{R}_{+}^{2}}\mathbb{E}\left[1_{\left\{ s>\gamma_{t}(X^{\eta})\right\} }1_{\mathbb{R}_{+}^{2}}(X_{t}^{\eta})\varLambda_{2}(X_{t}^{\eta},x)\wedge\tilde{\chi}_{\varepsilon}(X_{s}^{\eta},s)\right]\rmd\eta\rmd s.\label{u-S-01}
\end{align}

Now we are in a position to discretise (\ref{u-S-01}) and (\ref{t-S-01})
and obtain the numerical schemes. Namely, set $\hat{\omega}_{i_{1},i_{2}}=\hat{\omega}(x^{i_{1},i_{2}},0)$
for $i_{2}\geq0$. Let $\epsilon>0$ be another small constant to
take care of the differentiation of the singular kernel. Choose $\varepsilon\ll h$, and constants $h$, $\varepsilon$ and $\epsilon$
are fixed through numerical experiments, but they can be adjusted.

We assume that the bounded boundary layer is $B= \{x\in \mathbb{R}^2:0<x_2<\gamma\}$, and use a different mesh size $h_w$ for the boundary layer. Let $h_{i_1,i_2}$ be the vertical mesh size such that 
\[
h_{i_1,i_2} =
\begin{cases}
    h_w, & \text{in }B,\\
    h, &  \text{otherwise}.
\end{cases}
\]
\subsubsection{One copy scheme}

In this scheme, the expectations in the integral representations are treated by using the one-copy scheme, i.e., simply drop the expectation sign in the random vortex system. Due to the no-slip condition, the velocity and thus the heat conduction decreases rapidly within a thin layer adjoining the wall, which is commonly known as the boundary layer (see e.g. Chapter 4, \cite{LandauLifishitzFluid}).

Similar to the whole plane case, we discretise the SDE system using the Euler scheme:
\begin{align*}
        X_0^{i_1,i_2} = x^{i_1,i_2},
        \quad
        &X_{t_k}^{i_1,i_2} = X_{t_{k-1}}^{i_1,i_2} +\delta u(X_{t_{k-1}}^{i_1,i_2}, t_{k-1}) +
        \sqrt{2\nu} (B_{t_k}^1 - B_{t_{k-1}}^1),
        \\
        Y_0^{i_1,i_2} = x^{i_1,i_2},
        \quad
        &Y_{t_k}^{i_1,i_2} = Y_{t_{k-1}}^{i_1,i_2} +\delta u(Y_{t_{k-1}}^{i_1,i_2}, t_{k-1}) +
        \sqrt{2\kappa} (B_{t_k}^2 - B_{t_{k-1}}^2),        
\end{align*}
Furthermore, we set the force $F$ to be
\[F(x,t) = \alpha\left(\frac{\partial \theta}{\partial x_1}(x,t)-\frac{\partial \theta_0}{\partial x_1}(x_1)\right),\]
where $\theta_0(x_1)$ represents the external source of heat imposed on the boundary, and it is assumed to be almost constant $T_w>0$ on $[-L,L]$, vanish everywhere else, and has a small derivative.

The equation \eqref{u-S-01} is approximated by
\begin{align}
u(x,t_{k})  =& \sum_{i_{1}}\sum_{i_{2}\geq 0}hh_{i_1,i_2}\prod_{j=1}^k 1_{\mathbb{R}_{+}^{2}}(X_{t_{j}}^{i_{1},i_{2}})\varLambda_{2,\epsilon}(X_{t_{k}}^{i_{1},i_{2}},x)\wedge\hat{\omega}_{i_{1},i_{2}}\nonumber \\
 & +\sum_{i_{1}}\sum_{i_{2}\geq0}\delta hh_{i_1,i_2}\sum_{j=1}^{k}\prod_{l=j}^k 1_{\mathbb{R}_{+}^{2}}(X_{t_{l}}^{i_{1},i_{2}})\varLambda_{2,\epsilon}(X_{t_{k}}^{i_{1},i_{2}},x)\wedge F(X_{t_{j-1}}^{i_{1},i_{2}},t_{j-1})\nonumber \\
 & +\sum_{i_{1}}\delta h h_w\sum_{j=1}^{k}\prod_{l=j}^k 1_{\mathbb{R}_{+}^{2}}(X_{t_{l}}^{i_{1},0})\varLambda_{2,\epsilon}(X_{t_{k}}^{i_{1},0},x)\wedge\tilde{\chi}_{\varepsilon}(X_{t_{j-1}}^{i_{1},0},t_{j-1})\label{BS-01}
\end{align}
for $x=(x_{1},x_{2})$ with $x_{2}>0$, and 
\[
u(x,t)=\left(u^{1}((x_{1},-x_{2}),t),-u^{2}((x_{1},-x_{2}),t)\right)
\]
for $x=(x_{1},x_{2})$ with  $x_{2}<0$, where $\tilde{\chi}_{\varepsilon}$ is given
by (\ref{error-app-001}),
\begin{align*}
\tilde{\chi}_{\varepsilon}(x,t) & =-\nu\frac{1}{\varepsilon^{2}}\phi''\left(\frac{x_{2}}{\varepsilon}\right)\left.\frac{\partial}{\partial x_{2}}\right|_{x_{2}=0+}u^{1}(x,t)\\
 & =-\nu\frac{1}{\varepsilon^{2}}\phi''\left(\frac{x_{2}}{\varepsilon}\right)A_{2}^{1}((x_{1},0),t)
\end{align*}
with $\phi$ given by \eqref{phi-def}, whose second derivative is 
\[
\phi''(r) = 
\begin{cases}
    324(r-\frac{1}{2}),\quad &\frac{1}{3}\leq r\leq \frac{2}{3},\\
    0,\quad &\text{otherwise.}
\end{cases}
\]

The boundary integral is replaced with the third term on the right-hand side of \eqref{BS-01} since we assume $\varepsilon<h_w\ll h$ and $\tilde{\chi_\varepsilon}(x,t)$ only supports on $\frac{\varepsilon}{3}\leq x_2\leq \frac{2\varepsilon}{3}$, and thus it is sufficient to consider the diffusion starting within the boundary layer. 

To discretise (\ref{t-S-01}), we use the approximation 
\[
\frac{x_{2}}{\pi}\int_{-\infty}^{\infty}\frac{\theta_{0}(\xi_{1})}{|\xi_{1}-x_{1}|^{2}+x_{2}^{2}}\rmd \xi_{1} \approx T_w\frac{1}{\pi} \left( \tan^{-1} \left(\frac{L-x_1}{x_2}\right)+\tan^{-1} \left(\frac{L+x_1}{x_2}\right)\right)
\]
and thus
\begin{align}
\theta(x,&t_{k})  =\frac{T_w}{\pi} \left( \tan^{-1} \left(\frac{L-x_1}{x_2}\right)+\tan^{-1} \left(\frac{L+x_1}{x_2}\right)\right)\nonumber\\
 &-\sum_{i_{1}}\sum_{i_{2}\geq0}hh_{i_1,i_2}\prod_{j=1}^k 1_{\mathbb{R}_{+}^{2}}(Y_{t_{j}}^{i_{1},i_{2}})\varLambda_{2, \epsilon}(Y_{t_{k}}^{i_{1},i_{2}},x)\cdot R(x^{i_{1},i_{2}},t_{k};0)\varTheta_{i_{1},i_{2}}\label{BS-04}
\end{align}
and the gauge functional equation (\ref{Ga-02}) by $R(x^{i_{1},i_{2}},t_{k};t_{k})=I$, and
\[
R(x^{i_{1},i_{2}},t_{k};t_{i})=I-\sum_{l=i+1}^{k}\delta A\left(Y_{t_{l}}^{i_{1},i_{2}},t_{l}\right)R(x^{i_{1},i_{2}},t_{k};t_{l})1_{\mathbb{R}_{+}^{2}}(Y_{t_{l}}^{i_{1},i_{2}})
\]
for $i=0,\cdots,k-1$.
Again, $\nabla\theta(x,t)$ and $\nabla u(x,t)$ (which then update
the values of $A(x,t)$ and $F(x,t)$) may be calculated by differentiating
the iterations (\ref{BS-01}, \ref{BS-04}), which are defined by
\begin{align*}
A(x,&t_{k})  = \sum_{i_{1}}\sum_{i_{2}\geq0}h h_{i_1,i_2}\prod_{j=1}^k 1_{\mathbb{R}_{+}^{2}}(X_{t_{j}}^{i_{1},i_{2}})\nabla_x \varLambda_{2,\epsilon}(X_{t_{k}}^{i_{1},i_{2}},x)\wedge\hat{\omega}_{i_{1},i_{2}}\nonumber \\
 & +\sum_{i_{1}}\sum_{i_{2}\geq0}\delta h h_{i_1,i_2}\sum_{j=1}^{k}\prod_{l=j}^k 1_{\mathbb{R}_{+}^{2}}(X_{t_{l}}^{i_{1},i_{2}})\nabla_{x}\varLambda_{2,\epsilon}(X_{t_{k}}^{i_{1},i_{2}},x)\wedge F(X_{t_{j-1}}^{i_{1},i_{2}},t_{j-1})\nonumber \\
 & +\sum_{i_{1}}\delta h h_w\sum_{j=1}^{k}\prod_{l=j}^k 1_{\mathbb{R}_{+}^{2}}(X_{t_{l}}^{i_{1},0})\nabla_{x}\varLambda_{2,\epsilon}(X_{t_{k}}^{i_{1},0},x)\wedge\tilde{\chi}_{\varepsilon}(X_{t_{j-1}}^{i_{1},0},t_{j-1})
\end{align*}
and
\begin{align*}
\frac{\partial}{\partial x_{1}}&\theta(x,t_{k})  = \frac{T_w}{\pi}\left(\frac{x_2}{x_2^2+(L+x_1)^2}- \frac{x_2}{x_2^2+(L-x_1)^2}\right)\nonumber \\
 & -\sum_{i_{1}}\sum_{i_{2}\geq0}h h_{i_1,i_2}\prod_{j=1}^k 1_{\mathbb{R}_{+}^{2}}(Y_{t_{j}}^{i_{1},i_{2}})\frac{\partial}{\partial x_{1}}\varLambda_{2,\epsilon}(Y_{t_{k}}^{i_{1},i_{2}},x)\cdot R(x^{i_{1},i_{2}},t_{k};0)\varTheta_{i_{1},i_{2}}
\end{align*}
for $x=(x_{1},x_{2})$ with $x_{2}>0$, where the singular integral
kernel is given by 
\[
\nabla_{x}\varLambda_{2,\epsilon}(y,x)=\nabla_{x}\left(\left(1-e^{-\frac{|y-x|^{2}}{\epsilon}}\right)\varLambda_{2,\epsilon}(y,x)\right).
\]

Of course, we need (\ref{XB}) and (\ref{YB}), which take precisely the same forms, except for $u(x,t)=\mathscr{R}(u(\overline{x},t))$ if $x_{2}<0$.
This completes the scheme. 
Finally, we would like to note that this scheme is designed for laminar flows.

\begin{rem}
   It should be highlighted that the cost of computing the wall-bounded case is only marginally higher than the whole plane case. Though it seems that we need to store both diffusion paths $X$ and $Y$, as well as the path of $A$, indeed, we only need to keep track of the latter, similar to the whole plane case. In each iteration, only the spot values of $X$ and $Y$ are needed. 
\end{rem}

 \subsubsection{Multi-copy scheme}

Again, we may introduce another numerical scheme via the strong law of large numbers. By taking $2N$ independent Brownian particles that are modelled by the following discretised stochastic differential equations using the Euler scheme: for $1\leq m\leq N$,
\begin{align*}
&X_{0}^{m,i_{1},i_{2}}=x^{i_{1},i_{2}},\quad X_{t_{k}}^{m,i_{1},i_{2}}=X_{t_{k-1}}^{m,i_{1},i_{2}}+\delta u\left(X_{t_{k-1}}^{m,i_{1},i_{2}},t_{k-1}\right)+\sqrt{2\nu}(B_{t_{k}}^{1,m}-B_{t_{k-1}}^{1,m}),\\
&Y_{0}^{m,i_{1},i_{2}}=x^{i_{1},i_{2}},\quad Y_{t_{k}}^{m,i_{1},i_{2}}=Y_{t_{k-1}}^{m,i_{1},i_{2}}+\delta u\left(Y_{t_{k-1}}^{m,i_{1},i_{2}},t_{k-1}\right)+\sqrt{2\kappa}(B_{t_{k}}^{2,m}-B_{t_{k-1}}^{2,m}),
\end{align*}
where $u(x,t)=\mathscr{R}(u(\overline{x},t))$ if $x_{2}<0$, and the velocity $u$ and temperature $\theta$ are approximated by 
\begin{align*}
u(x,&t_{k})  = \frac{1}{N}\sum_{m=1}^N \sum_{i_{1}}\sum_{i_{2}\geq0}h h_{i_1,i_2}\prod_{j=1}^k 1_{\mathbb{R}_{+}^{2}}(X_{t_{j}}^{m,i_{1},i_{2}})\varLambda_{2,\epsilon}(X_{t_{k}}^{m,i_{1},i_{2}},x)\wedge\hat{\omega}_{i_{1},i_{2}}\nonumber \\
 & +\frac{1}{N}\sum_{m=1}^N\sum_{i_{1}}\sum_{i_{2}\geq0}\delta h h_{i_1,i_2}\sum_{j=1}^{k}\prod_{l=j}^k 1_{\mathbb{R}_{+}^{2}}(X_{t_{l}}^{m,i_{1},i_{2}})\varLambda_{2,\epsilon}(X_{t_{k}}^{m,i_{1},i_{2}},x)\wedge F(X_{t_{j-1}}^{m,i_{1},i_{2}},t_{j-1})\nonumber \\
 & +\frac{1}{N}\sum_{m=1}^N\sum_{i_{1}}\delta h h_w\sum_{j=1}^{k}\prod_{l=j}^k 1_{\mathbb{R}_{+}^{2}}(X_{t_{l}}^{m,i_{1},0})\varLambda_{2,\epsilon}(X_{t_{k}}^{m,i_{1},0},x)\wedge\tilde{\chi}_{\varepsilon}(X_{t_{j-1}}^{m,i_{1},0},t_{j-1})\label{BS-01}
\end{align*}
and 
\begin{align*}
   \theta(x,&t_{k})  = \frac{T_w}{\pi} \left( \tan^{-1} \left(\frac{L-x_1}{x_2}\right)+\tan^{-1} \left(\frac{L+x_1}{x_2}\right)\right)\nonumber \\
 &-\frac{1}{N}\sum_{m=1}^N \sum_{i_{1}}\sum_{i_{2}\geq0}h h_{i_1,i_2} \prod_{j=1}^k 1_{\mathbb{R}_{+}^{2}}(Y_{t_{j}}^{m,i_{1},i_{2}})\varLambda_{2}(Y_{t_{k}}^{m,i_{1},i_{2}},x)\cdot R(x^{i_{1},i_{2}},t_{k};0)\varTheta_{i_{1},i_{2}}.
\end{align*}
Similarly, for each $m$, the gauge functional equation is given by $R(x^{i_{1},i_{2}},t_{k};t_{k})=I$, and
\begin{align*}
R(x^{i_{1},i_{2}},t_{k};t_{i})=I-\sum_{l=i+1}^{k}\delta A\left(Y_{t_{l}}^{m,i_{1},i_{2}},t_{l}\right)R(x^{i_{1},i_{2}},t_{k};t_{l})1_{\mathbb{R}_{+}^{2}}(Y_{t_{l}}^{m,i_{1},i_{2}})
\end{align*}
for $i=0,\cdots,k-1$. As for the derivatives, similar to the whole plane case, we have
\begin{align*}
A(x,&t_{k})  = \frac{1}{N}\sum_{m=1}^N \sum_{i_{1}}\sum_{i_{2}\geq0}h h_{i_1,i_2}\prod_{j=1}^k 1_{\mathbb{R}_{+}^{2}}(X_{t_{j}}^{m,i_{1},i_{2}})\varLambda_{2,\epsilon}(X_{t_{k}}^{m,i_{1},i_{2}},x)\wedge\hat{\omega}_{i_{1},i_{2}}\nonumber \\
 & +\frac{1}{N}\sum_{m=1}^N \sum_{i_{1}}\sum_{i_{2}\geq0}\delta h h_{i_1,i_2}\sum_{j=1}^{k}\prod_{l=j}^k 1_{\mathbb{R}_{+}^{2}}(X_{t_{l}}^{m,i_{1},i_{2}})\nabla_{x}\varLambda_{2,\epsilon}(X_{t_{k}}^{m,i_{1},i_{2}},x)\wedge F(X_{t_{j-1}}^{m,i_{1},i_{2}},t_{j-1}) \\
 & +\frac{1}{N}\sum_{m=1}^N \sum_{i_{1}}\delta h h_w\sum_{j=1}^{k}\prod_{l=j}^k 1_{\mathbb{R}_{+}^{2}}(X_{t_{l}}^{m,i_{1},0})\nabla_{x}\varLambda_{2,\epsilon}(X_{t_{k}}^{m,i_{1},0},x)\wedge\tilde{\chi}_{\varepsilon}(X_{t_{j-1}}^{m,i_{1},0},t_{j-1})
\end{align*}
and
\begin{align*}
\frac{\partial}{\partial x_{1}}&\theta(x,t_{k})  = \frac{T_w}{\pi}\left(\frac{x_2}{x_2^2+(L+x_1)^2}- \frac{x_2}{x_2^2+(L-x_1)^2}\right) \\
 & -\frac{1}{N}\sum_{m=1}^N \sum_{i_{1}}\sum_{i_{2}\geq0}h h_{i_1,i_2}\prod_{j=1}^k 1_{\mathbb{R}_{+}^{2}}(Y_{t_{j}}^{m,i_{1},i_{2}})\frac{\partial}{\partial x_{1}}\varLambda_{2,\epsilon}(Y_{t_{k}}^{m,i_{1},i_{2}},x)\cdot R(x^{i_{1},i_{2}},t_{k};0)\varTheta_{i_{1},i_{2}}.
\end{align*}
The rest of the scheme remains the same as in the one-copy case. This completes the SLN scheme for the wall-bounded Oberbeck-Boussinesq flows.

\subsubsection{Numerical experiments}
Here we carried out two sets of numerical experiments based on the one-copy scheme for the half-plane case. We consider two different Prandtl numbers: when the kinematic viscosity $\nu = 1$, and thermal diffusivity $\kappa = 0.15$, and when $\nu = 0.15$, $\kappa = 1$.

We choose the typical length scale $L= 2\pi$, consider the flow in the rectangle region $[-L,L]\times[0,L]$ and assume that $\alpha = 0.0005$. The parameter we use to smooth out the Biot-Savart kernel is chosen to be $\epsilon = 0.1$. 

In the experiment presented, we set the initial velocity to be of the form
\[
u(x,0) = (-10\sin(2x_2),0),
\]
and the initial temperature is given by 
\[
\theta(x,0) = 0.01(8\pi^2-{x_1}^2-{x_2}^2).
\]

Thus, $\omega(x,0) = 20\cos(2x_2)$, and 
\[
\hat{\omega}_{i_1,i_2} =\begin{cases}
    20\cos(2i_2 h), \quad &i_2>0,\\
    0,    \quad &i_2=0.
\end{cases} 
\]
The time step is $\delta = 0.01$, and the boundary layer thickness is taken to be $\gamma = \pi/40$. The mesh size is set to be $h= 2\pi/40 \approx 0.157$ and $h_w = \gamma/20 \approx 0.00393$, and $\varepsilon = h_w/2$. We assume that the external heat source is at temperature $T_w = 0.1$, and 
\[
\theta_0'(x_1) = -0.2x_1. 
\]

The numerical experiment results at times $t = 0.2$, $t= 1.0$, $t=1.8$ with Prandtl number $\Pr=0.15$ are shown in Figure \ref{OB-W-sim-015}, and the results at the given times with $\Pr=6.67$ are shown in Figure \ref{OB-W-sim-667}.

\begin{figure}[H]
     \centering
     \begin{subfigure}[b]{0.3\textwidth}
     \includegraphics[width=\textwidth]{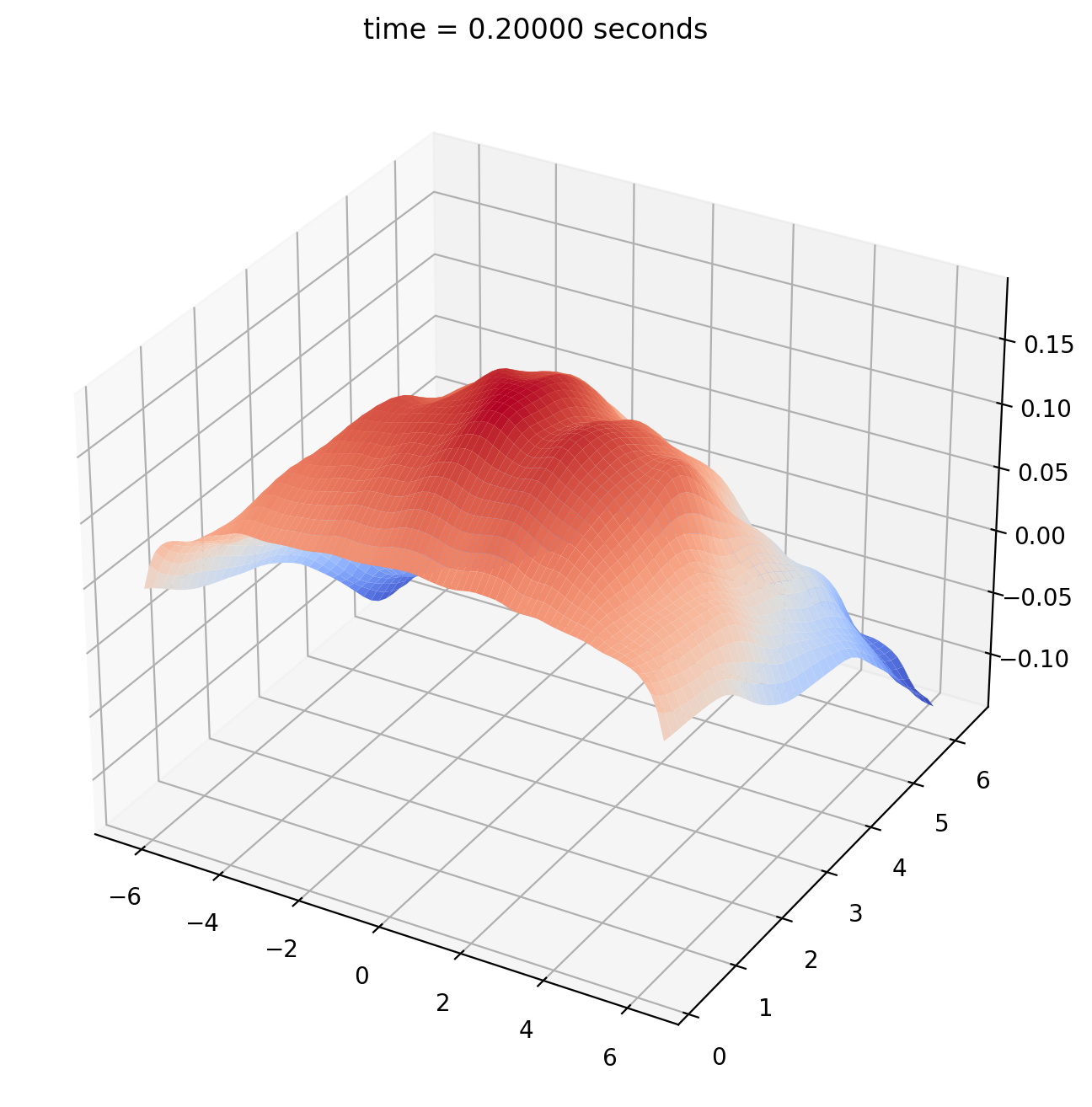}
     \caption{Temperature at $t=0.2$.}
     \end{subfigure}
     \hspace{0.5cm}
     \begin{subfigure}[b]{0.3\textwidth}
     \includegraphics[width=\textwidth]{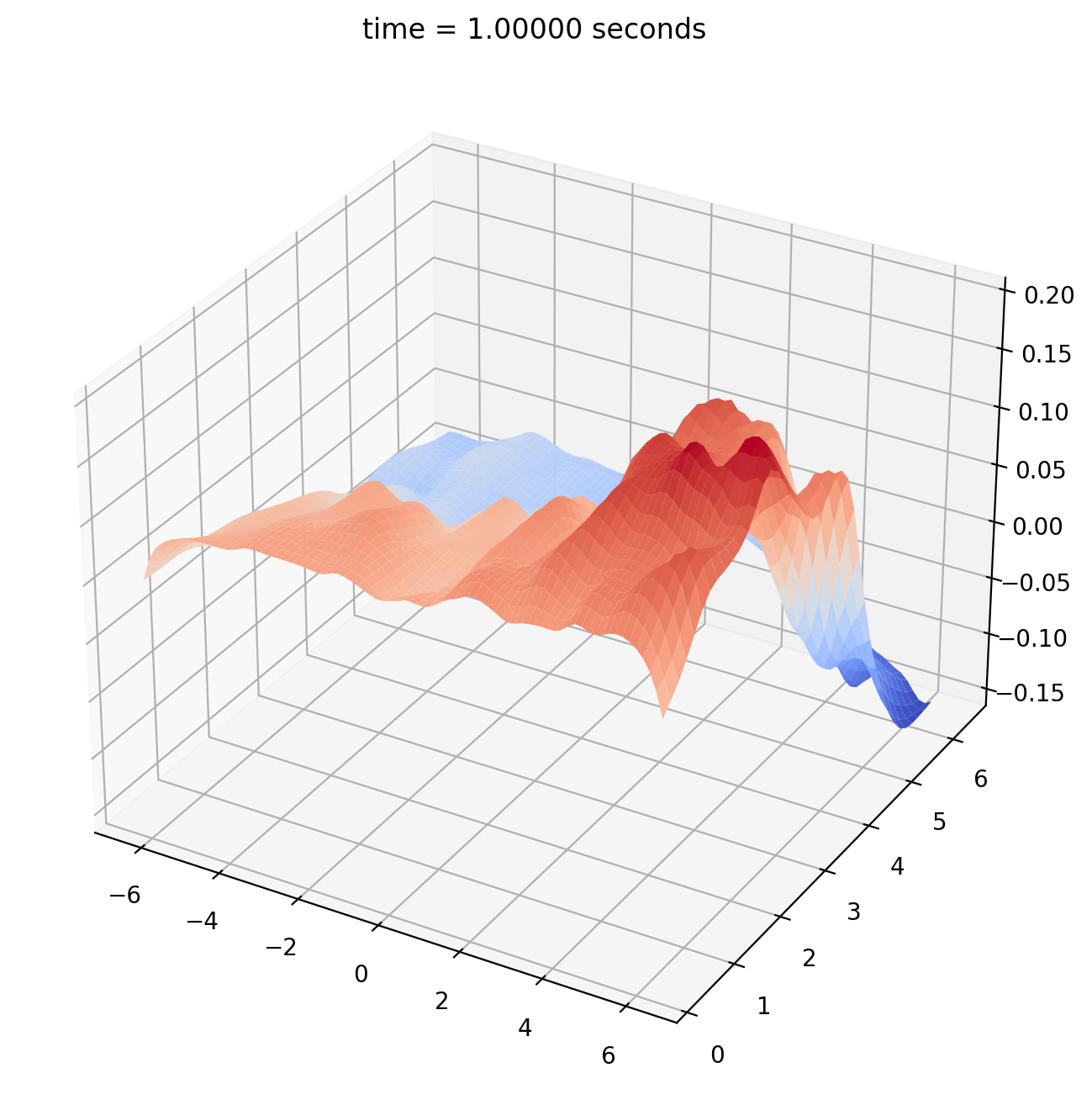}
     \caption{Temperature field at $t=1.0$.}
      \end{subfigure}
      \hspace{0.5cm}
     \begin{subfigure}[b]{0.3\textwidth}
     \includegraphics[width=\textwidth]{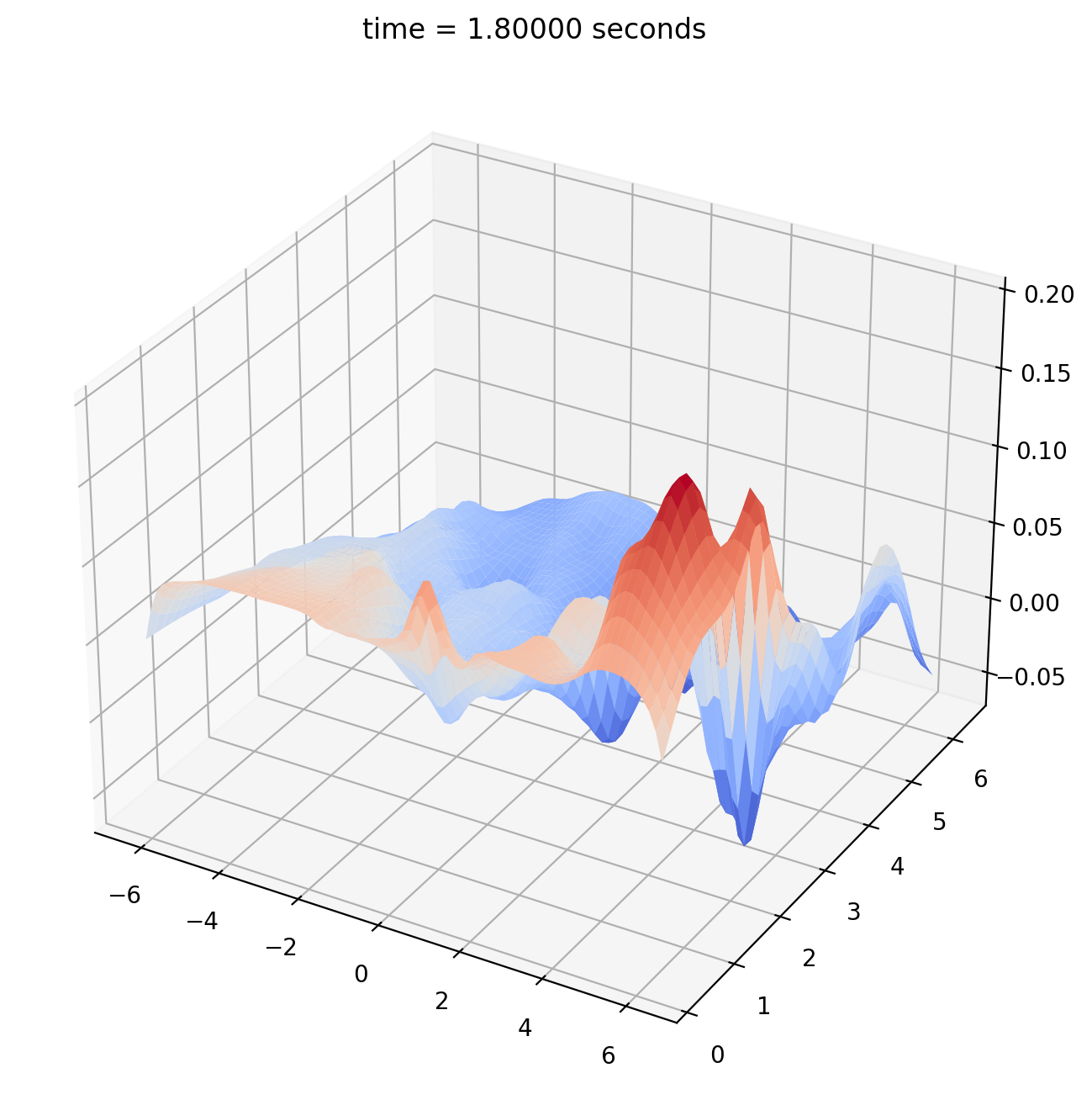}
     \caption{Temperature field at $t=1.8$.}
     \end{subfigure}
     
     \vspace{0.5cm}
      \begin{subfigure}[b]{0.3\textwidth}
     \includegraphics[width=\textwidth]{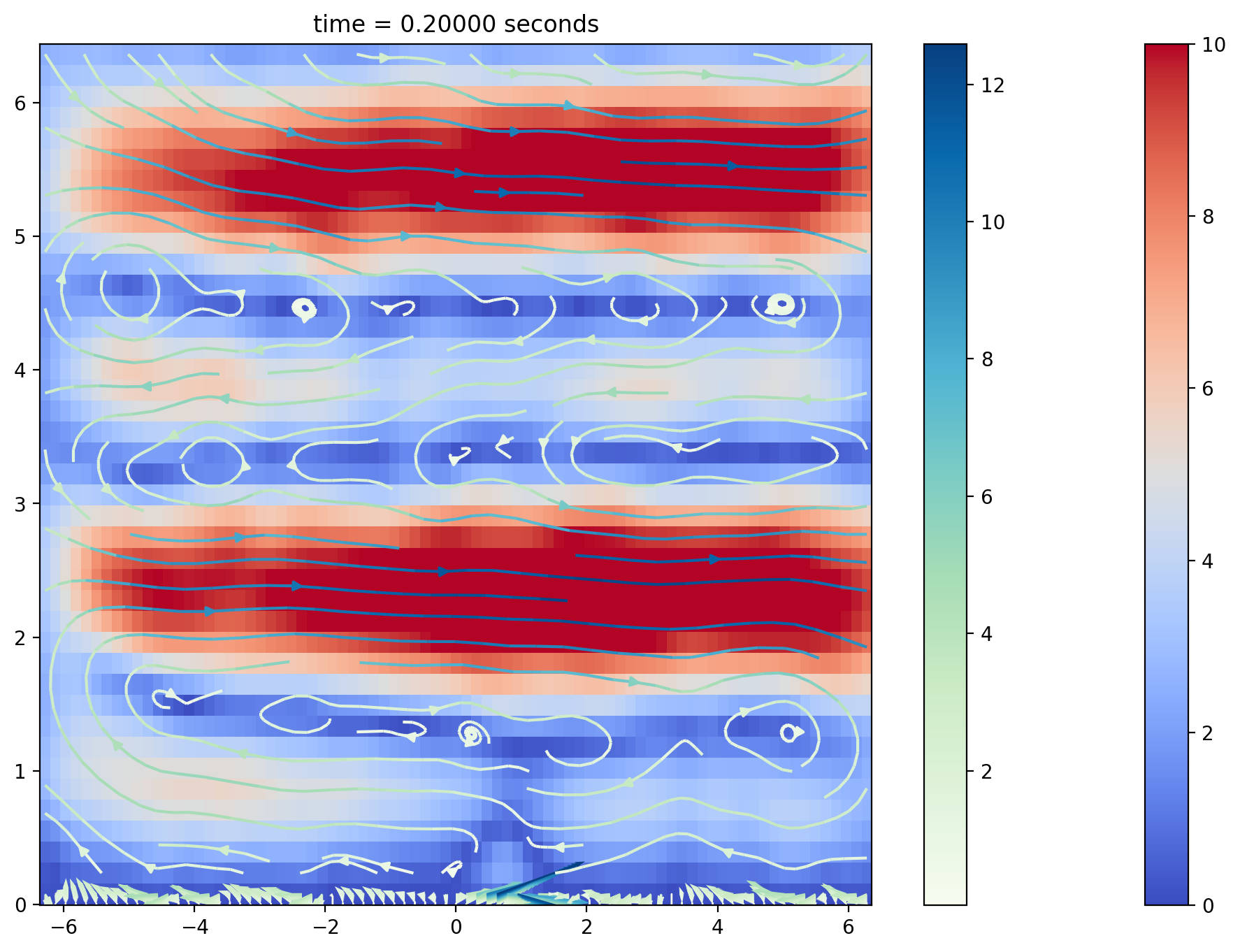}
     \caption{Velocity at $t=0.2$.}
     \end{subfigure}
     \hspace{0.5cm}
     \begin{subfigure}[b]{0.3\textwidth}
     \includegraphics[width=\textwidth]{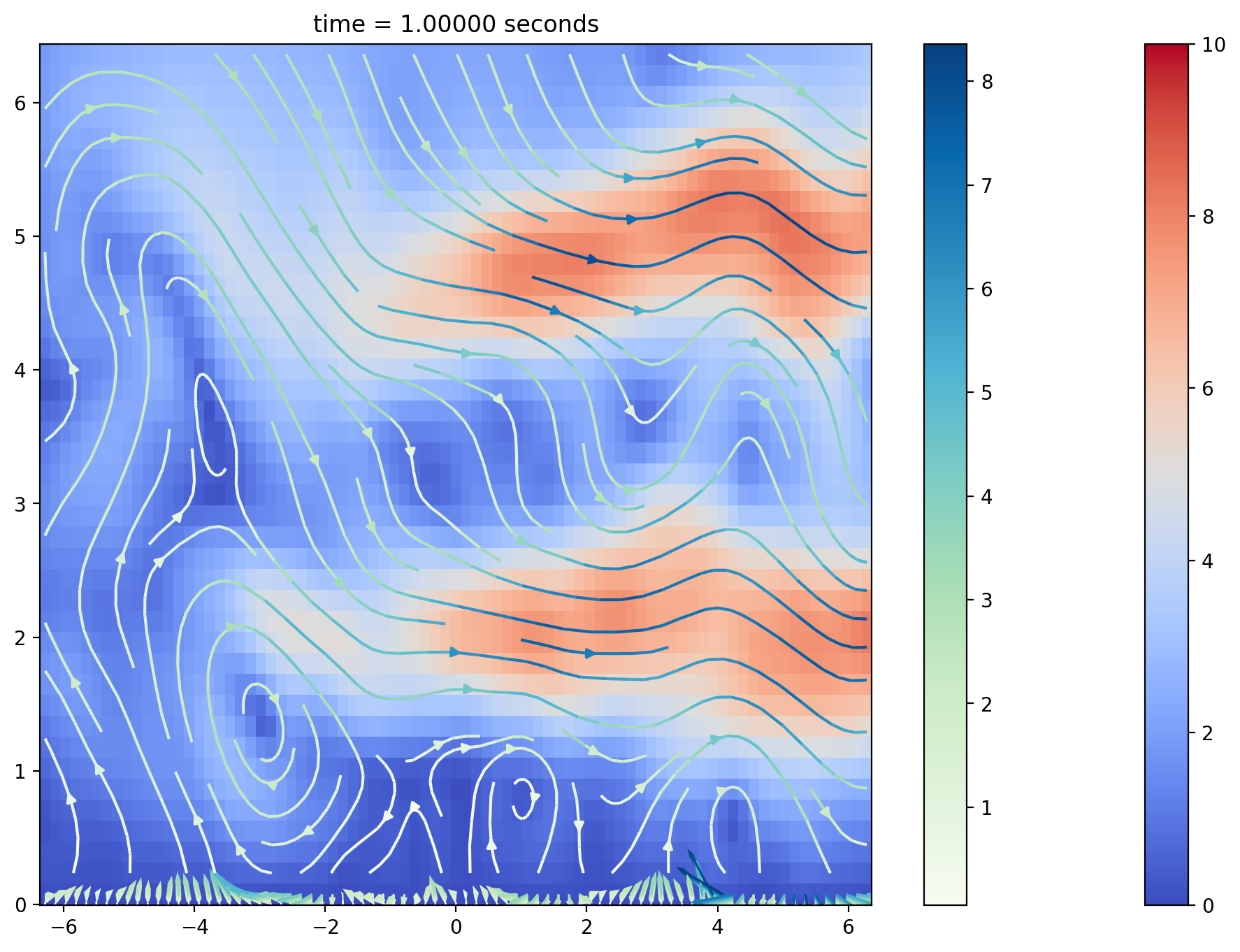}
     \caption{Velocity at $t=1.0$.}
     \end{subfigure}
      \hspace{0.5cm}
        \begin{subfigure}[b]{0.3\textwidth}
     \includegraphics[width=\textwidth]{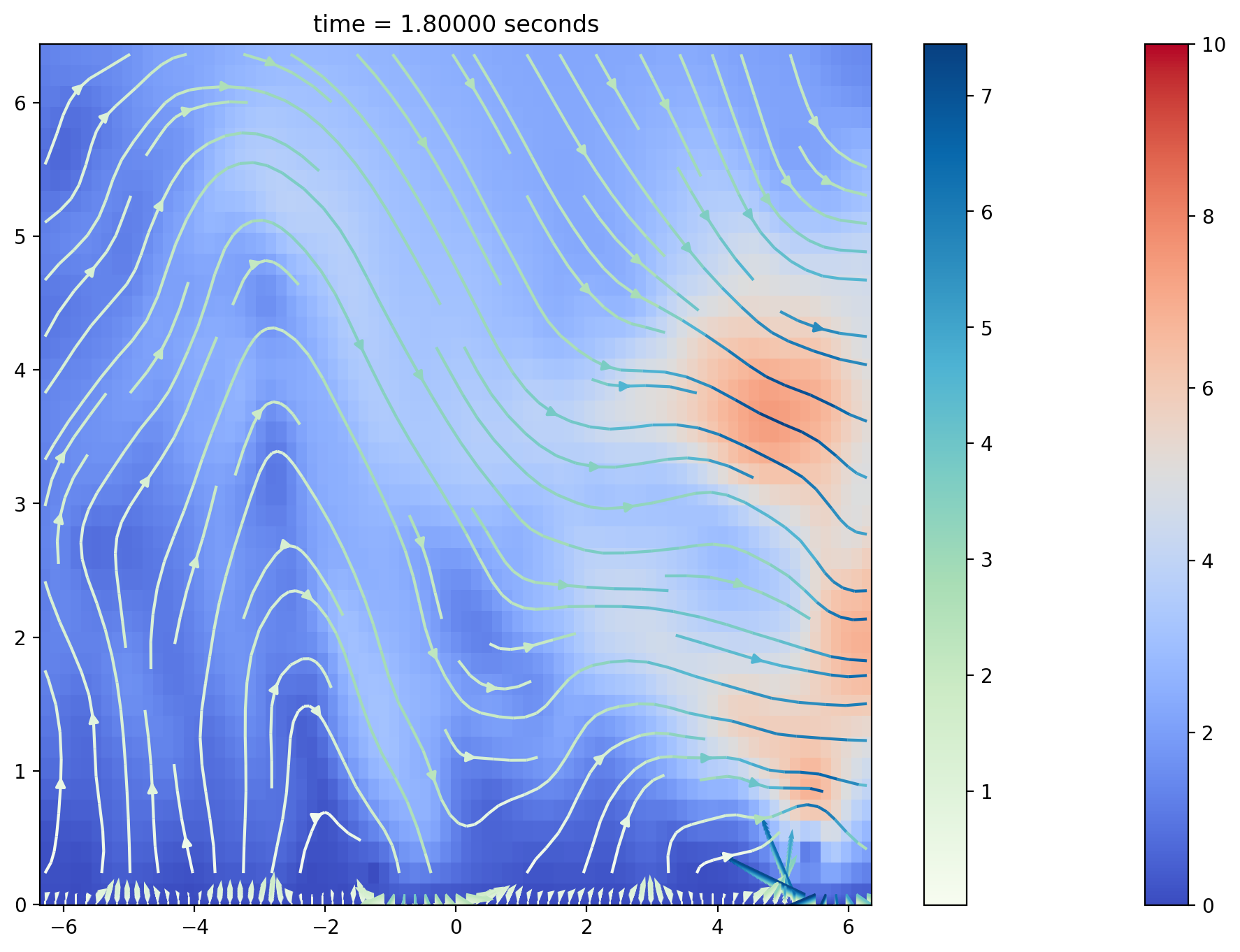}
     \caption{Velocity at $t=1.8$.}
     \end{subfigure}

     \vspace{0.5cm}
      \begin{subfigure}[b]{0.3\textwidth}
     \includegraphics[width=\textwidth]{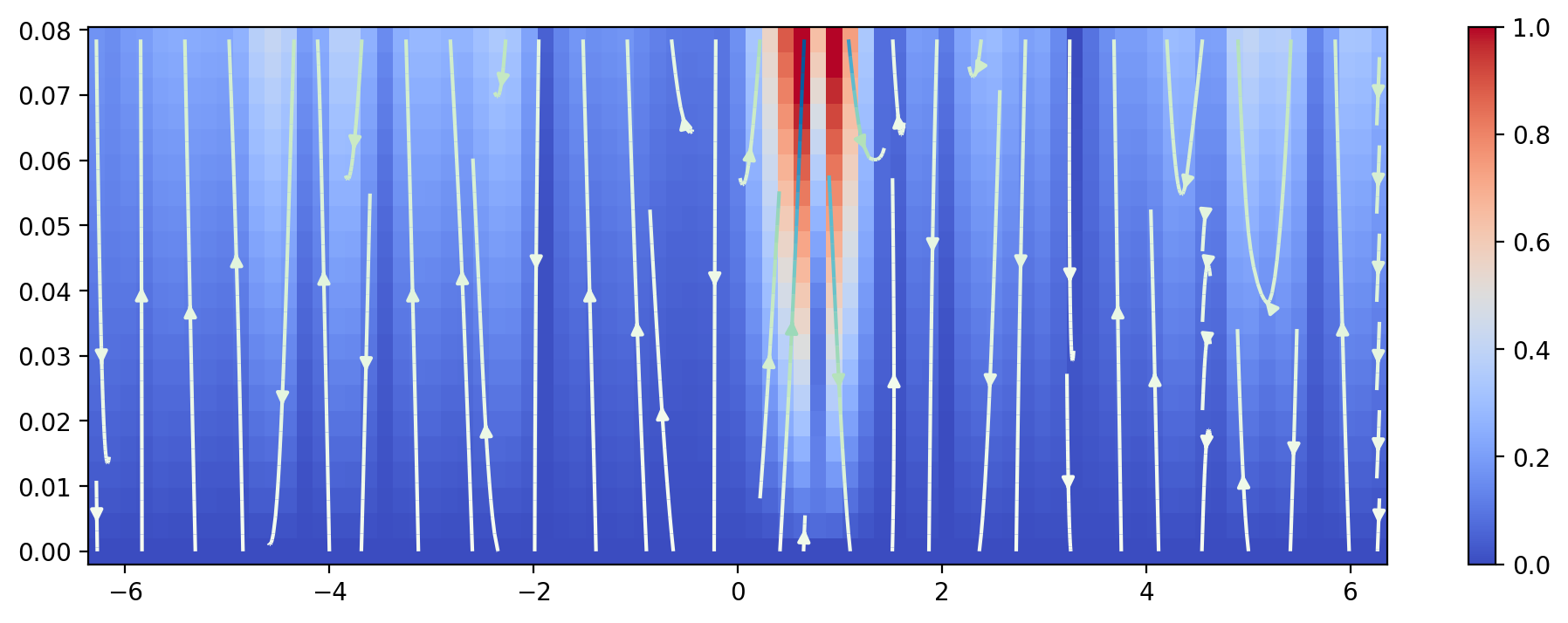}
     \caption{Boundary layer velocity at $t=0.2$.}
     \end{subfigure}
     \hspace{0.5cm}
     \begin{subfigure}[b]{0.3\textwidth}
     \includegraphics[width=\textwidth]{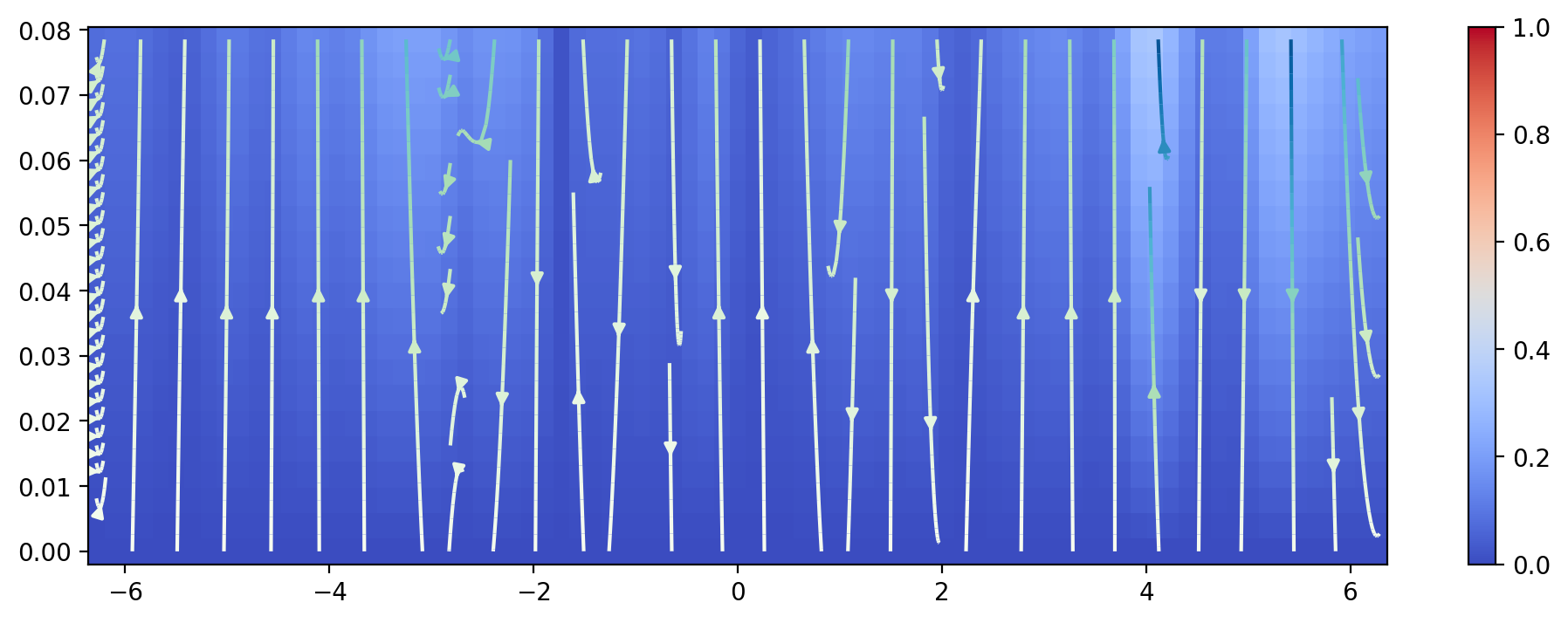}
     \caption{Boundary layer velocity at $t=1.0$.}
     \end{subfigure}
      \hspace{0.5cm}
      \begin{subfigure}[b]{0.3\textwidth}
     \includegraphics[width=\textwidth]{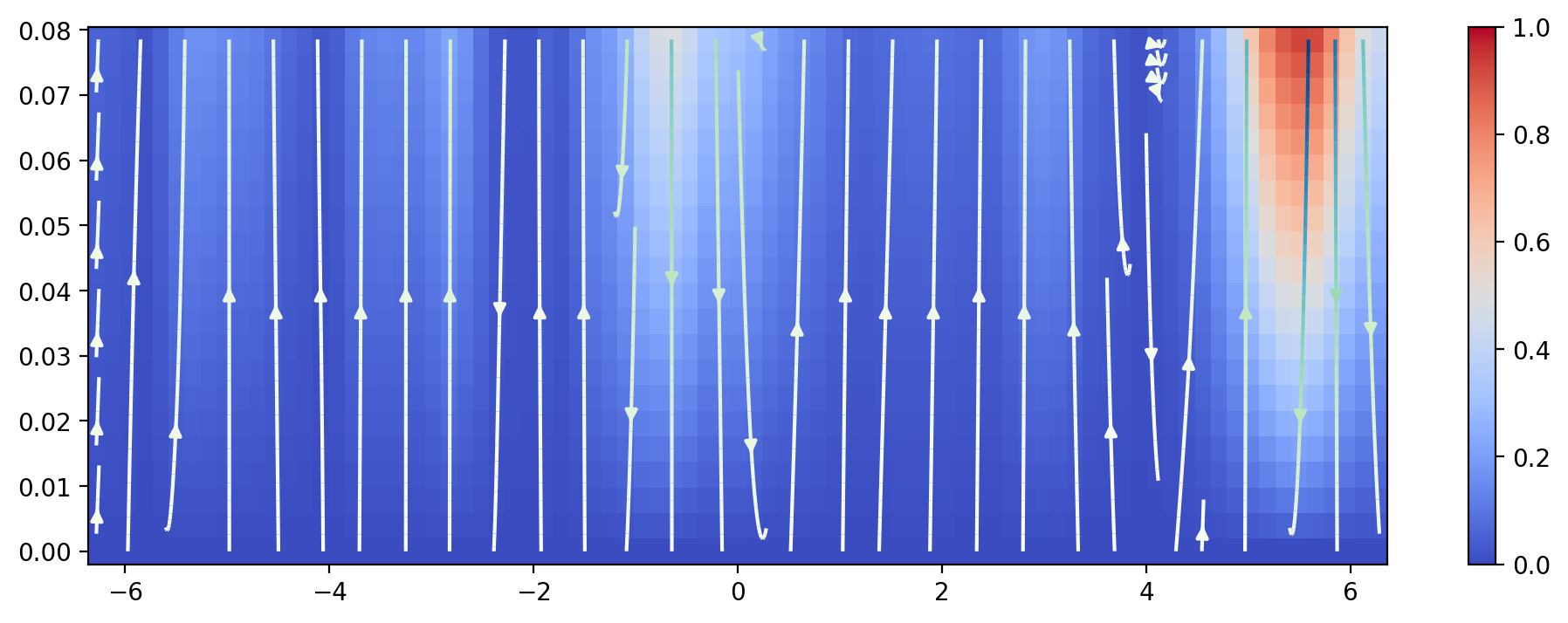}
     \caption{Boundary layer velocity at $t=1.8$.}
     \end{subfigure}
 \caption{Temperature and velocity fields of Oberbeck-Boussinesq flows on $\mathbb{R}^2_+$ with Prandtl number $\mathrm{Pr} = 0.15$.}
 \label{OB-W-sim-015}
 \end{figure}
 
 \begin{figure}[H]
     \centering
     \begin{subfigure}[b]{0.3\textwidth}
     \includegraphics[width=\textwidth]{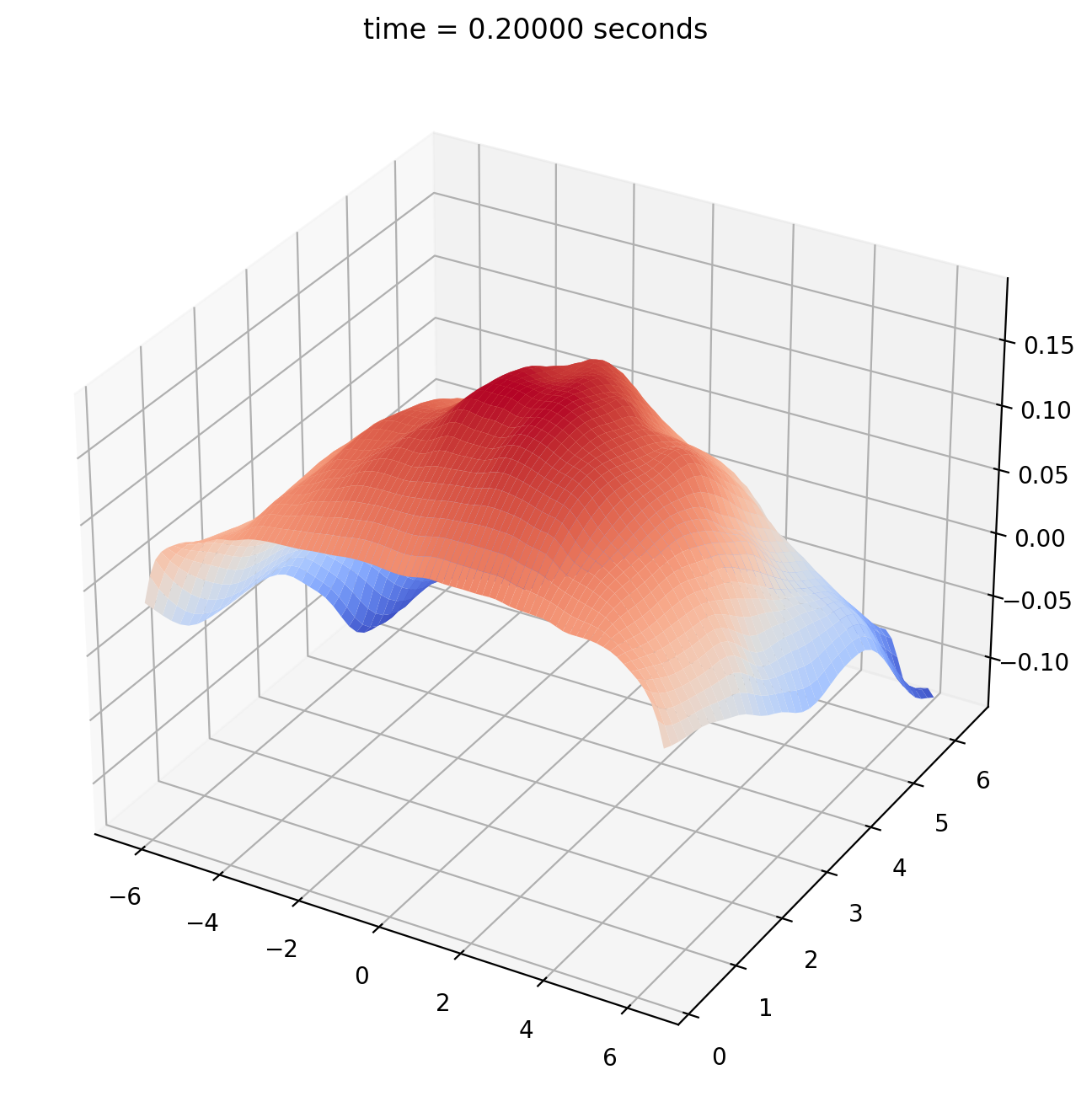}
     \caption{Temperature at $t=0.2$.}
     \end{subfigure}
     \hspace{0.5cm}
     \begin{subfigure}[b]{0.3\textwidth}
     \includegraphics[width=\textwidth]{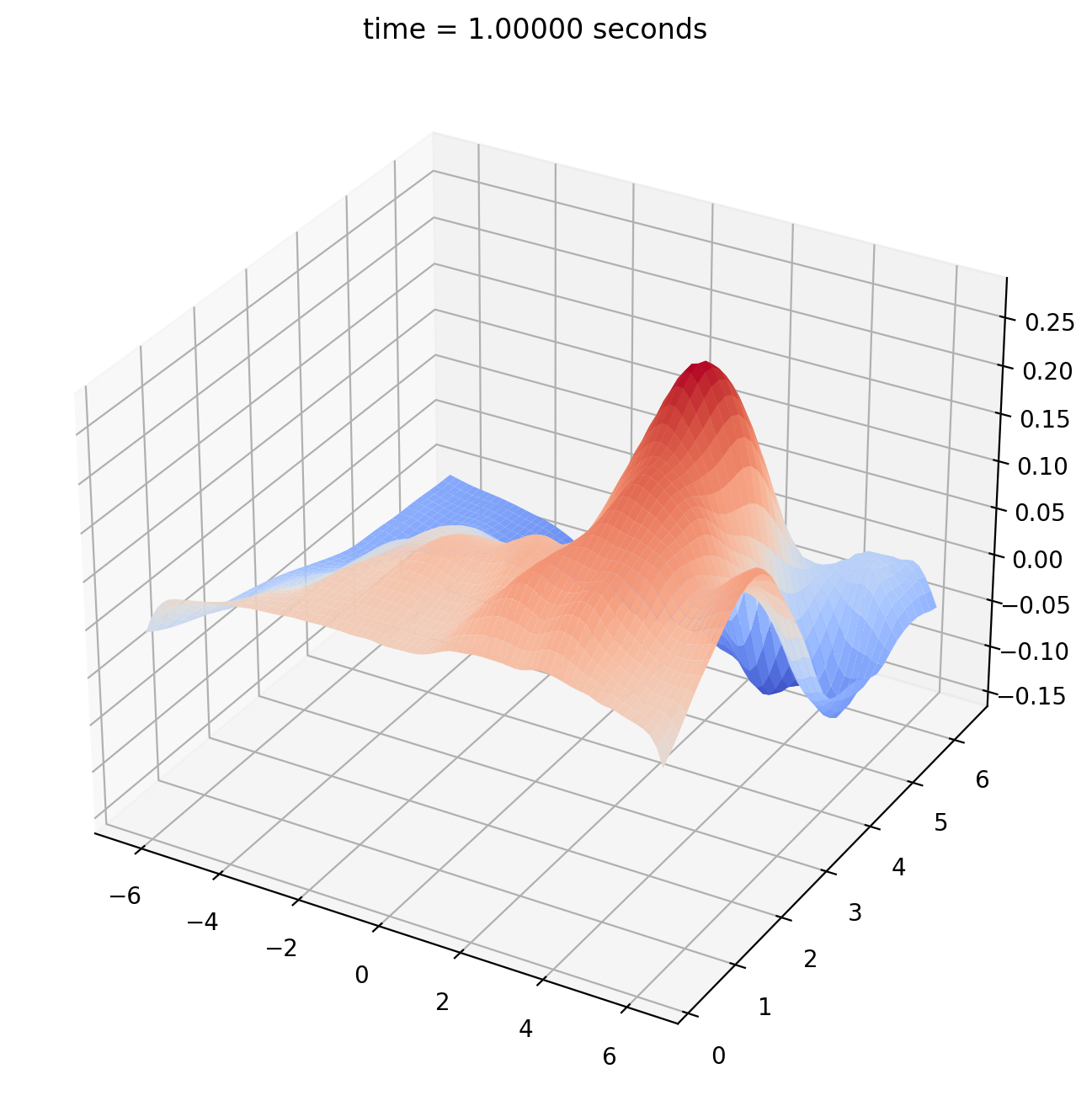}
     \caption{Temperature field at $t=1.0$.}
      \end{subfigure}
      \hspace{0.5cm}
     \begin{subfigure}[b]{0.3\textwidth}
     \includegraphics[width=\textwidth]{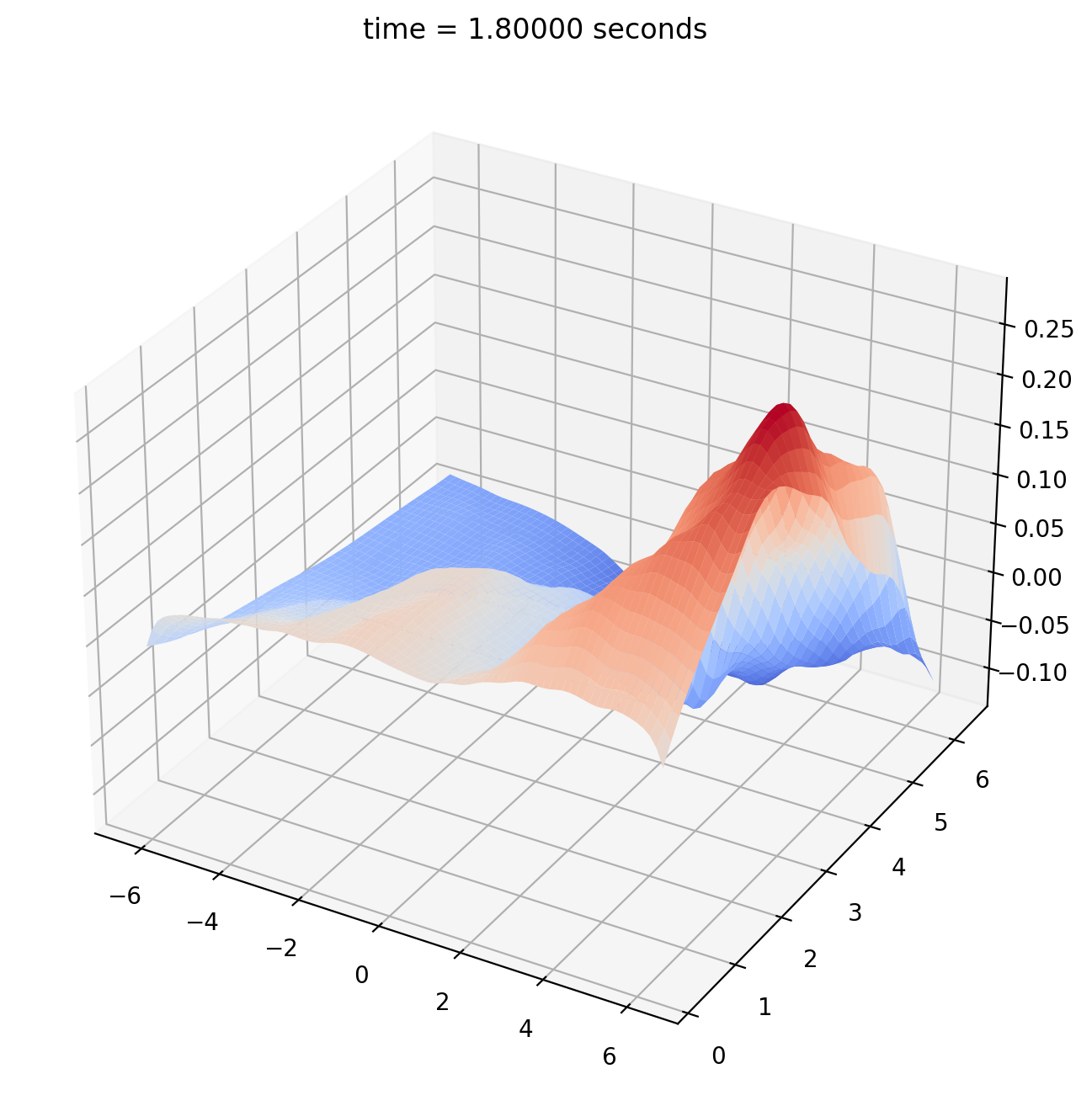}
     \caption{Temperature field at $t=1.8$.}
     \end{subfigure}
     
     \vspace{0.5cm}
      \begin{subfigure}[b]{0.3\textwidth}
     \includegraphics[width=\textwidth]{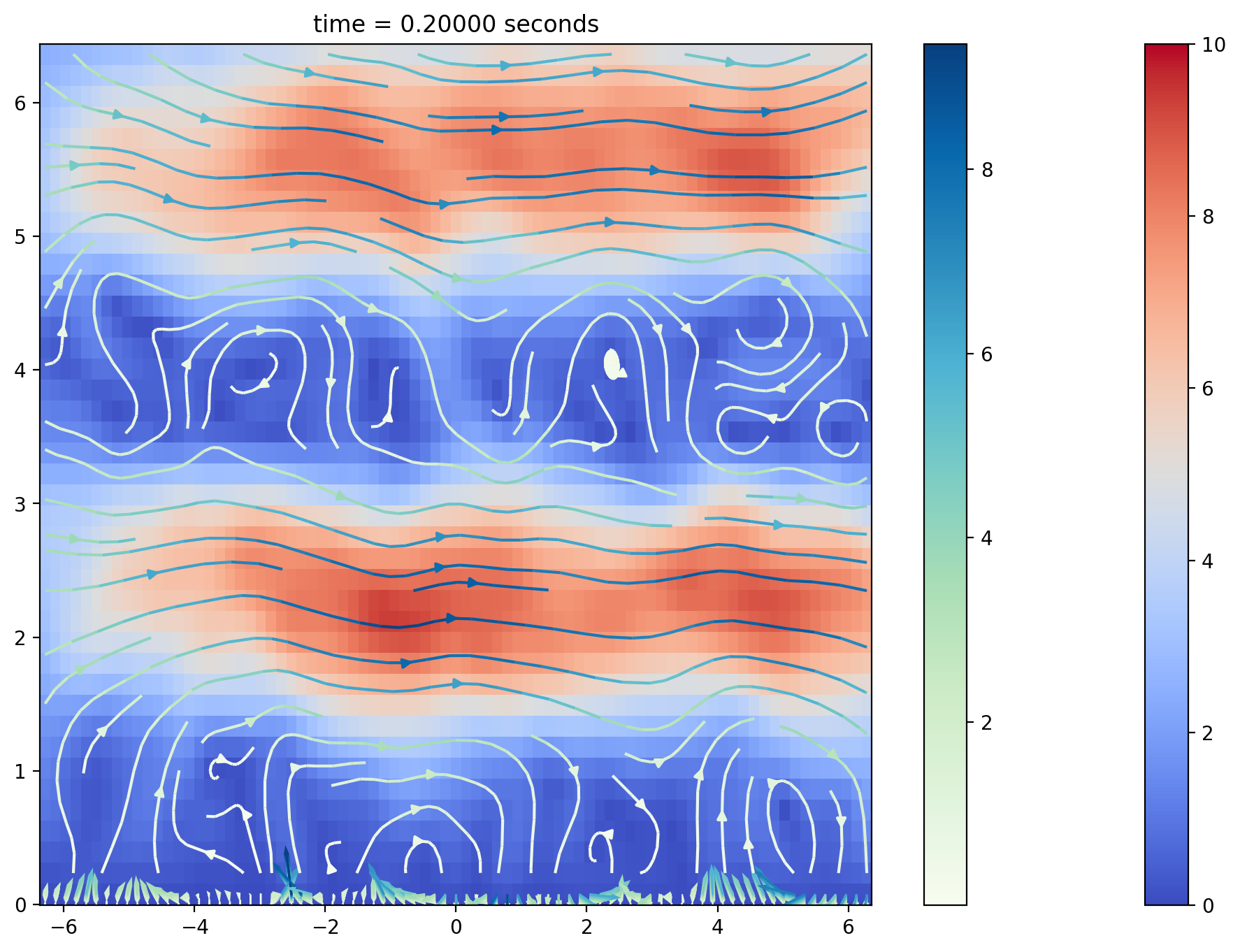}
     \caption{Velocity at $t=0.2$.}
     \end{subfigure}
     \hspace{0.5cm}
     \begin{subfigure}[b]{0.3\textwidth}
     \includegraphics[width=\textwidth]{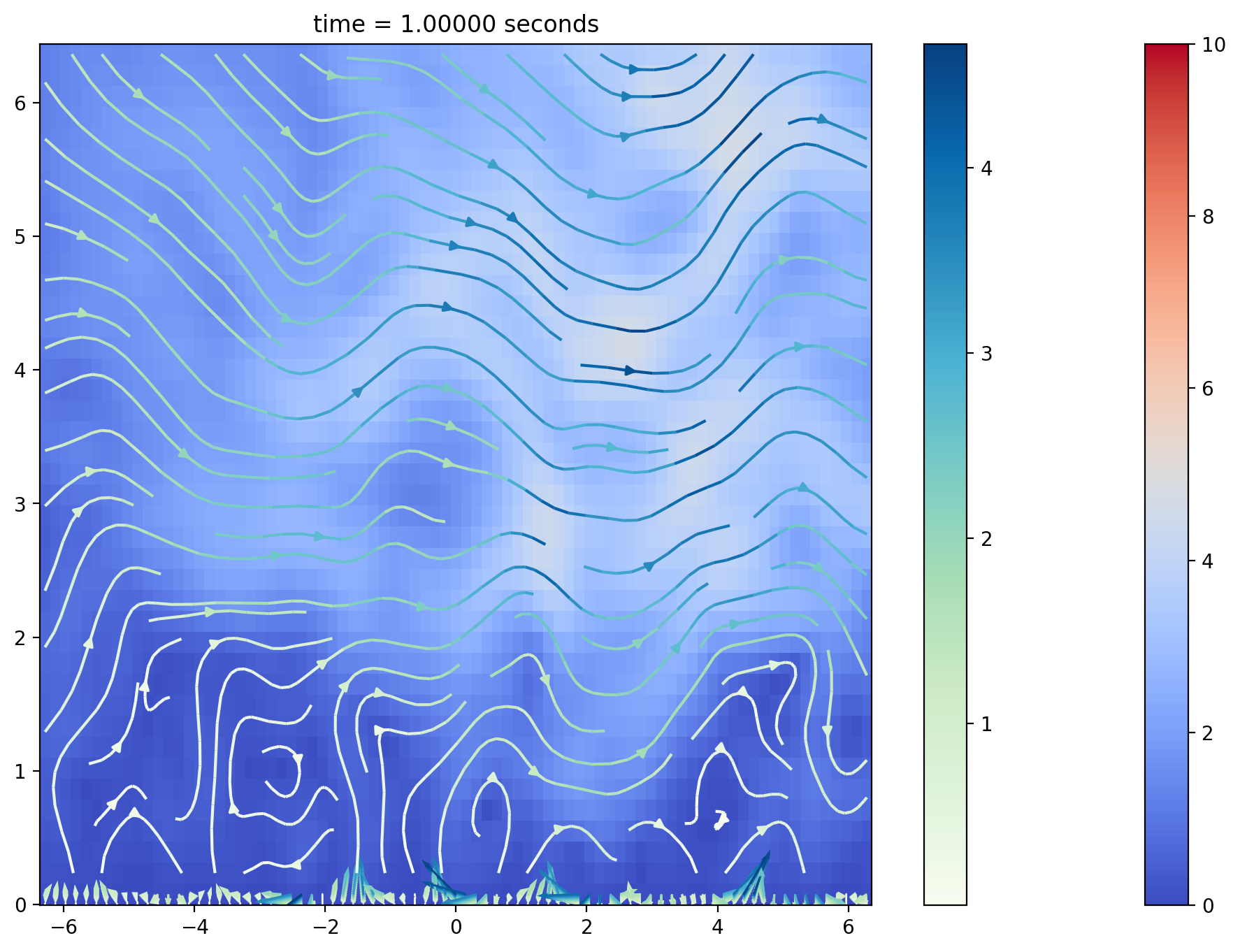}
     \caption{Velocity at $t=1.0$.}
     \end{subfigure}
      \hspace{0.5cm}
        \begin{subfigure}[b]{0.3\textwidth}
     \includegraphics[width=\textwidth]{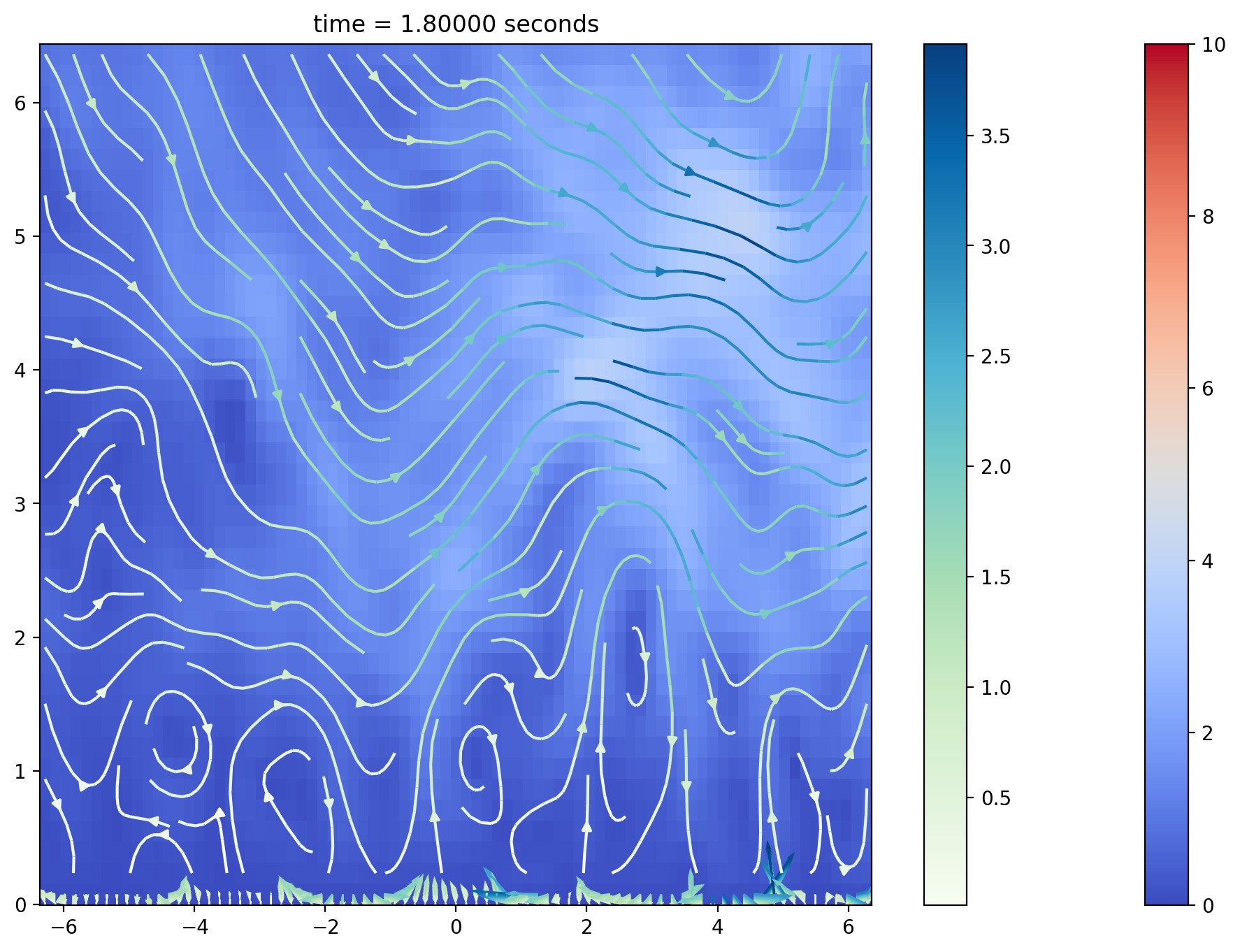}
     \caption{Velocity at $t=1.8$.}
     \end{subfigure}

     \vspace{0.5cm}
      \begin{subfigure}[b]{0.3\textwidth}
     \includegraphics[width=\textwidth]{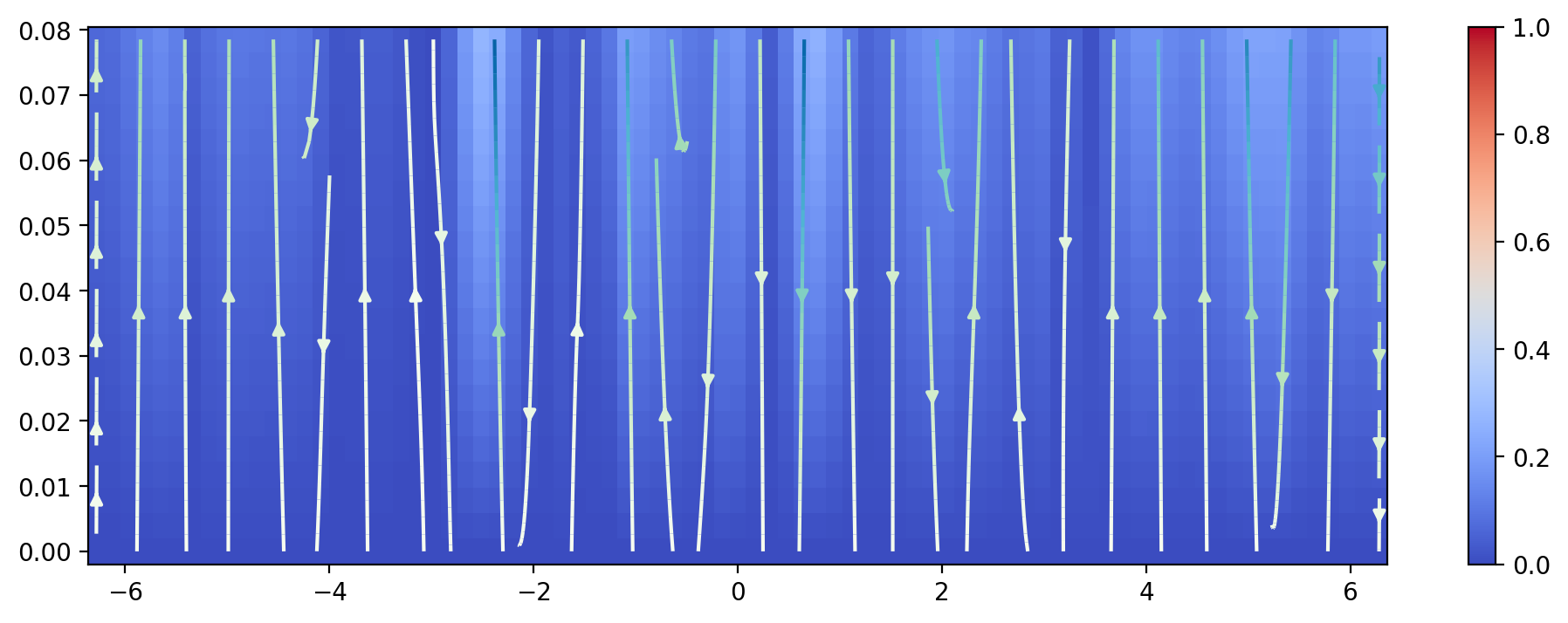}
     \caption{Boundary layer velocity at $t=0.2$.}
     \end{subfigure}
     \hspace{0.5cm}
     \begin{subfigure}[b]{0.3\textwidth}
     \includegraphics[width=\textwidth]{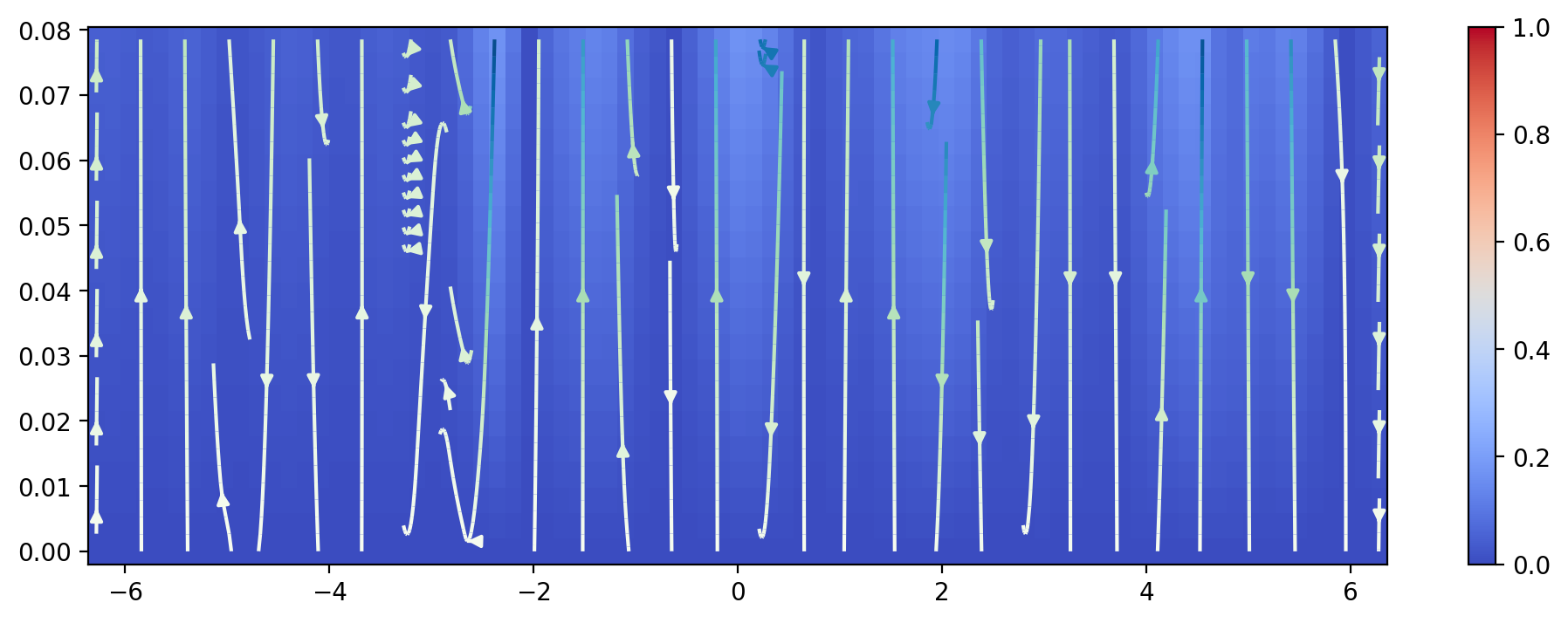}
     \caption{Boundary layer velocity at $t=1.0$.}
     \end{subfigure}
      \hspace{0.5cm}
      \begin{subfigure}[b]{0.3\textwidth}
     \includegraphics[width=\textwidth]{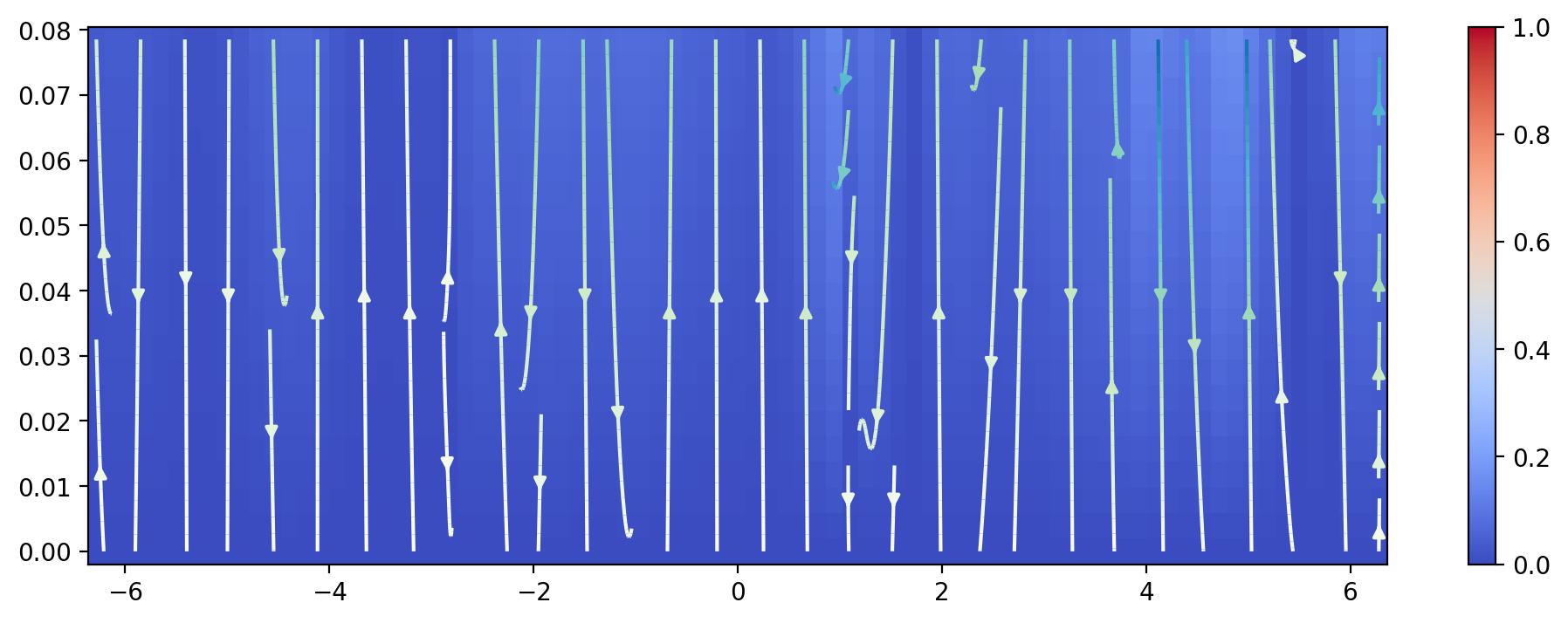}
     \caption{Boundary layer velocity at $t=1.8$.}
     \end{subfigure}
 \caption{Temperature and velocity fields of Oberbeck-Boussinesq flows on $\mathbb{R}^2_+$ with Prandtl number $\mathrm{Pr} = 6.67$.}
 \label{OB-W-sim-667}
 \end{figure}
The simulations demonstrate well the regular pattern known as the B{\'e}nard convection, and also reveal the detailed hairy type of flows within the thin boundary layer, confirming the theoretical results and observations.

\section*{Data Availability Statement}

No data are used in this article to support the findings of this study.

\section*{Acknowledgement}

ZQ is supported
partially by the EPSRC Centre for Doctoral Training in Mathematics
of Random Systems: Analysis, Modelling and Simulation (EP/S023925/1).

\end{document}